\documentclass[a4paper,UKenglish,cleveref, autoref, thm-restate]{lmcs}
\pdfoutput=1
\usepackage[utf8]{inputenc}

\usepackage{lastpage}
\lmcsdoi{19}{1}{19}
\lmcsheading{}{\pageref{LastPage}}{}{}%
{Nov.~10,~2021}{Mar.~14,~2023}{}

\usepackage[ruled, lined, linesnumbered, longend]{algorithm2e}
\usepackage{amsfonts}
\usepackage{amsmath}
\usepackage{amssymb}
\usepackage{amsthm}
\usepackage{mathtools}
\usepackage{caption}

\begin{document}

\newcommand{\ignore}[1]{}
\newcommand{\grt}{>}
\newcommand{\rep}{\leftarrow}
\newcommand{\rew}{\rightarrow}
\newcommand{\eqv}{\leftrightarrow}
\newcommand{\gr}{Gr\"obner\,}
\newcommand{\Groebner}{Gr\"obner\,}
\newcommand{\Grobner}{Gr\"obner\,}
\newcommand{\groebner}{Gr\"obner\,}
\newcommand{\flatten}[1]{\overline #1}
\newcommand{\mset}[1]{\{\!\!\{ #1 \}\!\!\}}

\SetKwInput{kwRequires}{Requires}
\SetKwInput{kwEnsures}{Ensures}
\SetKwInput{kwInput}{Input}
\SetKwInput{kwOutput}{Output}
\SetKw{kwGoto}{Go to}

\DeclarePairedDelimiter{\set}{\{}{\}}


\title[Modularity, Combination, AC Congruence Closure]{
Modularity and Combination of Associative Commutative Congruence Closure Algorithms enriched with Semantic Properties}


\author{Deepak Kapur}
\begin{center}
\address{Department of Computer Science\\
University  of New Mexico, Albuquerque, NM, USA}

\email{kapur@unm.edu}

\thanks{Research partially supported by the NSF award: CCF-1908804.}
\end{center}








\sloppy


\ignore{Issues: In combination using lexicographic ordering, how many new constants need to be added.}

\maketitle

\begin{abstract}
Algorithms for computing congruence
closure of ground equations over uninterpreted symbols and
interpreted symbols satisfying associativity and commutativity (AC) properties 
are proposed. The algorithms are based on a framework for computing a
congruence closure as a canonical ground rewrite system in which nonflat
terms are abstracted by new 
constants as proposed first in Kapur's congruence closure 
algorithm (RTA97). The framework
is general, flexible, and has been extended also to develop congruence
closure algorithms for the cases when associative-commutative
function symbols can have additional properties including idempotency, nilpotency,
identities, cancelativity and group properties as well as their various combinations. Algorithms are modular; their correctness
and termination proofs are simple, exploiting modularity. Unlike earlier algorithms, the proposed algorithms neither rely on complex AC compatible well-founded orderings on nonvariable terms nor need to use the associative-commutative unification and extension rules in completion for generating canonical rewrite systems for congruence closures. They are particularly suited for integrating into the Satisfiability modulo Theories (SMT) solvers. It is shown that a \Groebner basis algorithm for polynomial ideals with integer coefficients can be formulated as a combination  algorithm for the congruence closure over an Abelian group with $+$ and the congruence closure over the AC symbol $*$ with the identity 1, linked using the distributivity property of $*$ over $+$.
\end{abstract}

\vspace*{-2mm}
\section{Introduction}
\label{sec:introduction}
\vspace*{-1mm}

Equality reasoning arises in many applications
including compiler optimization, functional languages, and
reasoning about data bases, most importantly, reasoning about
different aspects of software and hardware. The significance of the congruence closure algorithms
on ground equations in compiler optimization and verification
applications was recognized in the mid
70's and early 80's, leading to a number of algorithms for computing the
congruence closure of ground equations on uninterpreted function symbols
\cite{DST80,Shostak78,NO80}. 
Whereas congruence closure algorithms were
implemented in earlier verification systems \cite{Shostak78,NO80,RRL,context}, their role has become particularly critical in Satisfiability modulo Theories (SMT) solvers as 
a glue to combine different decision procedures for various theories. More recently, many new uses of the congruence closure are being explored in the {\bf egg} project at the University of Washington \cite{egg2021}.

We present algorithms for the congruence closure of ground equations which
in addition to uninterpreted function symbols, have symbols with
the associative (A) and commutative (C) properties. 
Using these algorithms, it can be decided whether another ground equation follows
from a finite set of ground equations with associative-commutative (AC) symbols as well as uninterpreted symbols. 
Canonical forms (unique normal forms) can be associated with congruence classes. 
Further, a unique reduced ground rewrite system serves as a presentation of the congruence closure of a finite set of ground equations, which allows checking whether two different
finite sets of ground equations define the same congruence closure or if one is contained
in the other. In the presence of disequations on ground terms with AC and uninterpreted
symbols, a finite set of ground equations and disequations can be checked for satisfiability.

The main contributions of the paper are (i) a modular combination
framework for the congruence closure of ground equations with multiple AC symbols, uninterpreted symbols, and constants, leading to (ii) modular and simple algorithms that can use flexible termination orderings on ground terms and do not need to use AC/E unification algorithms for generating canonical forms; (iii) the
termination and correctness proofs of these algorithms are modular and easier.
The key insights are based on extending the author's previous work presented in \cite{KapurRTA97,KapurJSSC19}:
introduction of new constants for nested subterms, resulting in flat and constant equations, extended to purification of mixed subterms with many AC symbols; AC ground terms are flattened and new constants for pure AC terms in each AC symbol are introduced, resulting in disjoint subsets of ground equations on single AC symbols with shared constants.
The result of this transformation is
a finite union of disjoint subsets of
ground equations with shared constants: (i) a finite set of constant equations, (ii) a finite set of flat equations with uninterpreted symbols, and (iii) for each AC symbol, a finite set of equations on pure flattened terms in a single AC symbol.

With the above decomposition, reduced canonical rewrite systems are generated for each of the subsystems using their respective termination orderings that extend a common total ordering on constants. A combination of the reduced canonical rewrite systems is achieved by propagating constant equalities among various rewrite systems; whenever new implied constant equalities are generated, each of the reduced canonical rewrite systems must be updated with additional computations to ensure their canonicity. Due to this modularity and factoring/decomposition, the termination and correctness proofs can be done independently of each subsystem, providing considerable flexibility in choosing termination orderings. 

The combination algorithm terminates when no additional implied constant equalities are generated. Since there are only finitely many constants in the input and only finitely many constants are needed for purification, the termination of the combination algorithm is guaranteed. The result is a reduced canonical rewrite system corresponding to the AC congruence closure of the input ground equations, which is unique for a fixed family of total orderings on constants and different
pure AC terms in each AC symbol. The reduced canonical rewrite system can be used to generate canonical signatures of ground terms with respect to the congruence closure.

\ignore{

This is achieved by using an algorithm to generate a reduced canonical rewrite system for terms constructed using a single AC symbol and constants. Such a reduced canonical rewrite system is made disjoint by ensuring that if the left side of a rewrite rule is a constant, then the right side is also a constant.

In this way, different reduced canonical rewrite systems for flattened terms with different AC symbols are disjoint except for sharing constants. These disjoint reduced canonical AC rewrite systems are combined with a reduced canonical rewrite system for congruence closure in uninterpreted symbols; once again, this rewrite system also only shares constants with other rewrite systems. The result is a combined reduced canonical rewrite system in which all rules other than the constant rules have function terms on their left sides such that no two different left sides are identical, after the constants in all rules have been reduced to their canonical forms using constant rewrite rules. 

This generalizes the approach
of \cite{KapurRTA97}.}

The framework provides flexibility
in choosing orderings on constants and terms with different AC symbols, enabling canonical forms suitable for applications instead of restrictions imposed due to the congruence closure algorithms. Interpreted AC symbols can be further enriched with properties
 including commutativity, idempotency, nilpotency, existence of
identities, cancelation and group properties, without 
restrictions on orderings on mixed terms.
Of particular interest is a new algorithm for generating congruence closure of ground equations with a cancelative AC symbol; this is discussed in detail as a separate topic because of new kinds of critical pairs, called {\bf cancelative disjoint superposition}, needed to generate a canonical rewrite system. 
Termination and correctness proofs of these congruence closure algorithms are modular and simple
in contrast to complex arguments and proofs
in \cite{BTVJAR03,NR91}.
\ignore{Although not elaborated in this paper,
canonical rewrite systems for conditional congruence closures for
uninterpreted and AC symbols can also be constructed
by using the construction of a canonical Horn rewrite system from
Horn rules on constants as proposed in \cite{KapurJSSC19}.}
These features of the proposed algorithms make them attractive for integration into SMT solvers as their implementation does not need heavy duty infrastructure including AC unification, extension rules, and AC compatible orderings.

The next subsection contrasts in detail the results of this paper with previous methods, discussing the advantages of the proposed framework and the resulting algorithms. Section 2 includes definitions of congruence closure with uninterpreted and
interpreted symbols. This is followed by a review of key constructions used
in the congruence closure algorithm over uninterpreted symbols as proposed
in \cite{KapurRTA97}. Section 3 
introduces purification and flattening of ground terms with AC and uninterpreted symbols by extending the signature
and introducing new constants.
This is followed by an algorithm first reported in \cite{KapRTA85} for computing the congruence closure and the associated
canonical rewrite system from ground equations with a single AC symbol and constants. In section 4, it is shown how
additional properties of AC symbols such as idempotency, nilpotency, identity and their combination can be integrated into the algorithm. Section 5 presents a new algorithm for generating a congruence closure algorithm for ground equations with a cancelative AC symbol; new kinds of superpositions and critical pairs are needed to generate a canonical rewrite system since unlike most congruence closure algorithms, it becomes necessary to explore terms larger than appearing in the input and/or a query to get congruences on terms using the cancelation inference rule. That section also presents a
congruence closure algorithm for ground equations with an AC symbol satisfying group properties; it is shown to be related to Gaussian elimination.
Section 6 discusses how the propagation of new constant equalities deduced from various algorithms update various kinds of ground canonical rewrite system.
In Section 7, an algorithm for computing congruence closure of AC ground equations with multiple AC symbols and constants is presented.
Section 8 generalizes to the case
of combination of AC symbols and uninterpreted symbols.
Section 9 shows how lexicographic orderings in which a constant can be bigger than a nonconstant ground term can be supported in the proposed framework.
Section 10 discusses a variety of examples illustrating the proposed algorithms.
Section 11 illustrates the power
and elegance of the proposed framework by demonstrating how a congruence closure  algorithm for two AC symbols with some additional properties, particularly distributivity, can be further generalized to formulate a \Groebner basis algorithm on
polynomial ideals over the integers. Section 12 concludes with some ideas for further investigation.

\subsection{Related Work}

Congruence closure algorithms have been developed
and analyzed for over four decades \cite{DST80,Shostak78,NO80}.
The algorithms presented here use the framework first informally
proposed in \cite{KapurRTA97} for congruence
closure in which the author separated the algorithm
into three parts: (i) constant equivalence closure, and (ii)
nonconstant flat terms related to constants by
flattening nested terms by introducing new constants to stand for them,  and (iii) updating of nonconstant rules as constant equivalence closure
evolves and propagates new constant equalities. This simplified the presentation, the correctness
argument as well as the complexity analysis, and made the framework easier to generalize  to other settings including conditional congruence closure \cite{KapurJSSC19} and semantic congruence closure \cite{BK20}. Further, it enables the generation
of a unique reduced ground canonical rewrite system for a congruence
closure, assuming a total ordering on constants; most importantly, the framework gives freedom in
choosing orderings on ground terms, leading to desired canonical forms
appropriate for applications. 
\ignore{The conditional congruence closure algorithm in \cite{KapurJSSC19}
made this framework explicit and generalized it to conditional
congruence closure of ground Horn equations on uninterpreted
symbols; extensions to interpreted symbols with simple properties
such as idempotency, commutativity, nilpotency, etc. were also mentioned.}

To generate congruence closure in the presence of AC symbols, the proposed framework 
builds on the author and his collaborators' work dating back to 1985,
where they demonstrated how an ideal-theoretic approach based on \Groebner basis
algorithms could be employed for word problems and unification problems over commutative
algebras \cite{KapRTA85}. 

Congruence closure algorithms on ground equations with interpreted symbols can be viewed as special cases of the Knuth-Bendix completion procedure \cite{KB} on (nonground) equations with  
universal properties characterizing the semantics of the interpreted symbols. In case of 
equations with AC symbols, Peterson and Stickel's extension of the Knuth-Bendix completion \cite{PS81} using extension rules, AC unification and AC compatible orderings can be used for congruence closure over AC symbols. For an arbitrary set $E$ of universal axioms characterizing the semantics of interpreted symbols, $E$-completion with coherence check and E-unification along with $E$-compatible orderings need to be used. Most of the general purpose methods do not terminate in general. Even though the Knuth-Bendix procedure can be easily proved to terminate on ground equations of uninterpreted terms, that is not necessarily the case for its extensions for other ground formulas.

In \cite{BL81}, 
a generalization of the Knuth-Bendix
completion procedure \cite{KB} to handle AC symbols \cite{PS81}
is adapted to develop decision algorithms for word problems over finitely presented commutative semigroups; this is equivalent to the congruence closure of ground equations with a single AC symbol on constants. Related methods using extension rules introduced to handle AC symbols and AC rewriting for solving word problems over other finite presented commutative algebras were subsequently reported in \cite{LeChenadec83}. 

In \cite{NR91}, the authors used the completion
algorithm discussed in \cite{KapRTA85} and a total AC-compatible polynomial reduction ordering on congruence classes of AC ground terms to establish the existence
of a ground canonical AC system first with one AC symbol. To extend their method to multiple AC symbols, particularly the instances of distributivity property relating ground
terms in two AC symbols $*$ and $+$,
the authors had to combine an AC-compatible total reduction ordering on
terms along with complex polynomial interpretations with
polynomial ranges, resulting in a complicated
proof to orient the distributivity axiom from left to right.
Using this highly specialized generalization of polynomial orderings, it was proved in \cite{NR91} that every ground AC theory has a finite canonical system
which also serves as its congruence closure.

The proposed approach, in contrast, is orthogonal to
ordering arguments on AC ground terms; any total ordering
on original constants and new constants is extended to one of many possible orderings on pure terms with
a single AC symbol; this ordering on AC terms suffices to compute a canonical ground AC rewrite system. Different orderings on AC terms with different AC symbols can be used; for example, ground terms of an AC symbol $+$ could be oriented in a completely different way than ground terms for another AC symbol $*$. Instances of the distributivity property expressed on different AC ground terms 
can also be oriented in different nonuniform ways.
This leads to flexible
ordering requirements on uninterpreted and interpreted symbols
based on the properties desired of canonical forms.

In \cite{Marche91}, a different approach was taken for computing a finite canonical rewrite system for ground equations on AC symbols. Marche first proved a general result about AC ground theories that for any finite set of ground equations with AC symbols, if there is an equivalent canonical rewrite system modulo AC, then that rewrite system must be finite. 
He gave an AC completion procedure, which does not terminate even on ground equations; he then proved its termination on ground equations with AC symbols using a special control
on its inference rules using a total ordering on AC ground terms in \cite{NR91}. Neither in \cite{NR91} nor in \cite{Marche91}, any explicit mention is made of uninterpreted symbols appearing in ground equations.

Similar to \cite{BL81}, several approaches based on adapting Peterson and Stickel's generalization of the Knuth-Bendix completion procedure to consider special ground theories have been reported \cite{LeChenadec83,Marche91}. 
In \cite{CCAC00,BTVJAR03}, the authors adapted Kapur's congruence closure
\cite{KapurRTA97} using its key ideas to an abstract inference
system.
Various congruence closure algorithms, including Sethi, Downey and Tarjan \cite{DST80}, Nelson and Oppen \cite{NO80} and Shostak \cite{Shostak78}, from the literature can be expressed as different combinations of these inference steps. They also proposed an
extension of this inference system to AC
function symbols, essentially integrating it with 
Peterson and Stickel's generalization
\cite{PS81} of the Knuth-Bendix completion procedure for AC symbols.
All of these approaches based on Peterson and Stickel's generalization used extension rules introduced in \cite{PS81} to
define rewriting modulo AC theories so that a local-confluence test for rules
with AC symbols could be developed using AC unification. During completion on ground terms, rules with variables appear in intermediate computations. All of these approaches suffer
from having to consider many unnecessary inferences due to extension rules and AC unification, as it is well-known that AC unification can generate doubly exponentially
many unifiers \cite{KN92}.

An approach based on normalized rewriting was proposed in \cite{MarcheNormalized} and
decision procedures were reported for ground AC theories with AC symbols satisfying additional properties including idempotency, nilpotency and identity as well as their combinations. This was an attempt to integrate Le Chenadec's method \cite{LeChenadec83} for finitely presented algebraic structures with Peterson and Stickel's AC completion, addressing weaknesses in E-completion and constrained rewriting; AC compatible termination orderings are difficult to design in the presence of AC symbols with identity, idempotency and nilpotency  properties. Such approach had to redefine local confluence for normalized rewriting and normalized critical pairs, leading to a complex
completion procedure whose termination and proof of correctness needed extremely
sophisticated machinery of normalized proofs.

The algorithms presented in this paper, in contrast, are very different and are based on an approach first presented in \cite{KapRTA85} by the author with his collaborators. 
Their termination and correctness proofs are based on
the termination and correctness proofs of a congruence closure algorithm for uninterpreted symbols (if present) and the termination and correctness of an algorithm for
deciding the word problems of a finitely presented commutative semigroup using Dickson's Lemma. Since the combination is done by propagating equalities on shared constants among various components, the termination and correctness proofs of the combination algorithm become much easier since there are only finitely many constants to consider, as determined by the size of the input ground equations.

\ignore{There is no need to use complex AC term ordering, resulting in
considerable flexibility in the use of termination ordering on ground AC terms, instead of having to worry about AC compatible termination ordering on 
general terms and preserved under ground substitutions; different instances of nonground AC terms can be compared in different flexible ways as mentioned above for the distributivity equation. 
The completion method used in \cite{NR91} is related to \cite{KapRTA85}, however it must use AC compatible orderings.}
\ignore{This is a sharp contrast to approaches based on adapting \cite{PS81} using extension rules and AC unification; as a result, many joinable critical pairs are generated, making such approaches extremely inefficient, and thus not suitable for integration into SMT solvers.}

A detailed comparison leads to several reasons why the proposed algorithms are simpler, modular, easier to understand and prove correct: (i) there is no need in the proposed
approach to use extension rules whereas almost all other approaches
are based on adapting AC/E completion
procedures for this setting requiring considerable/sophisticated
infrastructure including AC unification and
E/normalized rewriting.
As a result, proofs of correctness and
termination become complex 
using heavy machinery including proof orderings
and normalized proof methods 
not to mention arguments dealing with fairness of completion procedures.
(ii) all require complex total
AC compatible orderings. 
In contrast, ordering restrictions in the proposed algorithms
are dictated by individual components; little restriction is imposed for the
uninterpreted part; independent separate orderings on $+$-monomials for
each AC symbol + that extend a common total ordering
on constants can be used, thus giving considerable flexibility in choosing termination orderings overall.
In most related approaches except for \cite{NR91}, critical
pairs computed using expensive AC unification steps are needed,
which are likely to make the algorithms inefficient; it is
well-known that many superfluous critical pairs are generated due
to AC unification. These advantages make us predict that the
proposed algorithms can be easily integrated with  SMT solvers since
they do not require sophisticated machinery of AC-unification and
AC-compatible orderings, extension rules and AC completion.
\ignore{Further canonical rewrite systems generated are simpler: constant
rules, flat rules of the form $f(\cdots) \rew c$ for uninterpreted
symbols, and for each AC symbol $+$, $a_1 + ... + a_j = b_1 + ... + b_k,$
where $a,b$'s are constants.}

\ignore{
In the proposed framework, different AC symbols in contrast are considered separately
and ground equations in which multiple AC symbols appear, are
separated/purified using new constant symbols. Consequently, ground equations can be oriented in a flexible way in our
framework without having to design complex termination orderings. So
ground distributivity equations are purified and 
new symbols are introduced for $a * b, a * c, b + c$ respectively, say $u_1, u_2,
v_1,$, producing $u_1 + u_2 = a * v_1$; to have flexibility for
orienting this equation and hence the distributivity equality,
two new constants are introduced: $u_1 + u_2 = v_2$ and $a * v_1 = v_2$. 
There are equalities over $*$ terms as well as $+$ terms: $a * b = u_1, a * c = u_2, a * v_1 =
v_3, u_1 + u_2 = v_2, b + c = v_1$. 
Canonical rewrite systems can be easily generated assuming total orderings on constants including new constants,
$+$ and $*$ terms independent of ground distributivity equations.
An ordering  
with $v_3 \grt v_2$ ensures that 
the distributivity equality is oriented from left to  
right as is usually the case but the distributivity could also be oriented from right to left by making $v_2 \grt v_3$.}

\ignore{Consider the case when $(a * a) + (b * (a + c)) = (a
+ b) * c$. Purification leads to $a * a = w_1, a + c = w_2, b *
w_2 = w_3, a + b = w_4, w_4 * c = w_5,  w_1 + w_3 = w_5$. }

\ignore{On the original signature, the rewrite system is: $LDO = \{ ~ a * (b 
+ c)  \rew (a* b + a * c) \}$ with 
trivial equalities: $\{ c * (a *  b)  = b * (a * c) , (a * c) * ( b + c)  = c * (a * b 
+ a * c) ,  b * (a* b + a* c))   = (a * b)  * (b + c) \}$ in
which both left and right sides are equivalent after applying distributivity.}

\ignore{For orienting the distributivity equality from right to left,
the ordering constraint is $v_2 \grt v_3$ 
giving the rewrite
system: $LD' = \{ a * b \rew u_1, 
~ a* c \rew u_2, ~ a * v_1 \rew v_2, ~u_1 + u_2 \rew v_3, ~b + c 
\rew v_1, v_2 \rew v_3, v_1 \rew v_2 \}$ with
additional equalities $\{ c * u_1 = b * u_2, u_2 * v_2 = c * v_3, b *
v_3  = u_1 * v_3\}$ which is oriented depending upon the
ordering on the constants.
On the original signature, the rewrite system $LDO' = \{ (a* b +
a * c) \rew a * (b+c) \}$ with trivial equalities, again,
$\{ c * (a * b)  = b * (a  * c) , (a * c) * ( b + c)  = c * (a *
(b + c)),  b * (a* (b +  c))   = (a * b)  * (b + c) \}$.
A reader can observe that
using the extended signature gives considerable
flexibility for choosing total orderings on ground terms in generating locally confluent systems on the
original signature.}

\ignore{'These rewrite rules also define a decision procedure for  AC congruence closure on
the extended signature as well as on the original signature as
long as subterms in the original signature for which new symbols are introduced, are
replaced using the associated new symbols, thus transforming it
from a term in the original signature to a term in the extended signature.}

\ignore{The canonical form of $a + b$ is of
course $a * b$; assuming $g \grt *$, the canonical form of $g(e)$
is $ a * c$, implying that the canonical form of $a * b * c$ is
$a * b * c$; assuming $* \grt g$, the canonical form of $g(e)$ is
$g(f)$ implying that the canonical form of $a * b * c$ is $b * g(f)$.
If there are multiple ways to define a new constant,
  which one to use as definition? smaller one in the ordering?}

\ignore{Much like in the uninterpreted case, the resulting rewrite system
serves as a decision procedure of $CC(S)$ even when the input $S$
includes Horn equations.}

Since the proposed research addresses combination of congruence closure algorithms, a brief comparison with the extensive literature on the modular combination of decision procedures for theories with disjoint signatures but sharing constants, as well as  combining matching and unification algorithms must be made. Combination problems are a lot easier to consider for ground theories. The reader would notice that much like the seminal work of Nelson and Oppen where they showed that for combining satisfiability procedures of theories under certain conditions, it suffices to share equalities on shared constants and variables, the proposed combination framework only shares constant equalities even though there are no additional restrictions such as that ground equational theories in this paper satisfy the stable finiteness condition (see \cite{BaaderT97} for a detailed discussion of implications of Nelson and Oppen's requirements for equational theories). More general combination results for decidable theories can be found in \cite{ArmandoBRS09} in which a variable inactivity condition (disallowing literals of the form $x = t$  in clauses thus prohibiting paramodulation on a variable), a kind of ``pseudo" collapse-free equation, are imposed; such combination frameworks are likely to need more sophisticated machinery to consider AC symbols.

\section{Preliminaries}
\label{sec:preliminaries}

Let $F$ include a finite set $C$ of constants, a finite
set $F_U$ of uninterpreted symbols, and a finite set $F_{AC}$ of AC symbols. i.e., $F = F_{AC} \cup F_U \cup C.$
Let $GT(F)$ be the ground terms constructed from $F$;
sometimes, we will write it as $GT(F, C)$ to highlight the constants $C$ of $F$.
We will abuse the terminology by calling a $k$-ary  function
symbol as a function symbol 
if $k > 0$  and constant if $k =
0$. A function term is meant to be a nonconstant term
with a nonconstant outermost symbol.
Symbols in $F$ are either uninterpreted (to mean no semantic property of
such a function is assumed) or interpreted satisfying properties
expressed as universally quantified equations (called universal equations).


\subsection{Congruence Relations}

\ignore{Let $S$ be a finite set $\{  (\bigwedge h_{1}^i = h_2^i) \implies (c_{1}^i =
c_2^i) | 1 \le i \le k  \}$ of conditional (Horn) equations where $h$'s and $c$'s are
ground terms from $GT(F)$; the hypothesis in a Horn (conditional) equation could be empty (meaning
True and hence an unconditional equation). In addition, a Horn
equation of the form $(\bigwedge h_{1}^i = h_2^i) \implies
{\tt False}$ is allowed.  A Horn equation $(\bigwedge h_{1}^i = h_2^i) \implies
{\tt True}$ is trivial.\footnote{If the hypothesis of a Horn
equation includes {\tt False}, the Horn equation is trivial.}
Trivial Horn equations are deleted. {\tt True} is deleted from the hypothesis of an equation. 
${\tt True} \implies c_1 = c_2$ is the same as the unconditional
Horn equation $c_1 = c_2$. Henceforth, it is assumed that this preprocessing
is performed on the input as well as any Horn equation generated
during the algorithms.}

\ignore{There are two goals: (i) generate congruence closure to decide
whether another conditional equation follows from $S$ or not,
and (ii) generate canonical forms of constrained
terms as well as conditional equations so that the algorithm can
be used dynamically in an SMT solver. Canonical
form computation can be used to decide membership in conditional
congruence closure; however, it can be much more expensive if
the cost of generating canonical forms is not amortized over numerous membership tests.}

\begin{defi}
Given a finite set $S = \{ a_i = b_i \mid 1 \le i \le m \}$ of
ground equations where $a_i, b_i \in GT(F)$, the congruence
closure $CC(S)$ is inductively defined as follows: (i) $S
\subseteq CC(S)$, (ii) for every $a \in GT(F)$, $a = a \in
CC(S)$, (iii) if $a = b \in CC(S),$ $b = a \in CC(S)$, (iv)
if $a = b$ and $b = c \in CC(S)$, $a = c \in CC(S)$, and (v) for every
nonconstant $f \in F$ of arity $k > 0$, if for all $1 \le k$,  $a_i = b_i \in CC(S)$, then $f(a_1, \cdots, a_k) = f(b_1, \cdots, b_k)
\in CC(S)$. Nothing else is in $CC(S)$.
\end{defi}

$CC(S)$ is thus the smallest relation that includes $S$ and is
closed under reflexivity, symmetry, transitivity, and under function application. It is easy to see that
$CC(S)$ is also the equational theory of $S$ \cite{BN98,BK20}.

\subsection{A Congruence Closure Algorithm for Uninterpreted Symbols}
The algorithm in \cite{KapurRTA97,KapurJSSC19} for computing congruence closure of a finite set $S$ of ground equations over uninterpreted function symbols and constants serves as the main building block in this paper. The algorithm extends the input signature by introducing new constant
symbols to recursively stand for each nonconstant subterm and generates two types
of equations: (i) constant equations, and (ii) flat terms of the form
$h(c_1, \cdots, c_k)$ equal to constants (also called flat equations).\footnote{This representation of ground equations can also be viewed as an equivalence relation on nodes, named as constants, in a DAG representation of ground terms in which nonleaf nodes are labeled with function symbols and represent subterms, with a edge from the function labeling the node to the nodes serving as its arguments. A DAG representation supports full sharing of ground subterms.}
It can be proved that the congruence closure of ground equations on the extended signature when restricted to the original signature, is indeed the congruence closure of the original equations \cite{BK20}. A disequation on ground terms is converted to a disequation on constants by introducing new symbols for the ground terms. 

As given in Section 3 of \cite{BK20}, assuming a total ordering on constants, the algorithm (i) computes a canonical rewrite system from constant equations, picking a canonical representative, which is the least constant among all constants in each congruence class, and (ii) replaces constants in the left side of each flat rule by their canonical representatives, and identifies flat rules with identical left sides, replacing them by a single rule with the least constant on its right side, thus possibly generating new constant equalities. These steps are repeated if at least one new constant equality is generated in (ii).

Using a total ordering on constants (typically with new constants introduced to extend the signature being smaller than constants from the input), the output of the algorithm in \cite{KapurRTA97,KapurJSSC19,BK20} is a reduced canonical rewrite system $R_S$ associated with $CC(S)$ (as well as $S$) that includes {\it function} rules, also called {\it flat} rules, of the form $h(c_1, \cdots, c_k) \rew d$ 
and {\it constant} rules $c \rew d$ such that
no two left sides of the rules are identical; further, all constants are in canonical forms. The following result is proved in \cite{KapurRTA97} (see also \cite{NOEncode}).
\begin{thmC}[\cite{KapurRTA97}]
Given a set $S$ of ground equations, a reduced canonical rewrite system $R_S$ on the extended signature, consisting of nonconstant flat rules $ h(c_1, \dots, c_k) \rew d,$ and constant rules $c \rew d$, can be generated from $S$ in $O(n^2)$ steps. The complexity can be further reduced to $O(n*log(n))$ steps if  all function symbols are binary or unary. For a given total ordering $\gg$ on constants, $R_S$ is unique for $S$, subject to the introduction of the same set of new constants for nonconstant subterms. 
\end{thmC}

As shown in \cite{DST80}, function symbols of arity $> 2$ can be encoded using binary symbols using additional linearly many steps.

The canonical form of a ground term $g$ using $R_S$ is denoted by $\hat{g}$ and is its canonical signature (in the extended language). Ground terms $g_1, g_2$ are congruent in $CC(S)$ iff $\hat{g_1} = \hat{g_2}$. 

The output of the above algorithm can also be viewed as consisting of disjoint rewrite systems corresponding to (i) a rewrite system on constants representing Union-Find algorithm in which the rewrite rules correspond to the edges in a forest of trees such that if there is an edge from a constant $c_i$ to $c_j$ toward a root (implying that $c_i > c_j$), then there is a rewrite rule $c_i \rew c_j$, and (ii) for each nonconstant uninterpreted symbol $h \in F_U$, the set of uninterpreted function symbols, there is a finite set of flat rewrite rules  of the form $h(c_1, \cdots, c_k) \rew d$ on root constants in the forest with the property that no two left sides of the rewrite rules are identical. However, if Union-Find algorithm is used for generating a canonical rewrite system on constant equations, ordering on constants cannot be chosen a priori and instead should be chosen dynamically based on how various equivalence classes of constants get merged to maintain the balanced trees representing each equivalence class. For simplicity, in the rest of the paper, a total ordering on constants a priori will be assumed for generating a canonical rewrite system $R_S$ to represent the congruence closure $CC(S)$ of $S$.

\subsection{AC Congruence Closure}
The above definition of congruence closure $CC(S)$ is extended to consider
interpreted symbols. Let $IE$ be a finite set of universally quantified equations with variables, specifying properties of 
interpreted function symbols in $F$.
For example, the properties of an AC
symbol $f$ are: $\forall x, y, z, ~f(x, y) = f(y, x)$ and $f(x, f(y, z)) = f(f(x, y), z).$ An idempotent symbol $g$, for another example, is specified as $\forall x, ~g(x, x) = x$. 
To incorporate the semantics of these properties:

{\em (vi) from a universal axiom $s = t \in IE$, for 
all variables $x_i$'s in $s, t$ and
any ground substitution
$\sigma$ of $x_i$'s, $\sigma(s) = \sigma(t) \in CC(S)$.}

$CC(S)$ is thus the smallest relation that includes $S$ and is
closed under reflexivity, symmetry, transitivity, function application, and the 
substitution of variables in $IE$ by ground terms.

Given a finite set $S$ of ground equations with uninterpreted and interpreted symbols, the congruence closure membership problem is to check whether another  ground equation $u = v \in CC(S)$
(meaning semantically that $u = v$ follows from $S$, written as
$S \models u = v$). A related problem is whether given two sets $S_1$ and $S_2$ of ground equations,
$CC(S_2) \subseteq CC(S_1)$, equivalently $S_1 \models S_2$. Birkhoff's theorem relates the syntactic properties, the equational theory, and the semantics of $S$.

If $S$ also includes ground disequations, then besides checking the unsatisfiability of
a finite set of ground equations and disequations, new disequations 
can be derived in case $S$ is satisfiable.
The inference rule for deriving new disequations for an uninterpreted symbol is:

$f(c_1, \dots, c_k) \neq f(d_1, \cdots, d_k) \implies (c_1 \neq d_1 \lor \cdots \lor c_k \neq d_k).$

\noindent
In particular, if $f$ is unary, then the disequation is immediately derived.

\ignore{Consider particularly an example from \cite{CCVAR}, from
$g(f(a), h(b)) \neq g(f(b), h(a))$, the disequation $a \neq b$ is derived; this is
obtained by two applications of the above rule: from the disequation,
$f(a) \neq f(b) \lor h(a) \neq h(b)$, each of the disequations in the disjunction
implies $a \neq b$.}

Disequations in case of interpreted symbols generate more interesting formulas. In case of a commutative symbol $f$, for example, the disequation $f(a, b) \neq f(c, d)$ implies $(a \neq c \lor b \neq d) \land  (b \neq c \lor a \neq d ).$
For an AC symbol $g$, there are many possibilities; as an example, for the disequation
$g(a, g(b, c)) \neq g(a, g(a, g(a, c)))$ implies
$(a \neq g(a, c) \lor b \neq a \lor c \neq a) \land ( b \neq g(a, c) \lor  g(a, c) \neq g(a, a)) \cdots ).$

To emphasize the presence of AC symbols, let $ACCC(S)$ stand for the AC congruence closure of $S$ with $AC$ symbols; we will interchangeably use $ACCC(S)$ and $CC(S)$.

\subsection{Flattening and Purification}

Following \cite{KapurRTA97}, ground equations in $GT(F)$ with AC symbols are transformed into three kinds of equations by introducing new constants for subterms: (i) constant equations of the form $c = d,$
(ii) flat equations with uninterpreted symbols of the form $h(c_1, \cdots, c_k) = d$, and (iii) for each
$f \in F_{AC},$ $f(c_1, \cdots, c_j) = f(d_1, \cdots, d_{j'})$, where $c$'s, $d$'s are constants in $C$, $h \in F_U$, and every AC symbol $f$ is viewed to be variadic (including $f(c)$ standing for $c$).
Nested subterms of every $AC$ symbol $f$ are repeatedly 
{\bf flattened}: $f(f(s_1, s_2), s_3)$ to $f(s_1, s_2, s_3)$ and
$f(s_1, f(s_2, s_3))$ to $f(s_1, s_2, s_3)$ until all arguments to
$f$ are constants or nonconstant terms with outermost symbols different from $f$.\footnote{There are two types of flattening being used: (i) one for flattening arguments of AC symbols, converting them to be of variadic, (ii) creating flat terms in which function symbols have constants as arguments by extending the signature. This abuse of terminology can hopefully be disambiguated from the context of its use.}
Nonconstant arguments of a mixed AC term $f(t_1, \cdots, t_k)$ are transformed to $f(u_1, \cdots, u_k)$, where
$u_i$'s are new constants, with $t_i = u_i$ if $t_i$ is not a constant. 
A subterm whose outermost function symbol is uninterpreted, is also flattened by introducing new constants for their nonconstant arguments.
These transformations are recursively applied on the equations including those with new constants.

\ignore{

A mixed flattened term in many AC symbols is {\bf purified} using new constant symbols to stand for its nonconstant subterms; this gives unmixed flattened terms only in a single AC symbol applied on constants. Nested uninterpreted terms are also flattened using new constants so that every uninterpreted symbol only has constants as arguments.

Below, $c, d, c_i$'s, $d_j$'s are constants and $s_i$'s, $t_j$'s are either constant or nonconstant subterms. New constants are introduced only for nonconstant subterms. 
}
New constants are introduced only for nonconstant subterms and their number is 
minimized by introducing a single new constant
for each distinct subterm irrespective of its number of
occurrences (which is equivalent to representing terms by directed acyclic graphs (DAGs) with full sharing
whose each non-leaf node is assigned a distinct constant).
As an example, $((f(a, b) * g(a)) + f(a + (a + b), (a * b) + b)) * ((g(a) + ((f(a, b) + a) + a)) + (g(a) * b)) = a$ is purified and flattened with new constants $u_i$'s, resulting in
$\{ f(a, b) = u_1,  g(a) = u_2, u_1 * u_2 = u_3, a + a + b = u_4, a * b = u_5, u_5 + b = u_6, f(u_4, u_6) = u_7, u_3 + u_7 = u_8, u_2 * b = u_9, u_2 + u_1 + a + a + u_9 = u_{10}, u_7 * u_{10} = a\}$.

\ignore{
\vspace*{2mm}
\noindent
$FP(\{c = d \}) = \{ c = d \},$ ~~
$FP(\{f(c_1, \cdots,
c_k) = d\}) = \{f(c_1, \cdots, c_k) = d \}$, \\
\noindent
$FP(\{d = f(c_1, \cdots, c_k) \}) =  \{f(c_1, \cdots, c_k) = d \},$ ~~$FP(\{f(c_1, \cdots,
c_k) = f(d_1, \cdots,
d_{k'})\}) = \{f(c_1, \cdots,
c_k) = f(d_1, \cdots,
d_{k'})\}, f \in F_{AC},$ \\
\noindent
$FP(\{f(f(t_1, t_2), t_3) = s\}) = FP(\{f(t_1, t_2, t_3) = s\}), f \in F_{AC},$\\
\noindent
$FP(\{c = f(f(t_1, t_2), t_3) \}) = FP(\{f(t_1, t_2, t_3)) = c \}), f \in F_{AC},$\\
\noindent
$FP(\{f(t_1, f(t_2, t_3)) = s \}) = FP(\{f(t_1, t_2, t_3) = s\}), f \in F_{AC},$\\

\noindent
$FP(\{c = f(t_1, f(t_2, t_3)) \}) = FP(\{f(t_1, t_2, t_3) = c\}), f \in F_{AC},$\\ 
\noindent
$FP(\{f(c_1, \cdots, c_j) = f(d_1, \cdots, d_{j'})\}
=\{f(c_1, \cdots, c_j) = f(d_1, \cdots, d_{j'})\}, f \in F_{AC},$\\
\noindent
$FP(\{ f(t_1, \cdots, t_k) = c \}) =  \bigcup_{i=1}^k
FP(\{t_i = u_i\}) \cup \{
f(u_1, \cdots, u_k) = c\}, f \notin F_{AC} ~\lor~ (f \in F_{AC} \land fn(t_i) \neq f,$\ \
\noindent
$FP(\{ c = f(t_1, \cdots, t_k) \}) = FP(\{ f(t_1, \cdots, t_k) = c \}), f \notin F_{AC}$\\
\noindent
$FP(\{ f(t_1, \cdots, t_k) = g(s_1, \cdots, s_{k'}) \})$ = 
$FP(\{ f(t_1, \cdots, t_k) = u\})  \cup FP(\{g(s_1, \cdots, s_{k'}) = v\} \cup \{ u = v \}, f \neq g,$
\noindent
$FP(\{ s \neq t \}) = FP(\{ s = u\}) \cup FP(\{ t = v\}) \cup \{ u \neq v \}.$

\noindent
where new constant symbols $u_i$'s, $v_j$'s are introduced only
for nonconstant subterms, but constant symbols are not replaced
by new symbols.

The output of $FP(S)$ is a finite set of equations consisting of:
(i) constant equations: $c = d$, (ii) flat equations: $h(c_1, \dots, c_k) = d$, and (iii) for each $f \in F_{AC}$,  $f(c_1, \cdots c_j) = f(d_1, \cdots, d_{j'}).$  The arguments of an AC symbol are represented as a multiset since the order and multiplicity do not matter. For an AC symbol $f$, let $f(M)$ be a pure flattened 
term $f(a_1, \cdots, a_k)$ with $M = 
\mset{a_1, \cdots, a_k}$, a multiset of constants; $f(M)$ is called an {\bf $f$-monomial}. In case $f$ has its identity $e$, i.e., $f(x, e) = x$, then $e$ is written as is $f()$ or $f(\mset{})$. A singleton constant $c$ is written as is or equivalently 
$f(\mset{c})$ or further abusing the notation, as $f(c)$.
}

The arguments of an AC symbol are represented as a multiset since the order does not matter but multiplicity does. For an AC symbol $f$, if a flat
term is $f(a_1, \cdots, a_k)$, it is written as $f(M)$ with a multiset 
$M = \mset{ a_1, \cdots, a_k}$; 
$f(M)$ is also called an {\bf $f$-monomial}. In case $f$ has its identity $e$, i.e., $\forall x, f(x, e) = x$, then $e$ is written as is, or 
$f(\mset{})$. 
A singleton constant $c$ is ambiguously written as is or equivalently 
$f(\mset{c})$.
An $f$-monomial $f(M_1)$ is equal to $f(M_2)$ iff the multisets $M_1$ and $M_2$ are equal.

Without any loss of generality, the input to the algorithms below are assumed to be the above set of flat ground equations on constants, equivalently, equations of the form $f(M_i) = f(M_{i'})$.

\section{Congruence Closure with a single Associative-Commutative (AC) Symbol}

The focus in this section is on interpreted symbols with the
associative-commutative properties; uninterpreted symbols are considered later.

Checking whether a ground equation on AC terms is in the congruence closure of a finite set $S$ of ground equations
is the word problem over finitely presented commutative algebraic structures, presented by $S$
characterizing their interpretations as discussed in \cite{KapRTA85,LeChenadec83,BL81}.
\ignore{Given a finite set $S$ of ground equations expressed in multiple AC symbols and constants, the membership problem in $ACCC(S)$, the AC congruence closure of $S$, is check whether a given equation $s = t$ is in $ACCC(S)$.
In the presence of disequations over AC ground terms, one is also interested in
determining whether the set of ground equations and disequations is satisfiable or not.}
Another goal is to associate a reduced canonical rewrite system as a unique presentation of $ACCC(S)$ and a canonical signature with every AC congruence class in the AC congruence closure of a satisfiable $S$.

For a single AC symbol $f$ and a finite set $S$ of monomial equations $\{ f(M_i) = f(M'_i) | 1 \le i \le k\}$, $ACCC(S)$ is the reflexive, symmetric and transitive closure of $S$ closed under $f$. Section 2.3 gives a general definition of semantic congruence closure both in the presence of uninterpreted and interpreted symbols including AC symbols. For a single AC symbol $f$, $ACCC(S)$ becomes:
if $f(M_1) = f(M_2) $ and $f(N_1) = f(N_2)$ in $ACCC(S)$, where $M_1, M_2, N_1, N_2$ could be singleton constants,
then $f(M_1 \cup N_1) = f(M_2 \cup N_2)$ is also in $ACCC(S)$. 

\ignore{
Below, $S_C$ is finite set of constant equations and
$S_{f}$ is a finite set of equations $\{ f(M_i) = f(N_i) ~|~ 1 \le i \le k\}$ on $f$-monomials $f \in F_{AC}$. }

\ignore{\footnote{In case an AC symbol has additional
  properties, then interpreted terms can be normalized further to embed
  those properties in their representation. If $f$ has both the same left and right identity $e$, them $f(M) =
f(M-\{e\})$; similarly if $f$ is idempotent, then $f(M) = f(M')$
with $M'$ being $M$ with duplicates removed; for nilpotent $f$,
$f(M) = f(M')$ with $M'$ being $M$ with duplicates replaced by
$0$, which if identity, is deleted unless $M'$ becomes
singleton or $0$. } 
}

\ignore{In the next two subsections, we present algorithms for generating a reduced canonical 
rewrite system from $S$, starting with the case of a single AC symbol followed by multiple AC symbols.
These algorithms (admittedly in their nonmodular form) have been known to the author since 1985
\cite{KapRTA85} (see Section 6 there), where word problems 
over finitely presented commutative algebras using
ideal-theoretic approach and \Groebner basis construction were first
discussed. This paper presents a modular and simpler version of these algorithms, resulting in considerable flexibility on termination orderings on terms. }

\ignore{We first present a method for generating
an AC congruence closure from
$S_f$ consisting of ground equations on flattened terms using a single AC symbol $f$. This is done by computing a canonical
rewrite system from $S_f$ which also associates canonical signatures with flattened  ground $f$ terms.}
\ignore{The congruence closure membership problem is the same as deciding the word
problem of a finitely presented commutative semigroup.
In the subsequent section, we generalize the method to combining several AC symbols and
constants.
Using this, it would be possible to decide the word
problem of finitely presented commutative rings with $+$ and $*$
with the universal distributivity law. Finally, a combination method of
these theories with $EUF$, the theory of uninterpreted symbols with
equality, will be presented. Using this, the word problem of
theories with many AC symbols and uninterpreted symbols can be
decided. }


As in \cite{KapurRTA97}, we follow a rewrite-based approach for computing the AC congruence closure $ACCC(S)$ by generating a canonical rewrite system from $S$.
To make rewrite rules from equations in $S$, a total ordering $\gg$ on the set $C$ of constants is extended to a total ordering on $f$-monomials and denoted as $\gg_f$.
One of the main advantages of the proposed approach is the flexibility in using termination orderings on $f$-monomials, both from the literature on termination orderings on term rewriting systems as well as well-founded orderings (also called admissible orderings) from the literature on symbolic computation including \Groebner basis.

Using the terminology from the \Groebner basis literature, an ordering $\gg_f$ on the set of $f$-monomials,
$GT(\{f\}, C)$, is called {\it admissible} iff 
i)
$f(A) \gg_f f(B)$ if the multiset $B$ is a proper subset of the multiset $A$ (subterm property), and 
(ii)  for any multiset $B$, $f(A_1) \gg_f f(A_2) \implies
f(A_1 \cup B) \gg_f f(A_2 \cup B)$ (the compatibility property). 
$f(\mset{})$
may or may  not be included in $GT(F)$ depending upon an application. Two popular families of admissible orderings on monomials are degree ordering (also called degree-lexicographic orderings).
A degree ordering compares multisets $f(A)$ and $f(B)$ first on their size and in case of equal size, the largest constant in the difference $A-B$ is bigger than every constant in $A-B$:
$f(A) \gg f(B)$ if and only if $|A| > |B|$ or 
if $|A| = |B|$, then $\exists c \in A-B, \forall d \in B-A, c \gg d$. In a lexicographic ordering.
size does not matter: $f(A) \gg f(B)$ if and only if $\exists c \in A-B, \forall d \in B-A, c \gg d$.

From $S$, 
a rewrite system $R_{S}$ is associated with $S$ by orienting nontrivial
equations in $S$ (after deleting trivial equations $t = t$ from
$S$) using $\gg_f$: a ground equation $f(A_1) = f(A_2)$ is oriented into
a terminating rewrite rule  $f(A_1) \rew f(A_2)$, 
where $f(A_1) \gg_f f(A_2)$.
The rewriting relation induced by 
this rewrite rule is defined below.
\begin{defi}
A flattened term $f(M)$ is {\em rewritten} in one step, denoted by $\rew_{AC}$ (or simply $\rew$),
using a rule $f(A_1) \rew f(A_2)$  to $f(M')$ iff $A_1 \subseteq M$
and $M' = (M - A_1) \cup A_2$, where $-, \cup$ are operations on
multisets.
\end{defi}

\vspace*{-1mm}
Given that $f(A_1) \gg_f f(A_2)$, it follows that $f(M) \gg_f f(M')$,
implying the rewriting terminates. Standard notation and concepts from
\cite{BN98} are used to represent and study properties of the reflexive and transitive closure of $\rew_{AC}$ induced by $R_{S}$; the reflexive, symmetric and transitive closure of $\rew_{AC}$
is the AC congruence closure $ACCC(S)$ of $S$. Below, the subscript $AC$ is dropped from $\rew_{AC}$ and $f$ is dropped both from
$S_f$ and $\gg_f$ whenever it is obvious from the context.

A rewrite relation $\rew$  defined by $R_{S}$ is called terminating iff there are no infinite rewrite chains of the form $t_0 \rew t_1 \rew \cdots \rew t_k \rew \cdots$. A rewrite relation $\rew$ is {\em locally confluent} iff
for any term $t$ such that $t \rew u_1, t \rew u_2$, there exists $v$ such that $u_1 {\rew}^* v, u_2 {\rew}^* v$. $\rew$ is confluent iff 
for any term $t$ such that $t \rew^* u_1, t \rew^* u_2$, 
there exists $v$ such that $u_1 \rew^* v, u_2 \rew^* v$. $\rew$ is canonical iff it is terminating and locally-confluent (and hence also terminating and confluent). A term $t$ is in normal form iff there is no $u$ such that $t \rew u$.

An $f$-monomial $f(M)$ is in normal form with respect to $R_{S}$ iff $f(M)$ cannot be rewritten using any rule in $R_{S}$.

Define a nonstrict partial ordering on $f$-monomials, called the {\bf Dickson} ordering, informally
capturing when an $f$-monomial rewrites another $f$-monomial: 
$f(M) \gg_D f(M')$ iff $M'$ is a subset of $M$. It is easy to see that
for any admissible ordering $\gg_f$,  $\gg_D \subseteq \gg_f$.
Observe that the strict subpart of this ordering, while
well-founded, is not total; for example, two distinct singleton multisets (constants) 
$\mset{a} \neq \mset{ b}$ 
cannot be compared. This ordering is later used to show the termination of the completion algorithm using the Dickson's lemma, which states that any infinite subset of $f$-monomials on finitely many constants must include at least two $f$-monomials comparable by the Dickson ordering. The Dickson's lemma is a combinatorial property on $k$-tuples of natural numbers which states that any infinite subset $A \subseteq \mathbb{N}^k$ must include at least two comparable $k$-tuples by component-wise ordering.  Assuming an ordering on $k$ constants, an $f$-monomial can be represented as a $k$-tuple corresponding to the number of times various constants occur in the $f$-monomial.

A rewrite system $R_S$ is called {\em reduced} iff neither the left side nor the right side of any rule in $R_S$ can be rewritten by any of the other rules in $R_S$.

\ignore{The reflexive, symmetric and transitive closure of $\rew$ induced by a finite set
$R$ of rewrite rules obtained by orienting ground equations on AC terms with
$f$ and constants is the congruence closure of $R$ when viewed as equations.}

As in \cite{KapRTA85}, the local confluence of $R_{S}$ can be checked using the following constructions
of superposition and critical pair.

\begin{defi}
Given two distinct rewrite rules  $f(A_1)
\rew f(A_2), ~ f(B_1) \rew f(B_2)$, 
let $AB = (A_1 \cup B_1) - (A_1 \cap
B_1)$; $f(AB)$ is then the {\em superposition} of the two rules,
and the {\em critical pair} is $(f((AB-A_1) \cup A_2), f((AB-B_1) \cup B_2))$.
\end{defi}

To illustrate, consider two rules $f(a, b) \rew a, f(b, c) \rew
b$; their superposition $f(a, b, c)$ leads to the critical pair
$(f(a, c), f(a, b))$. 

A rule can have a constant on its left side and a nonconstant on its right side. 
As stated before, a singleton constant stands for the multiset containing that constant. 

A critical pair is {\it nontrivial} iff the normal forms of its two
components in $\rew_{AC}$ as multisets are not the same (i.e., they are not joinable).  A nontrivial critical pair generates an
implied equality relating distinct normal forms of its two components.

For the above two rewrite rules, normal forms of two sides are $(f(a, c), a)$, respectively, 
indicating that the two rules are not locally confluent. A new derived equality
is generated: $f(a, c) = a$ which is in $ACCC(\{f(a, b) = a, f(b, c) = b \}).$

It is easy to prove that if $A_1, B_1$ are disjoint multisets,
their critical pair is trivial. 
Many critical pair criteria to identify additional trivial
critical pairs have been investigated and proposed in \cite{WeisBook,KapurPrime88,BachmairDersho}.

\begin{lem}
An AC rewrite system $R_{f}$ is {\em locally confluent} iff the
critical pair: $(f((AB-A_1) \cup A_2), f((AB-B_1)
\cup B_2))$ between every pair of distinct rules $f(A_1) \rew f(A_2),
f(B_1) \rew f(B_2)$ is joinable, where
$AB = (A_1 \cup B_1) - (A_1 \cap B_1)$.
\end{lem}

\begin{proof}
\,Consider a flat term $f(C)$ rewritten in 
two different ways in one step using not necessarily distinct rules:
 $f(A_1) \rew f(A_2),
f(B_1) \rew f(B_2)$. The result of the rewrites is: $(f((C-A_1)
\cup A_2), f((C-B_1) \cup B_2))$. Since $A_1 \subseteq C$ as well as $B_1  \subseteq C$, $AB \subseteq C$; let
$D = C-AB$. The critical pair is then
$(f(D \cup ((AB-A_1) \cup A_2)), f(D \cup ((AB-B_1) \cup B_2)))$,
all rules applicable to the critical pair to show its
joinability, also apply, thus showing the joinability of the
pair. The other direction is straightforward.
The case of when at least one of the rules has a constant on its left side is trivially handled.
\end{proof}

\ignore{\begin{proof}
\,Consider a flat term $f(C)$ rewritten in 
two different ways in one step using not necessarily distinct rules:
 $f(A_1) \rew f(A_2),
f(B_1) \rew f(B_2)$. The result of the rewrites is: $(f((C-A_1)
\cup A_2), f((C-B_1) \cup B_2))$. Since $A_1 \subseteq C$ as well as $B_1  \subseteq C$, $AB \subseteq C$; let
$D = C-AB$. The critical pair is then
$(f(D \cup ((AB-A_1) \cup A_2)) f(D \cup ((AB-B_1) \cup B_2)))$,
all rules applicable to the critical pair to show its
joinability, also apply, thus showing the joinability of the
pair. The other direction is straightforward.
The case of when at least one of the rules has a constant on its left side is trivially handled. $\Box$
\end{proof}
}
Using the above local confluence check, a completion procedure is designed in the classical manner; equivalently, a nondeterministic algorithm can be given as a set of inference rules \cite{BN98}.
If a given rewrite system is not locally confluent, then new 
rules generated from nontrivial critical pairs (that are not joinable)
are added until the resulting rewrite system is locally confluent. New rules
can always be oriented since an ordering on $f$-monomials is assumed to be total. This completion algorithm is a special case of \Groebner basis algorithm on monomials built using a single $AC$ symbol. The result of the completion algorithm is a locally confluent and terminating rewrite system for $ACCC(S)$.

Applying the completion algorithm on the two rules in the above example,
the derived equality is oriented into
a new rule  $f(a,c) \rew a$.
The system $\{ f(a, b) \rew a, f(b, c) \rew b, f(a, c) \rew a\}$
is indeed locally-confluent. This canonical rewrite system is a presentation
of the congruence closure of $\{ f(a, b) = a, f(b, c) = b\}.$
Using the rewrite system, membership in its congruence closure can be decided by rewriting: 
$f(a, b, b) = f(a, b, c) \in ACCC(S)$ whereas $f(a, b, b) \neq f(a, a, b)$.

\ignore{
Applying this construction on the above example, 
the superposition is from 1 and 2.1 is $f(a, c, u_2)$ and the
critical pair is: $ f(a,
b) = f(a, u_2),$. From 1 and 2.0, the superposition is 
$f(a, b, c)$ and the critical pair is the same as 
the previous one: $f(a, b) = f(a, u_2)$. 
Similarly, 2.0 and 2.1 gives the superposition $f(b, c, u_2)$ and
the associated critical pair is $f(b, b) = f(u_2, u_2)$. All
other critical pairs are trivial. 

}

\ignore{
Consider an example from \cite{BTVJAR03}, where $E = \{ f(a, c) =
a, ~ f(c, g(f(b, c))) = b, ~ g(f(b, c)) =f(b, c) \}$. Introducing
new symbols for $f(b, c)$, say $u_1$, and $g(u_1) = u_2$, gives
$\{ f(a, c) = a, ~ f(c, u_2) = b, ~ g(u_1) = u_2, f(b, c) = u_2,
u_1 = u_2 \}$. For simplicity, let us separate out the equations
on uninterpreted symbols: $\{ f(a, c) = a,  ~ f(c, u_2) = b, f(b, c) = u_2,
u_1 = u_2 \}$; since $u_2$ is equal to a $AC$ term, it is substituted in
every equation in which $u_2$ is an argument to $f$, thus
replacing $f(c, u_2) = b$ by $f(b, b, c) = b$, giving
$\{ f(a, c) = a,  ~ f(b, c, c) = b, f(b, c) = u_2,
u_1 = u_2 \}$

Orienting the above equations to terminating rewrite rules gives: $ \{1. ~f(a, c) \rew a, ~ 2.0~ f(b, c) \rew u_2,
~2.1.~ f(b, c, c) \rew b, ~ 2.2. ~g(u_2) \rew  u_2, ~3.~ u_1 \rew
u_2\}$. Applying the above completion algorithm,
the superposition of 1 and 2 is $f(a, b, c, c)$ giving
the critical pair: $\langle f(a, b, c), f(a, b) \rangle$ which
reduce to the same normal form: $f(a, b)$. 
The superposition between 1 and 2.0 gives the critical pair:
$\langle f(a, b), f(a, u_2) \rangle$ giving a new rule $4. f(a, b)
\rew f(a, u_2)$ resulting in the canonical rewrite system $R_S$ consisting
of $\{ 1, 2.1, 2.2, 2.0, 3, 4 \}$.
On the original signature,   it is:  $\{1.~ f(a, c) \rew a, ~2.1.~ f(c, f(b, c)) \rew
  b, ~ 2.2~g( g(f(b, c))) \rew g(f(b, c)), ~3.~f(b, c) \rew
  g(f(b, c)), 4. f(a, b) \rew f(a, a, b)\}$ which is nonterminating
  but confluent, a situation similar to the uninterpreted case
  \cite{KapurRTA97}.

}

A simple completion algorithm is presented for the sake of completeness. It takes as input, a finite set $S_C$ of constant equations and a finite set $S$ of equations on $f$-monomials, and a total ordering $\gg_f$ on $f$-monomials extending a total ordering ordering  $\gg$ on constants, and computes a reduced 
canonical rewrite system $R_f$ (interchangeably written as $R_S$) such that $ACCC(S) = ACCC(S_{R_f})$, where
$S_{R_f}$ is the set of equations $l = r$ for every $l \rew r \in R_f$.

\vspace{2.5mm}
\hspace{-3mm}
{\bf SingleACCompletion($S = S_f \cup S_C$, $\gg_f$):}
\vspace*{1mm}
\begin{enumerate}

\item Orient constant equations in $S_C$ into terminating rewrite rules $R_C$ using $\gg$. Equivalently, using Tarjan's Union-Find data structure, for every constant $c \in C$, compute, from $S_C$, the equivalence class $[c]$ of constants containing $c$ and make $R_C = \cup_{c \in C} \{ c \rew \hat{c} ~| ~c \neq \hat{c} ~\mbox{and}~ \hat{c} ~\mbox{is the least element in}~ [c]\}.$

Initialize $R_f$ to be $R_C$. Let $T := S_f$.

\item Pick an $f$-monomial equation $l = r \in T$ using some selection criterion (typically an equation of the smallest size) and remove it
from $T$. Compute normal forms $\hat{l}, \hat{r}$ using $R_f$. If equal, then discard the equation, otherwise, orient into 
a terminating rewrite rule using $\gg_f$. Without any loss of generality, let the rule be $\hat{l} \rew \hat{r}.$

\item Generate critical pairs between $\hat{l} \rew \hat{r}$ and every $f$-rule in $R_f$, adding them to $T$.\footnote{Critical pair generation can be done incrementally using some selection criterion for pairs of rules to be considered next, instead of generating all critical pairs of 
the new rule with all rules in $R_f$.}

\item  Add the new rule $\hat{l} \rew \hat{r}$ into $R_f$; {\bf interreduce} other rules in $R_f$ using the new rule. 

(i) For every rule $l \rew r$ in $R_f$ whose left side $l$ is reduced by $\hat{l} \rew \hat{r}$, remove $l \rew r$ from $R_f$ and insert $l = r$ in $T$. 

(ii) If $l$ cannot be reduced but $r$ can be reduced, then reduce $r$ by the new rule and generate a normal form $r'$ of the result. Replace
$l \rew r$ in $R_f$ by
$l \rew r'$. 

\item Repeat the previous three steps until the critical pairs among all pairs of rules in $R_f$ are joinable, and $T$ becomes empty.

\ignore{\item During the generation of $R_f$, if a rule with $c \rew m,$
is generated with a constant $c$, then introduce a new constant $u$, extend the constant ordering to include $c \gg u$ and replace $c \rew m$ by two rules $c \rew u$ to be included
in $R_C$ and $m \rew u$ to be included in $R_f$. This ensures that the rules in $R_f$ always have a nonconstant on their left sides.\footnote{This trick is found useful while combining multiple reduced canonical rewrite system the congruence closures of ground equations with different AC symbols.}}

\item Output $R_f$ as the canonical rewrite system associated with $S$.
\end{enumerate}

\vspace*{5mm}
{\bf Example 1:} As a simple example, consider ground equations on an AC symbol $*$: $1.~ a * a * b = a * a, 2. ~a * b * b = b * b$ with the ordering $a > b$ on constants. They are oriented from left to right as terminating rewrite rules (using either total degree ordering or pure lexicographic ordering). The overlap of the two left sides gives the superposition $a * a * b * b$, giving the critical pair: $(a * a * b, a * b * b)$; their normal forms are:$(a * a, b * b)$. The new rule is oriented as: $3.~ a * a \rew b * b.$ This rule simplifies rule 1 to
$1'.~ b * b * b \rew b * b.$ The completion algorithm terminates with $1', 2, 3$ as the canonical rewrite system.

\begin{thm}\label{termination}
The algorithm {\bf SingleACCompletion}  terminates, i.e., in Step 4, rules to $R_f$ cannot be added infinitely often.
\end{thm}

\begin{proof}
By contradiction. A new rule $\hat{l} \rew \hat{r}$  in Step 4 of the algorithm is added to $R_f$ only if no other rule can reduce it, i.e, for every rule $l \rew r \in R_f,$
$\hat{l}$ and $l$ are noncomparable in the Dickson ordering $\gg_D$.
For $R_f$ to be infinite, implying the nontermination of the algorithm means that $R_f$ must include
infinitely many noncomparable left sides in $\gg_D$,
\ignore{ The left side of the new rule $l_i \rew r_i$  cannot be reduced by any rule previously added, i.e., all rules $l_j \rew r_j, j < i,$ $l_i \gg_D l_j$ is false. Since rules
in $R_f$ are always in normal form due to interreduction, $R_f$ being infinite means there are infinitely many $f$-monomials serving as the left sides of rules which are noncomparable in $\gg_D$,} a contradiction to the Dickson's Lemma.
This implies that there cannot be infinitely many noncomparable $f$-monomials serving as the left hand sides of rules in $R_f$. In the interreduction step, new rules, if any, added to replace a deleted rule $l \rew r$ if $l$ can be reduced, are always smaller in $\gg_f$.
\end{proof}

\begin{thm}\label{correctness}
Given a finite set $S$ of ground equations with a single AC symbol $f$ and constants, and a total admissible ordering $\gg_f$ on flattened AC terms and constants, a canonical rewrite system $R_f$ is generated by the above completion procedure, which serves as a decision procedure for $ACCC(S)$. 
\end{thm}

The proof of the theorem is classical, typical of a correctness proof of a completion algorithm based
on ensuring local confluence by adding new rules generated from superpositions whose critical pairs are not
joinable.

\noindent
\begin{thm}\label{unique}
Given a total ordering $\gg_f$ on $f$-monomials, there is a unique reduced canonical rewrite system associated with $S_f$.
\end{thm}

\begin{proof}
By contradiction. Suppose there are two distinct reduced canonical rewrite systems $R_1$ and $R_2$ associated with $S_f$ for the same $\gg_f$.
Pick the least rule $l \rew r$ in $\gg_f$  on which $R_1$ and $R_2$ differ; wlog, let $l \rew r \in R_1$.
Given that $R_2$ is a canonical rewrite system for $S_f$ and $l = R \in ACCC(S_f)$, $l$ and $r$ must reduce using $R_2$ implying that there is a rule $l' \rew r' \in R_2$ such that $l \gg_D l'$; since $R_1$ has all the rules of $R_2$ smaller than $l \rew r$, $l \rew r$ can be reduced in $R_1$, contradicting the assumption that $R_1$ is reduced.
If $l \rew r' \in R_2$ where $r' \neq r$ but $r' \gg_f r$, then $r'$ is not reduced implying that $R_2$ is not reduced.
\end{proof}

\ignore{\begin{proof}
Proof by Contradiction. A new rule $\hat{l} \rew \hat{r}$  in Step 4 of the algorithm is added to $R_f$ only if no other rule can reduce it, i.e, for every rule $l \rew r \in R_f,$
$\hat{l}$ and $l$ are noncomparable in $\gg_D$.
For $R_f$ to be infinite, implying the nontermination of the algorithm means that $R_f$ must include
infinitely many noncomparable left sides in $\gg_D$,
\ignore{ The left side of the new rule $l_i \rew r_i$  cannot be reduced by any rule previously added, i.e., all rules $l_j \rew r_j, j < i,$ $l_i \gg_D l_j$ is false. Since rules
in $R_f$ are always in normal form due to interreduction, $R_f$ being infinite means there are infinitely many $f$-monomials serving as the left sides of rules which are noncomparable in $\gg_D$,} a contradiction to Dickson's Lemma.
\end{proof}
}

\ignore{The termination of critical pair generation in the completion
procedure is guaranteed by 
Dickson's lemma \cite{WeisBook}. Nontermination would imply that infinitely many rules with left sides which cannot be reduced by other existing rules can be generated contradicting that there cannot be noncomparable $f$-monomials in the Dickson ordering $\gg_D$ among an infinite subset of $f$-monomials. In other words, an infinite set of $f$-monomials must include two $f$-monomials
comparable by a sub multiset ordering.
The completion algorithm cannot keep adding rules which are irreducible with respect to existing rules. More details about how Dickson's lemma is used for proving termination of \Groebner basis like completion algorithms on associative commutative structures can be found in \cite{WeisBook}.
}

\ignore{
\begin{thm}
Given a total ordering $\gg_f$ on $f$-monomials, there is a unique reduced canonical rewrite system associated with $S_f$
\end{thm}
}
\ignore{\vspace*{-4mm}
\begin{proof}
Proof by Contradiction. Suppose there are two distinct reduced canonical rewrite systems $R_1$ and $R_2$ associated with $S_f$ for the same $\gg_f$.
Pick the least rule $l \rew r$ in $\gg_f$  on which $R_1$ and $R_2$ differ; wlog, let $l \rew r \in R_1$.
Given that $R_2$ is a canonical rewrite system for $S_f$ and $l = R \in ACCC(S_f)$, $l$ and $r$ must reduce using $R_2$ implying that there is a rule $l' \rew r' \in R_2$ such that $l \gg_D l'$; since $R_1$ has all the rules of $R_2$ smaller than $l \rew r$, $l \rew r$ can be reduced in $R_1$, contradicting the assumption that $R_1$ is reduced.
If $l \rew r' \in R_2$ where $r' \neq r$ but $r' \gg_f r$, then $r'$ is not reduced implying that $R_2$ is not reduced.
\end{proof}
}

The complexity of this specialized decision procedure has been
proved to require exponential space and double exponential upper
bound on time complexity \cite{Mayr82}. 

\ignore{It should be noted that the above problem/algorithm has been extensively studied in the literature in the early 80's \cite{BL81,Mayr82,LeChenadec83,KapRTA85}.}

\ignore{It is easy to prove that for a given total ordering on $f$-monomials $\gg_f$, there is a unique reduced canonical
rewrite system among
all canonical rewrite systems generated from $S_f$ using $\gg_f$; in this sense, such a unique reduced canonical rewrite system is itself canonical in the class of equivalent canonical rewrite systems. }

The above completion algorithm generates a unique reduced canonical rewrite system $R_f$ for the congruence closure $ACCC(S)$ because of interreduction of rules whenever a new rule is added to $R_f$. $R_f$ thus serves as its unique presentation. Using the same ordering $\gg_f$ on $f$-monomials; two sets $S_1, S_2$ of AC ground equations have
identical (modulo presentation of multisets
as AC terms) reduced canonical rewrite systems  $R_{S_1}=R_{S_2}$ iff $ACCC(S_1) = ACCC(S_2)$, thus generalizing the result for the uninterpreted case. 
\ignore{Every $f$-monomial in $GT(\{f\}, C)$ has its canonical signature--its canonical form computed using $R_{f}$ generated from $S$.}

\ignore{In the absence of any restriction on a total ordering $\gg_f$ on $f$-monomials, a reduced canonical rewrite system $R_S$ can have
a rule in which the left side is a constant whereas its right side is a nonconstant $f$-monomial. This would particularly be the case if pure lexicographic ordering a la \Groebner basis algorithm is used in which a constant symbol $a$, say, is $\gg$ any $f$-monomial in which $a$ does not appear.\footnote{Such an ordering can also be defined using $lpo$ in which a constant is bigger than every symbol including $f$.}
}

\section{Congruence Closure of a Richer AC Theory with Idempotency, Nilpotency, Identity, and their combination}

If an AC symbol $f$ has additional properties such as nilpotency, idempotency and/or unit, the above completion algorithm can be easily extended by modifying the local confluence check. Along with the above discussed critical pairs from a distinct pair of rules, additional critical pairs must be considered from each rule in $R_f$. This section explores such extensions in case of an AC symbol having idempotency, nilpotency, cancelativity and being an Abelian group. In each case, it is shown that the joinability of additional critical pairs suffices to ensure generation of a canonical rewrite system from ground equations on an AC symbol with additional semantic properties.

\subsection{Idempotency}

If an AC symbol $f$ is also idempotent, implying $\forall x, f(x, x) = x$, the above algorithm {\bf SingleACCompletion} is extended with minor modifications. In the presence of idempotency, multiple occurrences of the same constant as arguments to an idempotent AC symbol can be normalized to a single occurrence, implying that the arguments to an AC symbol can be represented as sets, instead of multisets. 
For any rule $f(M) \rew f(N)$ where $f$ is idempotent and $M, N$ do not have duplicates, for every constant $a \in M$, a superposition 
$f(M \cup \mset{a})$ is generated,
leading to a new critical pair 
$(f(N \cup \mset{a}), f(M))$ 
and check its joinability. This is in addition to critical pairs generated from pairs of distinct rules used in {\bf SingleACCompletion.}

For an example, from $f(a, b) \rew c$ with the idempotent $f$, 
the superpositions are $f(a, a, b)$ and $f(a, b, b)$, leading to the critical pairs: $(f(a, c), f(a, b))$ and $(f(b, c), f(a, b))$, respectively, which further reduce to $(f(a, c), c)$ and $(f(b, c), c)$, respectively.

With the addition of critical pairs generated from each rule in a rewrite system, the local confluence check becomes as follows:

\noindent
\begin{lem}
An AC rewrite system $R_{S}$ with an idempotent AC symbol $f$ (with $f(x, x) = x$) in which
no rule in $R_S$ has $f$-monomials with duplicates, is {\em locally confluent} iff (i) the
critical pair: $(f((AB-A_1) \cup A_2), f((AB-B_1)
\cup B_2))$ between every pair of distinct rules $f(A_1) \rew f(A_2),
f(B_1) \rew f(B_2)$  is joinable, where
$AB = (A_1 \cup B_1) - (A_1 \cap B_1)$, and (ii) for every rule $f(M) \rew f(N) \in R_S$ and for every constant $a \in M$, the critical pair,
$(f(M), f(N \cup \mset{a})$,
is joinable.
\end{lem}

\begin{proof}
\,Consider a flat term $f(C)$, possibly with duplicates, rewritten in 
two different ways in one step using not necessarily distinct rules and/or $f(x, x) \rew x.$
There are three cases: (i) $f(C)$ is rewritten in two different ways in one step using $f(x,x) \rew x$ to 
$f(C -\mset{a})$ and $f(C - \mset{b})$ 
with $a \neq b$: The idempotent rule can be applied again on both sides giving 
$f(C-\mset{a, b})$ since 
$C - \mset{a}$ includes 
$\mset{b,b}$
and $C - \mset{b}$ includes 
$\mset{a,a}$.

(ii) $f(C)$ is rewritten in two different ways, with one step using $f(x, x) \rew x$ and another using $f(M) \rew f(N)$: An application of the idempotent rule implies that $C$ includes a constant $a$, say, at least twice; the result of one step rewriting is: 
$(f(C - \mset{a}), f((C-M) \cup N)$.
This implies there exists a multiset $A$ such that 
$C = A \cup M \cup \mset{ a}$.
The critical pair generated from $f(M) \rew f(N)$ is 
$(f(M), f(N \cup \mset{a}))$. 
Irrespective of whether
$a \in M$ or not, 
$C - \mset{a}$ 
is a supermultiset of $M$; to show joinability of 
$(f(M), f(N \cup \mset{a})$
apply all rewrites to 
$f(C - \mset{a})$ 
as well as to 
$f((C-M) \cup N)$.\footnote{Hence the need for this extra critical pair due to idempotency.}

\ignore{The rewrite steps used to show the joinability of $(f(M), f(N \cup \{\{a\}\})$ apply also on
$(f(C - \{\{a\}\}), f((C-M) \cup N)$,

showing joinability.}
The proof of the third case is the same as that of Lemma 3.3 and is omitted.
\end{proof}

The completion algorithm {\bf SingleACCompletion} extends for an idempotent AC symbol by considering critical pairs for every rule in $R_f$. Its termination proof is similar based on Dickson's lemma. Theorems \ref{correctness} and \ref{unique} also extend to the idempotent case with similar related proofs.

{\bf Example 2:} Revisiting Example 1 from the previous section with the additional assumption that $*$ is idempotent, the equations simplify to:
$a * b = a, a * b = b.$ When oriented into terminating rewrite rules, $a * b \rew a, ~a * b\rew b$ generating a superposition $a * b$ with the associated critical pair: $(a, b)$ leading to a new rule: $a \rew b,$ assuming $a > b.$ This rule simplifies all the other two rules, leading to a canonical rewrite system $\{ a \rew b\}.$

\ignore{Consider an $f$-monomial
$f(M)$ where $f$ is idempotent. If $M$ does not have a duplicate element, say $a$, then the idempotent does not apply. Otherwise, $f(M)$ can be reduced by $f(x, x)$ in one case, and in another case by some other rule,  say $f(A) \rew f(B) \in R_f$, where $A$ does not have any duplicate element, generating the $(f(M-\{\{a\}\}), f((M-A) \cup B)$. Their joinability follows from the joinability of the critical pair $(f(B), f(A-\{a\}))$ generated from
$f(A) \rew f(B)$ since
there is $M'$ such that
$M = M' \cup A$.}

\ignore{

To capture these semantic properties in the multiset representation of flattened AC terms, $f(M)$ is further normalized as: (i) in case of nilpotency, identical constants in $M$ cancel each other, since $f(x, x ) = e$ for any term $x$ and replaced by the identity $e$ of $f$, (ii) in case of idempotency, i.e., $f(x, x) = x$, multisets reduce to sets, and (iii) in case of unit (identity), i.e., $f(x, e) = x$, the identity element is deleted from the multiset. Such $f$-monomials are called normalized with respect to these axioms.
The definition of rewriting remains the same. All equations and rules are assumed to be normalized using these universal axioms. Due to these additional properties, new superpositions may have to be considered to ensure local confluence. }

\ignore{
Case 1 $f$ is idempotent, then for
any rule $f(a, b) \rew M$, $f(a, b) = f(a, a, b) = f(a, b, b)$
two new critical pairs 
$(f(a, c), c), (f(b, c),  c)$ must be checked for joinability.
With these additional critical pairs, the local-confluence proof of Lemma 3.3 can be appropriately adapted to work for an idempotent AC rewrite system; the termination proof of the modified algorithm extends.}

\subsection{Nilpotency}

Consider a nilpotent AC symbol $f$ with the property that $\forall x, f(x, x) = e$ where $e$ is a constant, typically standing for the identity of $f$ satisfying the property that $\forall x, f(x, e) = x$. Below we do not assume that $f(x, e) = x$; that discussion is postponed to a later subsection where
an AC symbol $f$ has identity as well as is idempotent.

For every
rule $f(M) \rew f(N) \in R$, generate for every $a \in M$, an additional critical pair
$(f(N \cup \mset{a}), f((M-\mset{a}) \cup \mset{e})).$ 
With this additional check for joinability, local confluence can be proved as shown below. 

\noindent
\begin{lem}
An AC rewrite system $R_{S}$ with a nilpotent AC symbol $f$ (with $f(x, x) = e$), such that no monomial in a rule in $R_S$ has duplicates,
is {\em locally confluent} iff (i) the
critical pair: $(f((AB-A_1) \cup A_2), f((AB-B_1)
\cup B_2))$ between every pair of distinct rules $f(A_1) \rew f(A_2),
f(B_1) \rew f(B_2)$  is joinable, where
$AB = (A_1 \cup B_1) - (A_1 \cap B_1)$, and (ii) for every rule $f(M) \rew f(N) \in R_S$ and for every constant $a \in M$, the critical pair,
$(f(N \cup \mset{a}), f(M-\mset{a} \cup \mset{e}))$ is joinable.
\end{lem}

\begin{proof}
\,Consider a flat term $f(C)$, possibly with duplicates, rewritten in 
two different ways in one step using not necessarily distinct rules and/or $f(x, x) \rew e.$
As in the case of idempotency, there are three cases: (i) $f(C)$ is rewritten in two different ways in one step using $f(x,x) \rew e$ to 

$f((C -\mset{a, a}) \cup \mset{e})$ and $f(C - \mset{b,b} \cup \mset{e})$, 
$a \neq b$. After single step rewrites, the idempotent rule can be applied again on both sides giving 
$f(C-\mset{a,a, b,b}\cup \mset{e, e})$.

(ii) $f(C)$ is rewritten in two different ways, with one step using $f(x, x) \rew e$ and another using $f(M) \rew f(N)$: An application of the nilpotency rule implies that $C$ includes a constant $a$, say, at least twice; the critical pair is: 
$(f((C - \mset{a, a}) \cup \mset{e}), f((C-M) \cup N))$.

Consider two cases: $a \notin M$, implying that
$C$ still has $\mset{a, a}$ 
as a subset, so nilpotency can be applied to get: 
$f(((C - M) - \mset{a, a}) \cup \mset{e} \cup N)$ from $f((C-M) \cup N)$; $f((C - \mset{a, a}) \cup \mset{e})$ can be rewritten using $f(M) \rew f(N)$ to get the same result.

The second case of $a \in M$ implies
there is a joinable critical pair:
$(f((M-\mset{a}) \cup \mset{e}), f(N)).$
Rewrite steps that make the above critical pair joinable, apply also on 
$(f((C - \mset{a, a})\cup \mset{e}), f((C-M) \cup N))$, since $(M-\mset{a}) \cup \mset{e}$ is a multisubset of $(C - \mset{a, a}) \cup \mset{e}$.

The third case is the same as that of Lemma 3.3 and is omitted.
\end{proof}

{\bf Example 3:} Revisiting Example 1 from the previous section and assuming $*$ to be nilpotent, the above ground equations simplify to $e * b = e, ~ e * a = e$. The superposition from the terminating rules corresponding to the above equations is: $a * b * e$ leading to the critical pair: $(a * e, b * e)$ whose two terms have the same normal form $e$. The rewrite system $\{ e * a \rew e,~ e * b \rew e\}$ is reduced canonical.

\subsection{Identity}

If $f$ has identity, say $e$ satisfying $f(x, e) = x$, no additional critical pair is needed since from every rule $f(M) \rew f(N)$,
$(f(N \cup \mset{e}), f(M)) $ 
are trivially joinable.

\noindent
\begin{lem}
An AC rewrite system $R_{S}$ in which an AC symbol $f$ has an identity $e$ (with $f(x, e) = x$) such that every $f$-monomial in $R_S$ is already normalized using $f(x, e) \rew x,$ is {\em locally confluent} iff the
critical pair: $(f((AB-A_1) \cup A_2), f((AB-B_1)
\cup B_2))$ between every pair of distinct rules $f(A_1) \rew f(A_2),
f(B_1) \rew f(B_2)$  is joinable, where
$AB = (A_1 \cup B_1) - (A_1 \cap B_1)$.
\end{lem}

\begin{proof}
\,Consider a flat term $f(C)$, possibly with duplicates and the identity element $e$, rewritten in 
two different ways in one step using not necessarily distinct rules and/or $f(x, e) \rew x.$
There are three cases: (i) $f(C)$ is rewritten in two different ways in one step using $f(x,e) \rew x$ to 
$f(C -\mset{e})$ and $f(C - \mset{e})$, 
which is trivially joinable.

(ii) $f(C)$ is rewritten in two different ways, with one step using $f(x, e) \rew x$ and another using $f(M) \rew f(N)$: An application of the identity rule implies that $C$ includes the identity $e$; the result of one step rewriting is:
$(f(C - \mset{e}), f((C-M) \cup N))$
The first element in the critical pair can still be rewritten using $f(M) \rew f(N)$, giving 
$f(((C-M) - \mset{e}) \cup N)$; 
since $M$ cannot include $e$, $C-M$ still has $e$ because of which $f(x, e) \rew x$, giving 
$f((C-M) - \mset{e}) \cup N)$.
 
The third case is the same as that of Lemma 3.3 and is omitted.
\end{proof}
\subsection{Idempotency and Identity}

Consider a combination of idempotency and identity, namely a rewrite system in which an AC symbol $f$ has identity $e$ as well as is idempotent (with $\forall x, ~ f(x, x) \rew x, f(x, e) \rew x).$ In this case, it suffices to consider the additional critical pair due to idempotency since there is no additional critical pair to consider due to the identity.

\noindent
\begin{lem}
An AC rewrite system $R_{S}$ with an idempotent AC symbol $f$ that also has identity $e$,
is {\em locally confluent} iff (i) the
critical pair: $(f((AB-A_1) \cup A_2), f((AB-B_1)
\cup B_2))$ between every pair of distinct rules $f(A_1) \rew f(A_2),
f(B_1) \rew f(B_2)$  is joinable, where
$AB = (A_1 \cup B_1) - (A_1 \cap B_1)$, and (ii) 
for every rule $f(M) \rew f(N) \in R_S$ and for every constant $a \in M$, the critical pair,
$(f(M), f(N \cup \mset{a})$,
is joinable.
\end{lem}

\begin{proof}
It is assumed that all monomials in $R_S$ are already normalized using $f(x, x) \rew x, f(x, e) \rew x$.
The subpart of the proof when idempotency rule is used along with rules in $R_S$, or the identity rule is used along with rules in $R_S$ is the same as in the proofs above for the idempotency and identity cases. The only case left is when idempotency and identity rules are both used.

\,Consider a flat term $f(C)$, possibly with duplicates and the identity element, rewritten in 
two different ways in one step using  $f(x, x) \rew x$ in one case and $f(x, e) \rew x$ in the other case, giving the critical pair: 
$(f(C-\mset{a}), f(C-\mset{e}))$
where wlog, $a$ appears in duplicate in $C$ which has also $e$. There are two cases: $a \neq e$, in which case, apply $f(x, e) \rew x$ on the first component of the critical pair and $f(x, x) \rew x$ on the second component with the result in both cases to be 
$f(C - \mset{a, e})$. In case $a = e,$ 
then
if $C$ has at least three $e$'s, then the above case applies; if $C$ has only two $e$'s, then the critical pair
is 
$(f(C -\mset{e}), f(C-\mset{e}))$, 
which is trivially joinable.
\end{proof}
\subsection{Nilpotency and Identity}

Consider a combination of nilpotency and identity: $\forall x, ~ f(x, x) \rew e, f(x, e) \rew x.$ In this case, it suffices to consider the additional critical pair due to nilpotency since there is no additional critical pair to consider due to the identity.

\noindent
\begin{lem}
An AC rewrite system $R_{S}$ with $f(x, x) = e, f(x, e) = x$, such that
every rule in $R_S$ is normalized using $f(x, x) \rew e, f(x, e) \rew x$
is {\em locally confluent} iff (i) the
critical pair: $(f((AB-A_1) \cup A_2), f((AB-B_1)
\cup B_2))$ between every pair of distinct rules $f(A_1) \rew f(A_2),
f(B_1) \rew f(B_2)$  is joinable, where
$AB = (A_1 \cup B_1) - (A_1 \cap B_1)$, and (ii) for every rule $f(M) \rew f(N) \in R_S$ and for every constant $a \in M$, the critical pair,
$(f(N \cup \mset{a}), f((M-\mset{a}) \cup 
\mset{e}))$, 
is joinable.
\end{lem}

\begin{proof}
The subpart of the proof when nilpotency rule is used along with rules in $R_S$, or the identity rule is used along with rules in $R_S$ is the same as in the proofs above for the  nilpotency and identity cases. The only case left is when nilpotency and identity rules are both used.

\,Consider a flat term $f(C)$, possibly with duplicates and the identity element, rewritten in 
two different ways in one step using  $f(x, x) \rew e$ in one case and $f(x, e) \rew x$ in the other case, giving the critical pair: 
$(f((C-\mset{a, a}) \cup \mset{e}), f(C-\mset{ e}))$, 
where wlog, $a$ appears in duplicate in $C$ which has also $e$. There are two cases: $a \neq e$, in which case, apply $f(x, e) \rew x$ on the first component of the critical pair and $f(x, x) \rew e$ on the second component with the result in both cases to be 
$f(C - \mset{a,a, e})$. 
In case $a = e,$ then
if $C$ has at least three $e$'s, then the above case applies; if $C$ has only two $e$'s, then the critical pair
is 
$(f(C -\mset{e}), f(C-\mset{e}))$, 
which is trivially joinable.
\end{proof}

Algorithm {\bf SingleACCompletion} can be appropriately modified to include the joinability checks for additional critical pairs.
The termination proof from the previous subsection extends to each of these cases and their combination. Theorems \ref{correctness} and \ref{unique} generalize for each of the above cases as well.

\ignore{since $f(e, b) = f(a, a, b)$
and $f(e, a) = f(a, b, b)$,
two critical pairs $(f(a, c), f(e, b)), (f(b, c), f(e, a))$ must also be checked for joinability.
If $e$ is also the identity of $f$, i.e., $f(x, e) = x$, then the critical pairs normalize to $(f(a, c), b), (f(b, c), a)$. The local confluence proof of Lemma 3.3 can be adapted; the termination proof is the same.

In case of the identity, however, no additional critical pair needs to be considered. The local confluence proof of Lemma 3.3 and termination proof work the same way as for the pure AC case. In the presence of combination of these axioms, suitable combination of critical pairs are used; the proofs extend as well.}

\section{An AC symbol with cancelation}

An AC function symbols $f$ has the {\bf cancelation} property, equivalently is called {\bf cancelative} iff $\forall x, y, z, f(x, y) = f(x, z) \implies y = z.$ If $f$ does not have a constant, say $e$ serving as an identity of $f$ (i.e., $\forall x, f(x, e) = x$), then $y, z$ must be a nonempty multiset of constants for flat terms with $f$. In the presence of identity for $f$, $y$ or $z$ could be the empty subset, representing the identity $e$.

Congruence closure with cancelative function symbols can be subtle because it enables generating smaller equalities from larger equalities: from a cancelative AC $f$, for example, $f(a, b) = f(a, c) \implies b = c$. More interestingly, 
if a cancelative $f$ does not have an identity, then the cancelative congruence closure of $f(a, c) = a$ also includes
$f(b, c) = b$ for any constant $b \in C$; this is so because from $f(a, c) = a$, it follows that $\forall b \in C, f(a, b, c) = f(a, b)$, giving $f(b, c) = b$ by cancelation. This is a way to cancel $a$ by including a common subterm on both sides so that the result after cancelation is not meaningless. 

Given a finite set $S$ of ground equations with a cancelative AC $f$, its congruence closure $CACCC(S)$ is defined by the closure of the $ACCC(S)$ with the above universal cancelative axiom, i.e., 
if $f(A) = f(B) \in CACCC(S),$ then for any common nonempty multisubset $C$ of $A$ and $B$ and assuming neither $A-C$ nor $B-C$ become empty causing $f(A-C)$ or $f(B-C)$ to be meaningless, $f(A-C) = f(B-C) \in CACCC(S)$. It also follows that for
any nonempty multiset $D$, given that $f(A \cup D) = f(B \cup D) \in CACCC(S)$, every equation obtained after canceling every possible nonempty common multisubset $G$ such that neither $f((A \cup D) - G)$
nor $f((B \cup D)-G)$ are the empty set, 
$f((A \cup D) - G) = f((B \cup D)-G) \in CACCC(S)$. Of course, if $f$ has the identity $e$, then any of the above multisets can be the empty set representing $e$.

For example, the AC congruence closure of $SC_1 = \{ f(a, a, b) = f(a, b, b)\}$ includes all pairs of $f$-monomials that can be generated by adding $ f(a, a, b, X) = f(a, b, b, X)$ where $X$ is any multiset of constants; 
assuming the degree ordering on $f$-monomials extending
$a \grt b$, 
its canonical rewrite system is $R = \{ f(a, a, b) \rew f(a, b, b) \}$. If $f$ is assumed to be also cancelative, the cancelative congruence closure then includes $a = b$ along with the above congruences, since
$f(a, f(a, b)) = f(b, f(a, b)) \implies a = b$ by cancelation;
from $a = b$, one gets many other congruences including $f(a, a) =f(a, b), f(a, b) = f(b, b), f(a, \cdots, a) = f(b,\cdots, b)$ with equal numbers of $a$ and $b$. 
A canonical rewrite system associated with the cancelative congruence closure
$CACCC(SC_1)$ is $R_C= \{ a \rew b \}$ using which each of the above pairs can be shown to be in the cancelative congruence closure; the canonical forms of the two sides in each pair above are identical. 
 
Given a finite set $S = \{ f(A_i) = f(B_i) \mid 1 \le i \le k\}$ of ground equations, it is easy to see that if $f(A'_i) = f(B'_i)$ is generated from $f(A_i) = f(B_i)$ after cancelation of common subterms on both sides, the cancelative congruence closure of
$S$ is contained in that of $S'$ in which $f(A'_i) = f(B'_i) $ replaces $f(A_i) = f(B_i)$. That motivates simplifying ground equations by canceling out common subterms from both sides.

\ignore{To generate a canonical rewrite system 
$R_C$ from a finite set $S$ of $f$-monomial equations in a cancelative AC $f$,
new superpositions and associated critical pairs are proposed on an AC canonical rewrite system $R$ representing its AC congruence closure $ACCC(S)$ after it is cancelativity closed as defined below.
$R_C$ satisfy the properties that
$CACCC(R_c) = CACCC(R) = CACCC(S) = ACCC(R_c)$.}

\subsection{Cancelativity Closed:}
As a first step toward generating a cancelative congruence closure, equations are made {\bf cancelatively closed}.
An equation
$f(A) = f(B)$ is {\bf cancelatively closed} for a cancelative AC $f$ iff (i) if $f$ has an identity $e$, then $A \cap B$ is the empty multiset, implying there is no subterm common among the two sides; 
(ii) if $f$ does not have an identity,
then $A \cap B$ is a singleton multiset with exactly one of $A$ and $B$ being a constant (i.e., singleton multiset); this captures the property that a constant can be common on both sides to avoid the meaningless term $f(\mset{})$.
For examples, $f(a, b) = f(c, d)$ is cancelativity closed irrespective of whether $f$ has an identity or not;
$f(a, b) = a$ is not cancelativity closed if $f$ has an identity, in which case $a$ can be canceled from both sides of the equation resulting in $b = e$; however, $f(a, b) = a$ is cancelativity closed in the absence of $f$ having an identity, but
$f(a, b, b) = f(a, b)$ is not cancelatively closed since either $a$ or $b$ can be canceled without making the resulting equation meaningless.

A set of equations is cancelatively closed iff every equation in the set is cancelatively closed. Similarly, 
a rewrite rule $f(A) \rew f(B)$ is cancelatively closed if the corresponding equation $f(A) = f(B)$ is cancelativity closed; a set of rules is cancelatively closed iff every rule in the set is cancelatively closed.

\ignore{the cancelative closure of a single equation can lead to multiple equations; as an example, given $f(a, b, c) = f(a, b)$, two cancelativity closed equations are deduced: $f(a, c) = a$  and $f(b, c) = b.$ Further, for any constant $d \in C$, from
$f(a, c) = a$ as well as $f(b, c) = b$, congruences $f(c, d) = d$ enabling cancelation of $a$ and $b$, respectively, are generated;}

\ignore{Below, no assumption is made about whether a cancelative $f$ has an identity; the approach extends easily in the presence of an identity.}

\ignore{An equation $f(A) = f(B)$ is {\bf cancelativity closed} iff the multisets $A$ and $B$ have at most one common constant. Thus, $f(a, b, c) = f(a, b)$, for example, is not cancelativity closed, but $f(a, c) = a$ as well as $f(b, c) = b$ are cancelativity closed.}

\subsection{Cancelative Closure of $S$:}
Given a system $S$ of ground equations with a cancelative $AC$ symbol $f$, an equivalent system $S'$ is generated from $S$ that is cancelativity closed such that $CACCC(S) = CACCC(S')$ as follows: for every ground equation 
$f(A) = f(B) \in S$,
(i) if it is cancelatively closed, then it is included in $S'$; (ii) if $f$ has an identity, then in $f(A-B) = f(B-A)$ is included in $S'$ with one of the sides possibly being the identity; (iii) otherwise, if $f$ does not have an identity and (a) neither $B-A$ nor $A-B$ are the empty sets, then $f(A-B) = f(B-A)$ is included in $S'$, (b) if $A-B$ or $B-A$ is the empty set, then for every constant $c \in C$, a new equation $f((A-B) \cup \{c \}) = 
f((B-A) \cup \{c \})$ is included in $S'$.
Clearly, $CACCC(S) \subseteq CACCC(S')$; using cancelativity, it also follows that $CACCC(S') \subseteq CASCC(S)$. Let $CancelClose(S)$ be the set of cancelatively closed equations
generated from $S$. 

For example, $CancelClose(\{f(a,a,a) = f(b, b)\}) = \{f(a, a, a) = f(b, b)\}$,
$CancelClose(\{f(a, b, b) = f(a, b, a)\}) = \{b = a\}$, and
$CancelClose(\{f(a, a, a, b, b) = f(a, b, b, b) \} = \{f(a, a) = b\}$
irrespective
of whether $f$ has the identity or not.
If $f$ has the identity $e$, then $CancelClose(\{f(a, b) = a\}) = \{ b = e\}$, but in the absence of $f$ having an identity, $CancelClose(\{f(a, b) = a\}) = \{ f(b, c) = c \mid c \in C\}$ which will include $f(a, b) = a$, $f(b, b) = b$ and for any other constant $c \in C$, $f(b, c) = c$.

Running {\bf SingleACCompletion} on a finite set of cancelatively closed ground equations as well as keeping new ground equations derived during the completion
as cancelatively closed
(such as cancelativitely closing rules 
before adding to the current basis in steps 2 and 3(ii))
is not sufficient to generate a canonical rewrite system for a cancelative congruence closure 
This is illustrated by the following example: Let $SC_2 = \{ f(a, a, a) = f(b, b), f(b, b, b) = f(a, a).\}$. They are cancelatively closed since the two sides of each equation do not have any common subterms.
Orienting the above equations from left to right, generates a canonical rewrite system which is also cancelativity closed; however, the rewrite system does not represent a cancelative congruence closure of $SC_2$ since $f(a, a, b) = a$ as well as $f(a, b, b) = b$ are in the cancelative congruence closure (since $f(a, a, a, b) = f(b, b, b) = f(a, a)$ as well as $f(b, b, b, a) = f(a, a, a) = f(b, b)$)  but they have distinct normal forms with respect to the rules generated from $SC_2$ by orienting the equations in $SC_2$ from left to right.
The rules obtained from orienting the equations in $SC_2$ from left to right are non-overlapping, so ${\bf SingleACCompletion}$ does not generate any new rule.

\begin{thm}
{\bf SingleACCompletion} with steps 2 and 3(ii) modified to keep rules cancelativity closed does not generate a canonical rewrite system for an AC cancelative $f$ symbol.
\end{thm}

In a typical proof of local-confluence of rewrite rules, any term that includes both left sides of two distinct rules, are joinable when rewritten using these rules in two different ways. However, for a cancelative $f$, this step can generate a nontrivial cancelative superposition as follows:
given two distinct rules $f(A_1) \rew f(B_1)$ and $f(A_2) \rew f(B_2)$, 
$f(A_1 \cup A_2)$ can be disjointly rewritten using the above rules and shown to be joinable; but 
$f(A_1 \cup A_2) = f(B_1 \cup B_2)$ may not be cancelatively closed, possibly giving rise to new critical pairs $(f(A'_1 \cup A'_2), f(B'_1 \cup B'_2))$ obtained from the cancelative closure of $f(A_1 \cup A_2) = f(B_1 \cup A_2)$. For the local confluence of $\{ f(A_1) \rew f(B_1), f(A_2) \rew f(B_2) \}$, the joinability of  $(f(A'_1 \cup A'_2), f(B'_1 \cup B'_2))$ must be checked. 

For example, from $SC_2$ above, joining the two equations gives a congruent pair $(f(a, a, a, b, b, b), f(a, a, b, b))$ that is not cancelatively closed. The pair is trivially joinable as the rules corresponding to $SC_2$
can be disjointly applied; however, a new congruence pair is generated because of cancelation.
From
$f(a, a, a, b, b, b) = f(a, a, b, b)$, new congruences, 
$f(a, b, b) = b$ and $f(a, a, b) = a$
are generated because of cancelation and in the absence of $f$ having an identity; in the presence of the identity $e$ for $f$, only one congruence pair $f(a,b) = e$ is generated.
Consequently, orienting the two equations in $SC_2$ from left to right does not result in a canonical rewrite system for its cancelative congruence closure.
It will be shown that with the cancelative closure of the disjoint superposition which was trivial in the absence of cancelativity, must be accounted for in the local confluence check of a rewrite system for the cancelative congruence closure. For $SC_2$, $\{ f(a, a, a) \rew f(b, b), f(b, b, b) \rew f(a, a), f(a, a, b) \rew a, f(a, b, b) \rew b \}$ is a canonical rewrite system assuming $C = \{a, b \}.$ In case $f$ has the identity $e$, then the third and fourth rules can be simplified to be $f(a, b) \rew e.$

\ignore{
An inference rule such as cancelativity derives additional equational inferences of the form:
from $s_i \eqv s_{i+1} \eqv s_{i+2}$, generate
$t_i \eqv t_{i+2}$ by applying cancelativity on $s_i \eqv s_{i+2}.$ Each step in the above inference is obtained by enlarging a cancelatively closed equation, when combined together, they need not be. There is interaction also between the left side of a cancelativity-closed rewrite rule with the right side of the other rule, which must be accounted for during the computation of a cancelative congruence closure.}

\ignore{As the above example illustrates, 
$f(b, b, b) = f(a, a)$ can be used to enlarge 
$f(a, a, a) = f(b, b)$ to  $f(a, a, a, b) = f(b, b, b)$, giving $f(a, a, a, b) = f(a, a)$; due to cancelativity, it becomes $f(a, a, b) = a$; similarly, $f(a, b, b) = b$ is also generated.\footnote{This behavior appears to be different from the classical congruence closure in which it suffices to consider subterms appearing in a set of ground equations to generate their congruence closure. In contrast, by considering larger superterms of subterms (a superterm $f(b, b, b)$ of $f(b, b)$ appearing in the input ground equations), it is possible to generate additional congruences.} Whereas a canonical rewrite system for $ACCC(S)$ is $\{ f(a, a, a) \rew f(b, b), f(b, b, b) \rew f(a, a)\}$ using degree ordering on $f$-terms, a canonical rewrite system for cancelatively closed
$CACCC(SC_2)$ is $\{ f(a,a,a) \rew f(b, b), f(b, b, b) \rew f(a, a), f(a, b, b) \rew b, f(a, a, b) \rew b \}.$ }

\ignore{The congruence $f(a, b, b) \eqv b$ is derived by cancelativity from $f(a, b, b, b) \eqv f(a, a, a) \eqv f(b, b)$ in $ACCC(S).$\footnote{In case of an extensional rule as in \cite{BK20}, $f(x, y) = f(u, v) \implies (x = u \land y = v),$ $f(a, b) \eqv d \eqv f(a', b'),$ there is a superposition on the right sides of the two rules, generating a right-right critical pair. Given two rules $L_1 \rew R_1, L_2, \rew R_2,$ such that a nonvariable subterm $R_1$ unifies with $R_2$ using a mgu $\sigma,$ $\sigma(L_1) \eqv \sigma(R_1) \eqv \sigma(R_1)[p \rightarrow \sigma(L_2)]$ on which the inference rule may apply.}}

\ignore{Step 2 in $SingleACCompletion$ algorithm is changed as: (i) Step 2: Normal forms $\hat{l},\hat{r}$ lead to a derived equation $\hat{l} = \hat{r}$ is cancelatively closed to become $\hat{\hat{l}} = \hat{\hat{r}}$ which is oriented into a rewrite rule.}

\ignore{Let us consider some more illustrations exhibiting the subtlties of the issues involved in generating a cancelative congruence closure of ground equations on a single AC symbol. 
Unlike in most local confluence tests, if a term (superposition) in which the two distinct left sides of rules are disjoint subterms, critical pairs are joinable, in case of cancelativity, such a term can generate nontrivial critical pairs which need to be considered for joinability.}

\ignore{For illustration, consider another example: $SC_3 = \{ f(a, b) = f(c, d), f(a, c) = f(b, d')\}$. The term $f(a, b, a, c)$ from the two left sides of the equations has the left sides $f(a, b)$ as well as $f(a, c) $ as disjoint subterms; rewriting using $f(a, b) \rew f(c, d)$ gives $f(a, c, c, d)$ whereas rewriting using $f(a, c) \rew f(b, e)$ gives $f(a, b, b, e)$.
Applying cancelativity on the critical pair generates $f(b, b, e) = f(c, c, d)$ which is in the cancelative AC congruence closure $CACCC(S)$, and hence must be checked for joinability for generating a canonical rewrite system representing $CACCC(S)$. Without applying cancelativity, the pair is joinable by applying the rewrite rules on disjoint parts as is typically done in the local confluence proofs. This additional check is still not sufficient for checking local confluence of a cancelative AC rewrite system as illustrated in the case of the first example: the superposition is $f(a, a, a, b, b, b)$ with the critical pair is $(f(b, b, b, b, b), f(a, a, a, a, a))$; since the two terms do not have any common subpart, the cancelativity rule does not apply. 
}

\ignore{
{\bf Example 4:} For illustration, consider $SC_3 = \{ f(a, b) = f(c, d), f(a, c) = f(b, d')\}$. The term $f(a, b, a, c)$ generated from the two left sides of the equations 
equals $f(c, d, b, d'$ which after cancelation, generates $f(a, a) = f(d, d')$; however, the term $f(a, b, b, d')$ generated from the left side of the first equation and the right side of the second equation equals $f(c, d, a, c)$, which after cancelation, generates $f(b, b, d') = f(c, c, d)$ which follows from the above two equations.}

\ignore{
To contrast with the classical superpositions/overlaps generated from the left sides of the rewrite rules, a new superposition, called a {\bf left-right} superposition, is introduced to facilitate cancelation, possibly due to common subterms appearing on the different sides of different rules such as a possible overlap between the left side of a rule with the right side of another rule leading to cancelation. In case of $SC_2$, the right side of rule 1 overlaps with the left side of rule 2, and similarly, the right side of rule 2 overlaps with the left side of rule 1. 
}

\ignore{A universal conditional inference rule such as cancelativity\footnote{(or extensionality--see \cite{BK20}); see Section for related comments in the paragraph ending Section4 about "large" terms vs "small" terms.} enlarges congruence closure by adding adding inferences. From a sequence of inferences in an AC congruence closure, additional inferences are deduced; in order to include such inferences,  additional rewrite rules need to be added into a canonical rewrite system for $ACCC(S)$. This is discussed below.}

\ignore{
\subsection{Left-Right Superposition}
Given two distinct rules $f(A_1) \rew f(B_1)$ and $f(A_2) \rew f(B_2)$ with an AC $f$, a superposition is generated as before when $A_1 \cap A_2 \neq \emptyset$. In case of a cancelative $f$, a special superposition is also generated if $A_1 \cap B_2 \neq \emptyset$ or $B_1 \cap A_2 \neq \emptyset$ since there is an opportunity for common subterm to be generated and subsequently canceled out due to possible overlaps of the right side of a rule with the left side of another rule. 

Let $C$ be the nonempty common part, if any, of $A_1$ and $B_2$; similarly, let $D$ be the nonempty common part of $B_1$ and $A_2$, if any. The critical pair generated from the left-right superposition from the above rules is then $(f(A'_1 \cup A'_2), f(B'_1 \cup B'_2))$, where $A_1 = A'_1 \cup C, A_2 = A'_2 \cup D, B_1 = B'_1 \cup D, B_2 = B'_2 \cup C$, provided both $A'_1 \cup A'_2$
$B'_1 \cup B'_2$ are nonempty\footnote{If $f$ has an identity $e$, then this condition is not necessary}. Such a critical pair is called a 
{\bf left-right}, to contrast with a traditional {\bf left-left} critical pairs. $f(A_1 \cup A_2) = f(A'_1 \cup A'_2 \cup C \cup D)$ can be viewed as the associated superposition, 
which rewrites using rules 1 and 2 to $f(B'_1 \cup D \cup B'_2 \cup C)$. In case 
$A'_1 \cup A'_2$ or 
$B'_1 \cup B'_2$ is the empty set, then subsets of $C$ and $D$ are so selected that
neither $A'_1 \cup A'_2$ nor
$B'_1 \cup B'_2$ is the empty set; these selections of $C$ and $D$ must be done in every possible way; an example below illustrates this subtlety.
It is easy to see that if both $C$ and $D$ above are the empty set, implying that $A'_1 = A_1, A'_2 = A_2, B'_1 = B_1, B'_2 = B_2$
the left-right critical pair is trivial since both rules can be disjointly applied to show joinability. Even if exactly one of $C$ or $D$ is nonempty, even then the the left-right critical pair is trivial since $(f(A'_1 \cup A'_2), f(B'_1 \cup B'_2))$ is joinable. 

\ignore{
Define ${AB}_1 = (A_1 \cup B_2) - (A_1 \cap B_2)$, as the {\bf lcm} of $A_1, B_2$; then ${AB}_1 = A_1 \cup B'_2 = A'_1 \cup B_2$, corresponding to the common part, denoted by $C$, of $B_2$ and $A_1$ removed. Similarly, ${AB}_2 = (A_2 \cup B_1) - (A_2 \cap B_1)$ as the lcm of $A_2, B_1$ with ${AB}_2 = A_2 \cup B'_1 = A'_2 \cup B_1$, corresponding to the common part $D$ of $A_2$ and $B_1$. Since $f(A_1 \cup A_2) = f(B_1 \cup B_2) = f(A'_1 \cup A'_2 \cup C \cup D) = f(B'_1 \cup B'_2 \cup C \cup D),$ by the cancelation inference rule, $f(A'_1 \cup A'_2) = f(B'_1 \cup B'_2)$.}

\ignore{
Applying $f(A_1) \rew f(B_1)$ on
$A_1 \cup B'_2$ gives $f(B_1 \cup B'_2)$

using $f(A_2) \rew f(B_2)$ backward, one obtains $f(A'_1 \cup A_2) = f(A'_1 \cup A'_2 \cup D),$ where $D$ is the common part ({\bf gcd}) of $A_2$ and $B_1$. Since $A'_1 \cup B_2 = A_1 \cup B'_2$, $f(A_1 \cup B'_2)$ rewrites using $f(A_1) \rew f(B_1)$ to
$f(B_1 \cup B'_2) = f(B'_1 \cup B'_2 \cup D)$ with $D$ being common in
$f(A'_1 \cup A'_2 \cup D) \equiv f(B'_1 \cup B'_2 \cup D)$ which cancels out, generating $(f(A'_1 \cup A'2), f(B'_1 \cup B'_2))$ as the critical pair; a similar analysis on the common part $C$ generates the same critical pair. Such a critical pair is called {\bf left-right} and the corresponding superposition is {\bf left-right (lr)}, to contrast with the traditional {\bf left-left(ll)} critical pairs. 
}
In case of the rewrite rules corresponding to equations in $SC_2$ when oriented from left to right: $\{1. ~f(a, a, a) \rew f(b, b), ~2.~f(b, b, b) \rew f(a, a)\}$, 
$A_1 = \mset{a, a, a} ,  B_2 = \mset{a, a}, C = \mset{a, a}, A_2 = \mset{b, b, b}, B_1 = \mset{b, b}, D = \mset{b, b}$, in which case
$A'_1 \cup A'_2 = \mset{a, b}$ but $B'_1 \cup B'_2 = \mset{} = \emptyset$. If $C$ is chosen to be $\mset{a}$, then 
$A'_1 \cup A'_2 = \mset{a, a, b}$ and 
$B'_1 \cup B'_2 = \mset{a}$; similarly if
$D$ is chosen to be $\mset{b}$, then
$A'_1 \cup A'_2 = \mset{a, b, b}$ and 
$B'_1 \cup B'_2 = \mset{b}$. 
The resulting two new critical pairs, $f(a, a, b) = a$ and $f(a, b, b) = b$ which are not joinable by rules 1 and 2.
To generate a canonical rewrite system for $CACCC(SC_2)$, rules corresponding to their normal forms must be added, leading to
$R_{SC_2} = \{ 1. ~f(a, a, a) \rew f(b, b),~ 2.~f(b, b, b) \rew f(a, a), ~3.~f(a, a, b) \rew a, ~4.~f(a, b, b) \rew b \}$. It can be checked that both $l-l$ as well as $l-r$ superpositions among
distinct pairs of the above four rules lead to joinable critical pairs. In particular, consider the l-r superposition between rules 1 and 4: the right side of rule 1 has $\mset{b,b}$ in common with the left side of rule 4,i.e., $C = \mset{b,b}$; but there is nothing in common between the left side of rules 1 with the right side of rule 4, i.e., $D = \mset{}$;
the critical pair is $(f(a, a, a, a), b)$ which is joinable by rules 1 and 4, similar to the disjoint case when $C = D = \mset{}$.
\ignore{the above rules lead to joinable ll as well as lr critical pairs.}
$R_{SC_2}$ is the reduced canonical rewrite system representing the cancelative AC congruence closure, $CACCC(SC_2)$, of $SC_2$.

{\bf Example 4:} Consider $SC_3 = \{ f(a, b) = f(c, d), f(a, c) = f(b, d')\}$ which are cancelatively closed equations.
Orienting them from left to right using the total ordering $a \grt b \grt c \grt d \grt d',$  the left-right critical pair is $(f(a, a), f(d, d'))$ since $b$ is the common part of the left hand side of the first rule and the right hand side of the second rule, and $c$ is the common part of the right hand side of the first rule and the left hand side of the second rule. 
The other rule $f(b, b, d') \rew f(c, c, d)$ can be obtained from the left-left superposition of the two rules. The rewrite system is $\{ f(a, b) \rew f(c, d), f(a, c) \rew f(b, d'), f(b, b, d') \rew f(c, c, d), f(a, a) \rew f(d, d') \}$ is canonical and represents $CACCC(SC_3)$.
}

\ignore{Their normal forms give two new rules $\{ f(a, a, b) \rew a, ~ f(a, b, b) \rew b\}$ using which
the canonical rewrite system
$\{f(a, a, a) \rew f(b, b), ~f(b, b, b) \rew f(a, a), f(a, a, b) \rew a, ~ f(a, b, b) \rew b\}$ represents $CACCC(SC_2)$.}

\ignore{Superposition between the two rules after orienting the above equations from left to right, gives another rule: $f(b, b, e) \rew f(c, c, d).$ Critical pairs between the new rule and the first two rules are joinable. The three rule rewrite system is cancelatively closed. However, $f(a, a) = f(d, e)$ is congruent in the cancelative congruence closure of $S$ even though both terms are in normal form with respect to the three rewrite rules since $f(a, b, e) = f(c, d, e) = f(a, a, c)$ by first equation, which simplifies by cancelativity to $f(d, e) = f(a, a)$, and similarly, $f(a, c, d) = f(a, a, b) = f(b, d, e)$, which also simplifies by cancelativity to $f(a, a) = f(d, e).$ The inference sequence $f(a, b, e) \eqv f(c, d, e)\eqv f(a, a, c)$ is in $ACCC(R)$ but $f(a, a) \eqv f(d, e)$ is not in $ACCC(R)$; however it is in $CACCC(R).$

For the above example,
$b$ appears in the left hand side of the first rule as well as the right side of the second rule, a superposition $f(a, b, e)$ is generated giving a critical pair:
$f(d, e) = f(a, a)$ which is cancelatively closed. The resulting system
$\{ f(a, b) \rew f(c, d), f(a, c) \rew f(b, e), f(a, a) \rew f(d, e), f(b, b, e) \rew f(c, c, d) \}$ is canonical as well as cancelatively closed. Below, an algorithm is given, which uses new additional superpositions between a pair of rules sharing constants on the left and right sides along with left-left critical pairs.}

\subsection{Cancelative Disjoint Superposition:} Given two distinct cancelativity closed rewrite rules $f(A_1) \rew f(B_1), f(A_2) \rew f(B_2)$ such that the pair $(f(A_1 \cup A_2), f(B_1 \cup B_2)) $ has common subterms, a critical pair $(f((A_1 \cup A_2) - (B_1 \cup B_2)), f(( B_1 \cup B_2) - (A_1 \cup A_2))$ is generated; if $f$ does not have an identity, then the terms in the critical pair cannot be $f(\mset{})$ and multiple critical pairs are generated by ensuring that neither term in a critical pair 
is $f(\mset{})$. Further, the cancelative closure of these critical pairs must be generated.
This is illustrated above for $SC_2$.
In case of the example $SC_3$, a cancelative disjoint superposition is generated since $(f(a,b, a, c), f(c,d, b, d')) $ have the common subterm 
$f(b, c)$ generating the critical pair
$(f(a, a), f(d, d'))$.

It is easy to see that new equations generated from cancelative disjoint superposition are in the cancelatively closed congruence closure.

\ignore{

\begin{thm}
Let $R$ be a canonical AC rewrite system $R = \{ f(A_i) \rew f(B_i) | 1 \le i \le k\}$ with an AC symbol $f$.
If $f$ is cancelative, 
$R$ is cancelativity closed and 
the critical pairs generated from disjoint superpositions of distinct pairs of rules in $R$ are joinable, then $R$ is also a canonical rewrite system for its cancelative congruence closure $CACCC(S)$, where $S$ is the finite set of equations corresponding to rewrite rules in $R$.
\end{thm}}

\begin{thm}
Let $R$ be a canonical AC rewrite system $R = \{ f(A_i) \rew f(B_i) | 1 \le i \le k\}$ with an AC symbol $f$.
If $f$ is cancelative, 
$R$ is cancelativity closed and 
the critical pairs generated from cancelative disjoint superpositions of distinct pairs of rules in $R$ are also joinable, then $R$ is a canonical rewrite system for its cancelative congruence closure $CACCC(S)$, where $S$ is the finite set of equations corresponding to rewrite rules in $R$.
\end{thm}

\ignore{
It is easy to see that new equations generated from left-left as well as left-right critical pairs are in the cancelatively closed congruence closure. Addition rules obtained from
left-right critical pairs that are not joinable, are added so that the resulting rewrite system represents the associated cancelative congruence closure without needing cancelation in rewrite proofs.}


In a typical $AC$ congruence relation, a sequence of derivations is of the form:
$f(C_1) \eqv f(C_2) \eqv \cdots \eqv f(C_i) \eqv f(C_{i+1})$, where $f(C_j) \eqv f(C_{j+1})$ using a rule $f(A) \rew f(B)$ is such that either $A \subseteq C_j$ in which case $C_{j+1} = (C_j-A) \cup B$, or
$B \subseteq C_j$ in which case $C_{j+1} = (C_j-B) \cup A$. In a cancelative congruence relation, however, from two congruent terms in the above sequence of derivations, say $f(C_{j_1})$ and $f(C_{j_2})$, by cancelation, 
$f(C'_{j_1}) \eqv f(C'_{j_2})$ can be derived by canceling a common subterm $f(X)$ (i.e., $C_{j_1} = C'_{j_1} \cup X, C_{j_2} = C'_{j_2} \cup X$), from which another sequence of derivations is generated. Further, 
a cancellation free derivation of congruence of $f(C_{j_1})$ and $f(C_{j_2})$ is not a single step by some equation $f(L) = f(R) \in S$, but rather involves multiple steps by multiple equations in $S$. Without any loss of generality, a cancelation-free derivation can be rearranged so that cancelation is on two terms congruent by two different equations, say  $f(L_1) = f(R_1), ~f(L_2) = f(R_2) \in S$. Cancelative disjoint superposition and the associated critical pairs capture such interaction among equations. 

To illustrate, consider a proof of congruence of $f(a, a, b) \eqv^* a$ in $SC_2$: $f(a, a, a, b, b, b) \eqv f(b, b, b, b, b) \eqv f(a, a, b, b) $ involving two rules
$f(a, a, a) = f(b, b),~ f(b, b, b) = f(a, a)$,
from which by cancelation, $f(a, a, b) \eqv a$.

\ignore{Since the final step in this proof is cancelation, unlike in my previous case, to generate the above proof, terms need to be enlarged, guided by the equations in $SC_2$.}

In a proof below, it is shown that for any sequence of derivations in $CACCC(S)$ possibly using cancelation, a new sequence of derivations in $CACCC(R) $ can be constructed without any cancelation using additional rules generated from cancelative disjoint critical pairs.

\ignore{The following discussion is to provide some intuition about a cancelative AC congruence closure and how to obtain the associated canonical rewrite system before getting into a proof of the above theorem.}

\ignore{
As stated above, a cancelative congruence closure $CACCC(R)$ of $R$ when its rules are viewed as equations, has the property that if $f(A) = f(B) \in CACCC(R) $ and $A \cap B \neq \emptyset,$ then $f(A-B) = f(B-A) \in CACCC(R)$ also.}

\ignore{Consider the rewrite system obtained from
$SC_2$ by orienting the equations from left to right: $\{1. ~f(a, a, a) \rew f(b, b), ~2. ~ f(b, b, b) \rew f(a, a)\}$. Not assuming cancelativity, it is easy to see that the rewrite system is canonical. Consider a sequence of inferences: $f(a, a, a, b) \eqv f(b, b, b) \eqv f(a, a)$ in $ACCC(SC_2)$. If
$f$ is also cancelative, $f(a, a, b) \eqv a$ can be derived in $CACCC(SC_s)$, however $f(a, a, b) \eqv a$ is not in $ACCC(SC_2)$.}
To check whether $f(A) = f(B) \in CACCC(R),$ both $A$ and $B$ may have to be enlarged to $A', B'$ respectively using some multiset $C$ of constants such that $A' = A \cup C, B' = B \cup C,$ and $f(A'), f(B')$ have the same canonical form using $R$, i.e., $f(A') \eqv f(B')$ in $ACCC(R)$. Determining how much to enlarge $f(A), f(B)$ can however be a challenge; left-right superpositions address this challenge. 
\ignore{For the above example,
a left-right superposition is $f(a, a, a, b) = f(b, b, b) = f(a, a)$ as well as $f(b, b, b, a) = f(a, a, a) = f(b, b)$ on which using the cancelativity of $f$, $f(a, a, b) = a$ and $f(b,b, a) = b$, respectively, are derived to be in $CACCC(SC_2)$. Thus from a sequence of derivations in $ACCC(SC_2)$, a new sequence of derivations in $CACCC(SC_2)$ is generated using the cancelativity of $f$. }
\ignore{To illustrate, from $f(c, d, e) \eqv f(a, b, d') \eqv f(a, a, c) \in ACCC(R)$ above, a new inference $f(d, d') \eqv f(a, a)$ which is not in $ACCC(R)$ along with $f(c, d) \eqv f(a, b), f(b, d') \eqv f(a, c),$ which are in $ACCC(R).$}

\ignore{In general, from a sequence of inferences in $ACCC(R)$ without using cancelativity, new sequences of inferences are successively derived by using cancelativity in $CACCC(R)$. To check membership in $CACCC(R)$ for which there is a derivation using cancelativity, one approach is to enlarge terms in the derivation by additional constants to get a derivation without cancelativity. 
Consider a sequence of inferences due to equations in $R$: $f(A_0) = f(A_1) = \cdots = f(A_k)$; from the above sequence, it is possible to 
derive a new sequence by enlarging $A_i$'s by a multiset $X_i$ of constants, which after cancelation, is the original sequence. Repeating this process of enlarging multiple times, a new sequence of derivations can be generated without using cancelativity.}

\ignore{In the above example, consider a proof of $f(a, a, a) = f(a, d, d')$; they have the same canonical form due to the rule $f(a, a) \rew f(d, d')$ generated from $left-right$ superposition of rules 1 and 2, and thus $(f(a,a,a), f(a, d, d')) \in CACCC(R).$
In the absence of that rule and using the canonical rewrite system for $ACCC(R)$, $f(a, a, a)$ and $f(a, d, d')$ are in normal form.
A proof of $f(a, a, a) = f(a, d, d')$ is generated by enlarging both sides by adding $c$ to $f(a, a, a, c) = f(a,d,c,d')$ on which rule 2 applies to give $f(a, a, b, d') = f(a, a, d, d')$ on which rule 1 applies. Thus a rewrite proof of $f(a, a, b, d') = f(a, a, d, d')$ in $R$ without the additional rule $f(a, a) \rew f(d, d')$ is generated; on this, the cancelativity rule can apply, giving a proof of $f(a, a) = f(d, d')$. This is the intuition behind the proof below.}

\begin{proof}
{\it Sketch:}
Consider a sequence of inferences showing the congruence of $f(C) \eqv^* f(D) \in CACCC(S)$. In general, 
because of cancelation, this sequence cannot be $f(C) = f(C_1) \eqv f(C_2) \cdots f(C_i) \eqv f(C_{i+1}) = f(D)$, where $f(C_j) = f(L \cup X), f(C_{j+1}) = f(R \cup X)$ for any $1 \le j < i+1$, or
$f(C_j) = f(R \cup X), f(C_{j+1}) = f(L \cup X)$ for some $f(L) = f(R) \in S$. But instead it is broken into a chain of such subsequences where subsequences are connected using the cancelation inference rule.

For illustration, for $SC_3$, a possible inference sequence of $f(a, a, b) \eqv a$ 
is a chain of subsequence $f(a, a, a, b) \eqv f(b, b, b) \eqv f(a, a)$
in which the first and second inferences are due to the first equation and second equation, respectively, but then there is another subsequence $f(a, a, b) \eqv a,$
connected to the previous subsequence by cancelation
on two congruent terms $f(a, a, a, b) \eqv f(a, a)$, leading to a sequence of derivations:
$f(a, a, a, b) \eqv f(b, b, b) \eqv f(a, a)$, $f(a, a, b) \eqv a.$ Using the rule $f(a, a, b) \rew a$ generated from cancelative disjoint superposition from  rules 1 and 2, a cancelative-free derivation of $f(a,a, b) \eqv a$ can be generated.

\ignore{

$s_1 \eqv s'_1, s_2 \eqv s'_2 \cdots s'_i \eqv s_i,  s'_{i+1} \eqv s_{i+1} = t$ as $s'_j$ need not be the same as $s_{j+1}$; }

The following proof involves generating a cancelation-free sequence of derivations from an arbitrary sequence of derivations involving cancelation. When no cancelation is involved, then there is nothing new to construct. 

Assume the sequence of inferences relating $s \eqv^* t$ employs $k$ cancelation steps. Proof is by induction on $k$.

{\bf Basis k =1:} $s \eqv^* t \in CACCC(S)$ can be decomposed into two subsequences $s \eqv^*  s_1$ without cancelation, followed by $t_1 \eqv t_2 \eqv^* t $ without cancelation. In other words,  $s \eqv^* s_1\in ACCC(S)=ACCC(R)$ and $t_1 \eqv t_2 \eqv^* t \in ACCC(S)=ACCC(R)$ where $t_1 \eqv t_2$ was generated by applying cancelation on two congruent terms $s_{j_1} \eqv^* s_{j_2}$ in
a cancelation-free sequence of inferences $s \eqv^* s_1$. Without any loss of generality this sequence can be arranged so that the two congruent terms are at the end of the sequence as: $s_{j_1} \eqv s_{j_3} \eqv s{j_2}$ that includes a cancelative disjoint superposition. By assumption, the critical pairs corresponding to them are joinable, implying that there is a rewrite proof
of the pair $t_1 \eqv t_2$ generated after cancelation. Since both cancelation-free sequences have rewrite proofs, they can be glued together using a rewrite proof of $t_1 \eqv t_2$, resulting in a cancelation-free derivation.

\ignore{
There is a rewrite proof of $s' \eqv s_1$ using $R$ and further $s' = f(A \cup X), t' = f(s_1 \cup X) $

A new cancelation free derivation is constructed using rules generated from $l-r$ superposition similar to a derivation for $SC_3$ above: from $f(a, a, a, b) = f(b, b, b) = f(a, a)$, infer
by cancelation from $f(a, a, a, b) = f(a, a)$: $f(a, a, b) = a$;
wlog, the new rule added from $l-r$ superposition of rules 1 and 2 gives the new rule using a proof without cancelation is found by applying the rewrite rule $f(a, a, b) \rew a$. }

\ignore{Consider $(f(A), f(B)) \in CACCC(R)$. If $(f(A), f(B)) \in ACCC(R)$, i.e., cancelativity is not needed to show the congruence, then their  normal forms using $R$ are equivalent (modulo AC), implying the corresponding multisets of constants in the normal forms are the same. 

Proof by contradiction. Assume 
$(f(A), f(B)) \in CACCC(R)$ but $(f(A), f(B)) \notin ACCC(R)$ implying that cancelativity is applied to show the congruence of $f(A)$ and $f(B)$. Let $k$ be the number of applications of the cancelativity inference rule applied to show that $(f(A), f(B)) \in CACCC(R).$ 

The proof is by induction on $k$.

The base case of $k = 1$ implies that there is only one application of the cancelativity inference rule; without any loss of generality, assume that $f(A)$ and $f(B)$ are both in normal forms with respect to $R$.
$f(A) = f(B)$ is then derived from 
$f(A') = f(A \cup X), f(B') = f(B \cup X)$ for some  $f(X)$ using cancelativity, where $X$ is an nonempty multi-set of constants such that 
$(f(A'),f(B')) \in ACCC(R)$. In the above example, to derive  $f(a, a, b) = a \in CACCC(SC_2)$, enlarge both sides to $f(a, a, a, b) = f(a, a)$ which is in $ACCC(SC_2)$, since the normal form of $f(a, a, a, b)$ using the first rule followed by the second rule is indeed $f(a,a)$.

Assume $f(A') \rew^* f(B')$ (or vice versa) or $f(A'), f(B')$ both rewrite to a common term $f(C').$ 
Since $f(A)$ and $f(B)$ are not congruent using $R$, the enlargement of $f(A), f(B)$ by some $X$ could be due to the applications of rules in $R$ in the rewrite steps either from $f(A')$ to $f(B')$ or both from
$f(A'), f(B')$ to $f(C')$.

Consider the first case: 
$f(A') \rew f(A'') \rew^* f(B')$ by rewrite rules $f(L_1) \rew f(R_1)$ and $f(L_2) \rew f(R_2),$ respectively. In this case, one way to get a nontrivial $X$ for enlargement is due to the left-right superposition between $f(L_1)\rew f(R_1)$
and $f(L_2) \rew f(R_2)$) 
$R_1 = R'_1 \cup X, L_2 = L'_2 \cup X,$ and similarly $L_1 = L'_1 \cup X', R_2 = R'_2 \cup X',$ generating $f(L'_1 \cup X' \cup L'_2 \cup Y) = f(L_1 \cup L'_2 \cup Y) \rew f(R_1 \cup L'_2 \cup Y) = f(R'_1  \cup L'_2 \cup X \cup Y) = f(R'_1 \cup R_2 \cup Y) = f(R'_1 \cup R'_2 \cup X' \cup Y)$. Since the critical pair $(f(L'_1 \cup L'_2), f(R'_1 \cup R'_2))$ is joinable using $R$ and in the congruence closure of $R$, 
$f(L'_1 \cup X' \cup L'2 \cup Y)$ and $f(R'_1 \cup R'_2 \cup X' \cup Y)$ are joinable using $R$ and are in the congruence closure of $R$.

In the second case, 
$f(A') \rew f(C')$ and $f(B') \rew f(C')$ by
$f(L_1) \rew f(R_1)$ and $f(L_2) \rew f(R_2),$ respectively; their joinability follows from the joinability of the left-left critical pair. The case of where the second step is applied backward: $f(C') \rew f(B')$  is similar to the first case since $f(A') \rew f(C') \rew f(B').$

For the induction step, assume the joinability of an inference sequence in which $k$ applications of cancelative inference rules are applied, can be shown; the joinability of the step corresponding to the $(k+1)^{th}$ application of cancelativity is shown in the same way by case analysis as in the base case: $k=1$.
}

\ignore{
From an inference sequence establishing $f(A) = f(B)$, a new inference sequence can be constructed using $f(X)$ without the application of the cancelativity rule, implying that $(f(A \cup X), f(B \cup X)) \in ACCC(R)$ and hence $f(A \cup X), f(B \cup X)$ have equivalent canonical forms.

at be from $f(C) = f(D)$, infer $f(C-T) = f(D-T)$, where $T$ is a common nonempty multi-subset of constants in $C$ and $D$. 

there is a finite rewrite sequence of $f(A) = f(A_0) \rew f(A_1) \rew \cdots \rew f(A_{i_1+1}) \rew f(A_{i_1 + 2})$ using $R$ such that $f(A_{i_1}) \rew f(A_{i_1 + 1})$ by a rule $f(L) \rew f(R)$ and
$f(A_{i_1 + 1}) \rew f(A_{i_1 + 2})$ by a rule 
$f(L') \rew f(R')$ which by cancelativity, gives $f(A'_{i_1}) = f(A'_{i_1 + 2})$ after removing some common multi-subset from both sides. The common multi-subset is throughout the sequence

but there is a proof $f(A)=f(A_0) \eqv f(A_1) \eqv \cdots \eqv f(A_k) = f(B)$ where $A_{i+1} = (A_i - L \cup R)$ or
$A_{i+1} = A_i - R \cup L$ for some rule $f(L) \rew f(R) $ (thus implying that $(f(A), f(B)) \in ACCC(R)$ also) or there exists another inference sequence
$(f(A'), f(B')) \in ACCC(R)$ implying that 
$f(A') = f(A'_0) \eqv f(A'_1) \eqv \cdots \eqv f(A'_{k'}) = f(B'),$ such that $f(A_i) = f(A'_{i'}-A'_{j'}), f(A_{i+1}) = f(A'_{j'} - f(A'_{i'})).$ It is this second case which necessitates superposing the left hand side of a rule with the right side of another rule when they have a common part; $f(A'_{i'}) \eqv^* f(A'_{j'}) in ACCC(R)$.

}

{\bf Induction Step:} The induction hypothesis is that any derivation with $k$ cancelations can be converted into a cancelation free derivation. To prove that a derivation with $k+1$ cancelations can also be converted into a cancelation free derivation, isolate the last cancelation step by breaking the original derivation with $k+1$ cancelations into $k+1$ subsequences. Since by the induction hypothesis, for the first $k$ subsequences, there is a cancelation free derivation, that derivation can serve as a single cancelation-free subsequence and is linked to the $k+1$ subsequence built using the $(k+1)^{th}$ cancelation step. This case is the same as the $k=1$ case. Using the same argument as for the $k=1$ case, a cancelation free derivation is generated for these two new subsequences, which by induction hypothesis, gives a cancelation-free derivation for the original derivation with cancelation.
\end{proof}

\hspace{-3mm}
{\bf CancelativeACCompletion($S = S_f \cup S_C$, $\gg_f$):}
\vspace*{1mm}
\begin{enumerate}

\item Orient constant equations in $S_C$ into terminating rewrite rules $R_C$ using $\gg_C$ and interreduce them. Equivalently, using a union-find data structure, for every constant $c \in C$, compute, from $S_C$, the equivalence class $[c]$ of constants containing $c$ and make $R_C = \cup_{c \in C} \{ c \rew \hat{c} ~| ~c \neq \hat{c} ~\mbox{and}~ \hat{c} ~\mbox{is the least element in}~ [c]\}.$

\item $S_f = \bigcup_{f(A)-f(B) \in S_f} CancelClose(\{f(A) - f(B)\})$.
\item If $S_f$ has any constant equalities, remove them from $S_f$ and update $R_C$ using them. 

\item Initialize $R_f$ to be $R_C$. 
\item Let $T := S_f$.
\item Pick an $f$-monomial equation $l = r \in T$ using some selection criterion (typically an equation of the smallest size) and remove it
from $T$.

Compute normal forms $\hat{l}, \hat{r}$ using $R_f$. If equal, then discard the equation.

Otherwise, for each $l' = r'\in CancelClose(\{\hat{l} = \hat{r}\})$, 

\begin{enumerate}
    \item orient the result into a terminating rewrite rule using $\gg_f$; abusing the notation and without any loss of generality, each is oriented from left to right.

\ignore{\item For every rule
$\hat{l_i} \rew \hat{r_i},$
generate critical pairs between $\hat{l_i} \rew \hat{r_i}$ and every $f$-rule in $R_f$, adding them to $T$ (similar to Step 3 in {\bf SingleACCompletion}).\footnote{Critical pair generation can be done incrementally using some selection criterion for pairs of rules to be considered next, instead of generating all critical pairs of 
the new rule with all rules in $R_f$.}}
    \item Generate both (classical) critical pair (similar to Step 3 in {\bf SingleACCompletion})
and disjoint superposition critical pair between $l' \rew r'$
and every rule in $R_f$, adding them to $T$.\footnote{Critical pair generation can be done incrementally using some selection criterion for pairs of rules to be considered next, instead of generating all critical pairs of 
the new rule with all rules in $R_f$.}

\item  Add the new $l' \rew r' \in $ into $R_f$.

\item Interreduce other rules in $R_f$ using $l' \rew r'$  (as in Step 4 in {\bf SingleACCompletion}).
\end{enumerate}

\ignore{
For every rule $l \rew r$ in $R_f$ such that
$l$ can be reduced by $\hat{l_i} \rew \hat{r_i}$, remove $l \rew r$ from $R_f$ and insert $l = r$ in $T$. If $l$ cannot be reduced by any rule or by $\hat{l_i} \rew \hat{r_i}$, but $r$ can be reduced, then reduce $r$ by the new rule and generate a normal form $r'$ of the result. Replace
$l \rew r$ in $R_f$ by
$l \rew r'$ after cancelatively closing it.}

\item Repeat Step 6 until both types of critical pairs among all pairs of rules in $R_f$ are joinable, and $T$ becomes empty.

\ignore{\item During the generation of $R_f$, if a rule with $c \rew m,$
is generated with a constant $c$, then introduce a new constant $u$, extend the constant ordering to include $c \grt u$ and replace $c \rew m$ by two rules $c \rew u$ to be included
in $R_C$ and $m \rew u$ to be included in $R_f$. This ensures that the rules in $R_f$ always have a nonconstant on their left hand sides.\footnote{This trick is found useful while combining multiple reduced canonical rewrite system the congruence closures of ground equations with different AC symbols.}}

\item Output $R_f$ as the canonical rewrite system associated with $CACCC(S)$.
\end{enumerate}

\ignore{
\begin{algorithm}[H]
  \DontPrintSemicolon
  \caption{$\bf cancelativeACCompletion(s, \gg_f)$}
  \label{CancelAC}
  \kwInput{$S = S_C \cup S_f,$ a finite set of constant equations and $f$-monomial
  equations; $\gg_f:\text{a monomial ordering on $f$-monomials extending $\gg_C$}$
       }
  \kwOutput{$R_f, \text{A reduced canonical rewrite system representing $CACCC(S_C \cup S_f)$}$}
  $T \gets \emptyset$\;
  For{each $f(A)  = f(B) \in S_f$}{
  $T \gets T \cup CancelClose(\{ f(A) = f(B)\})$\;
  }
  \For{each $ c = d \in T \mid c, d \in C$}{
  $S_C \gets  S_C \cup \{ c = d\}$\;
  $T \gets T - \{c = d \}$\;
  }
  $R_f \gets {\bf SingleACCompletion(T \cup S_C, \gg_f)}$\;
  
  \end{algorithm}
  
  }
  \ignore{

\item Initialize $R_f$ to be $R_C$. Let $T := S_f$.

\item Pick an $f$-monomial equation $l = r \in T$ using some selection criterion (typically an equation of the smallest size) and remove it
from $T$. cancelatively close the equation; since multiple equations can be generated if $f$ does not have the identity, abusing the notation, let each such equation be $l_i = r_i.$ Compute normal forms $\hat{l_i}, \hat{r_i}$ using $R_f$. If equal, then discard the equation, otherwise, cancelatively close the equation in normal forms and orient the result into a terminating rewrite rule using $\gg_f$. Without any loss of generality, let the rule be $\hat{l_i} \rew \hat{r_i}.$

\item For every rule
$\hat{l_i} \rew \hat{r_i},$
generate critical pairs between $\hat{l_i} \rew \hat{r_i}$ and every $f$-rule in $R_f$, adding them to $T$ (similar to Step 3 in {\bf SingleACCompletion}).\footnote{Critical pair generation can be done incrementally using some selection criterion for pairs of rules to be considered next, instead of generating all critical pairs of 
the new rule with all rules in $R_f$.}
Also, generate {\bf left-left} critical pairs among each distinct pair of rules $\hat{l_{i_1}} \rew \hat{r_{i_1}}$ and
$\hat{l_{i_2}} \rew \hat{r_{i_2}}$ $i_1 \neq i_2$, if needed.
Generate also {\bf left-right} critical pairs among all distinct pairs of rules as defined above, adding them to $T$. Critical pairs and their normal forms are cancelatively closed.

\item  Add each new rule $\hat{l_i} \rew \hat{r_i}$ into $R_f$.

\item Interreduce other rules in $R_f$ using the new rule (as in Step 4 in {\bf SingleACCompletion}).

For every rule $l \rew r$ in $R_f$ such that
$l$ can be reduced by $\hat{l_i} \rew \hat{r_i}$, remove $l \rew r$ from $R_f$ and insert $l = r$ in $T$. If $l$ cannot be reduced by any rule or by $\hat{l_i} \rew \hat{r_i}$, but $r$ can be reduced, then reduce $r$ by the new rule and generate a normal form $r'$ of the result. Replace
$l \rew r$ in $R_f$ by
$l \rew r'$ after cancelatively closing it.

\item Repeat the previous three steps until the critical pairs among all pairs of rules in $R_f$ are joinable, and $T$ becomes empty.

\ignore{\item During the generation of $R_f$, if a rule with $c \rew m,$
is generated with a constant $c$, then introduce a new constant $u$, extend the constant ordering to include $c \grt u$ and replace $c \rew m$ by two rules $c \rew u$ to be included
in $R_C$ and $m \rew u$ to be included in $R_f$. This ensures that the rules in $R_f$ always have a nonconstant on their left hand sides.\footnote{This trick is found useful while combining multiple reduced canonical rewrite system the congruence closures of ground equations with different AC symbols.}}

\item Output $R_f$ as the canonical rewrite system associated with $S$.
\end{enumerate}
}

\ignore{
\begin{algorithm}[H]
  \DontPrintSemicolon
  \caption{$\bf {update}_{AC}(R_f, R_C, \gg_f)$}
  \label{updateAC}
  \kwInput{$R_f:\text{An AC canonical rewriting system on $f$-monomials},    
       R_C:\text{A canonical rewriting system on constants}, \gg_f:\text{a monomial ordering on $f$-monomials extending a total ordering on constants}$
       }
       
  \kwOutput{\text{An updated AC canonical rewrite system}, \text{new implied constant rules}}
  $S \gets \set{l = r  \mid NF(l, R_C) \neq l, l \rightarrow r \in R_f}$\;
  $R_f \gets R_f - \{ l \rew r \mid l = r \in S\}$\;
  $R^{'}_f \gets R_C \cup R_f$\;
  $NR_C \gets \emptyset$\;
  \While{$S \neq \emptyset$}{
    $T \gets \emptyset$\;
    \For{each $l = r \in S$}{
      $S \gets S - \{l = r\}$\;
      \uIf{$NF(l, R^{'}_f) \neq NF(r, R^{'}_f)$}{
        Orient $(NF(l, R^{'}_f), NF(r, R^{'}_f))$ as $l^{'} \rightarrow r^{'} \text{using} \gg_f$\;
        \For{each $l^{''} \rew r^{''} \in R^{'}_f$}{
         Compute a critical pair $(s_1, s_2)$ between $l^{'} \rightarrow r^{'}$ and $l^{''} \rew r^{''}$\;
         \lIf{$NF(s_1, R^{'}_f) \neq NF(s_2, R^{'}_f)$}{$T \gets T \cup \{ 
         NF(s_1, R^{'}_f) = NF(s_2, R^{'}_f)\}$}
        }
    $R^{'}_f \gets R^{'}_f \cup \set{l^{'} \rightarrow r^{'}}$\;
    Interreduce other rules in $R^{'}_f$ using $l^{'} \rightarrow r^{'}$ 
        as in step 4 of ${\bf SingleACCompletion}$\;
        {\bf if} $l^{'} \rightarrow r^{'}$ is a new constant rule {\bf then } $NR_C \gets NR_C \cup \set{l^{'} \rightarrow r^{'}}\;$
    }
    }
    $S \gets T$\;
  }
  \Return $(R^{'}_f, NR_C)$\;
\end{algorithm}

}

Using proofs similar to those for the correctness and termination of {\bf SingleACCompletion} including using Dickson's lemma, the correctness and termination of the above algorithm can be established. Theorems \ref{correctness} and \ref{unique} generalize for this case also.

\begin{thm}
The algorithm {\bf CancelativeACCompletion}  terminates, i.e., in Step 4, rules to $R_f$ cannot be added infinitely often.
\end{thm}

{\bf Example 4:} Run the above algorithm on $SC_3 = \{1.~ f(a, b) = f(c, d), ~2.~f(a, c) = f(b, d')\}$, where $f$ is cancelative, using 
the degree ordering on $f$-terms
extending the constant ordering $a \grt b \grt c \grt d \grt d'$.
Both the rules are cancelatively closed. Critical pairs among them are generated: the classical critical pair construction from the two rules is
$3.~f(b, b, d') \rew f(c, c, d)$; cancelative disjointly superposition constructions gives:$~4.~ f(a, a) \rew f(d, d'$. The resulting reduced canonical rewrite system for the AC cancelative congruence closure of $SC_3$
is: $R = \{1.~ f(a, b) \rew f(c, d), ~2. ~f(a, c) \rew f(b, d'), ~3.~f(b, b, d') \rew f(c, c, d),~4.~ f(a, a) \rew f(d, d')\}$.

\subsection{An AC function symbol being a group operation}

If an AC function symbol $f$ has in addition to identity, a unique inverse for every element, then it belongs to an algebraic structure of  an Abelian group. The completion procedure in that case becomes much simpler in contrast to the {\bf SingleACCompletion} and {\bf cancelativeACCompletion} discussed in the previous section, both in its behavior and complexity. Completion reduces to Gaussian elimination which can also be formulated in matrix form using Smith normal form or Hermite normal form. 

Given a finite set $S$ of ground equations with an AC $f$ on constants satisfying the group axioms with the identity 0, its congruence closure $AGCC(S)$ is defined by the closure of $ACCC(S)$ with the universal group axioms, i.e., 
if $f(A) = f(B) \in AGCC(S),$ then
$f(-A) = f(-B) \in AGCC(S)$, where
$-A$ stands for the multiset of the inverses of elements in $A$, i.e.,
if $A = \mset{a_1, \cdots, a_k}$ in $f(A)$, then $-A = \mset{-a_1, \cdots, -a_k}$. Thus,
$f(A) = f(-(-A)),$
as well as 
$f(A \cup \mset{0}) = f(A) \in AGCC(S)$.

Assuming a total ordering on constants, because of the cancelation property and using properties of the inverse function, including
$-f(x, y) = f(-x, -y)$, a ground equation $f(A) = f(B)$ can be standardized into $f(A') = f(B')$ where
$A' \cap B' = \emptyset$ and furthermore, $A'$ consists of the positive occurrences of the largest constant in the ordering in $A' \cup B'$ whereas $B'$ has both negative and positive occurrences of constants smaller than the one in $A'.$ 

From a standardized equation $f(A') = f(B')$, a rewrite rule $f(A') \rew f(B')$ is generated since $A'$ already has the largest constant on both sides of the equation.

A rewrite rule $f(A_1) \rew f(B_1)$ rewrites a term $f(C)$ iff (i) $f(C)$ has a submonomial $f(D)$ with the constant of $A_1$ such that $A_1 \subseteq D$ and the result of rewriting is
the standardized form 
of $f((C-A_1) \cup B_1)$, (ii) $f(C)$ has a submonomial $f(D)$ that has negative occurrences of the constant in $f(A_1)$,i.e., $-A_1 \subseteq D$ and the result of rewriting is
the standardized form 
of $f((C \cup A_1) \cup (- B_1))$. 
For a simple example,
$f(-a-a-b)$ is rewritten using $f(a, a) \rew f(-b, c)$ to $f(-a-a-b+a + a + b - c)$ which standardizes to $f(-c) = -c.$
A term $f(C)$ is in normal form with respect to a rewrite rule $f(A_1) \rew f(B_1)$ iff neither $A_1$ nor $-A_1$ is a subset of $C$.

A rewrite rule $f(A_1) \rew f(B_1)$ rewrites
another rule 
$f(A_2) \rew f(B_2) $ with the same constant in $A_1, A_2$ iff $A_1 \subseteq A_2$, leading to $f(A_2-A_1) = f(B_2 - B_1)$.
 If $A_2 - A_1$ is empty, then the equation reduces to $0 = f(B_2 - B_1)$, which is then brought to standardized form by picking the largest constant in the result.

Let $+$ be the binary symbol of an Abelian group traditionally instead of $f$ with 0 as the identity. Let $k a, k > 0,$ stand for $a + \cdots + a,$ added $k$ times; similarly $k a, k < 0,$ stand for $-a + \cdots + -a$ added $|k|$ times, the absolute value of $k$.

{\bf Example 5:} For a simple illustration, if there are rewrite rules $4 a \rew B_1, 6 a \rew B_2, 10 a \rew B_3$, then using the extended gcd computation, $2 a \rew B_2 - B_1$ is generated along with other rules without $a$ on their left sides are generated from equations:
$2 B_2 = 3 B_1, 5 B_2 - 5B_1 = B_3.$

A rewrite rule $f(A_1) \rew f(B_1)$ can repeatedly reduce another rewrite rule
$f(A_2) \rew f(B_2)$ with the same constant on their left sides until its normal form.
Inter-reduction among such rewrite rules can be sped up by using the extended gcd of positive coefficients of the standardized rewrite rules with the same constant appearing in their left hand sides.
Particularly, given a set $\{ f(A_1) \rew f(B_1), \cdots, f(A_k) \rew f(B_k)\}$ with the identical constant $c$ in their left hand sides, with $A_i$ is the multiset with $j_i$ occurrences of  $c$, a new rule
$f(A) \rew f(B)$ is generated, where $A$ is the multiset of $c$'s with as many occurrences as the extended gcd $g$ of $j_1, \dots j_k$, i.e., 
$g = \sum_{l} i_l * j_l,$ and $f(B) = f(B'_1 \cup \cdots \cup B'_k)$ where $B'_l$ is the number of occurrences of $B_l$ multiplied by $i_l.$ In addition, other rules get simplified after eliminating $c$ from them using the newly generated rewrite rule which eliminates the constant from other rules by rewriting.

For the above example, the extended gcd of $4 a, 6 a, 10 a$ is $2 a$ with a Bezout identity
being $2 a = 6 a - 4a$. Using the rule $2 a \rew B_2 - B_1$, the rules $4 a \rew B_1, 6 a \rew B_2, 10 a \rew B_3$ are normalized eliminating $a$ from them.

The above process of interreduction is repeated thus generating a triangular form with at most one rewrite rule for each distinct constant, which also serves as a canonical rewrite system $R_S$ for $CC(S).$

\begin{thm}
Given a finite set $S$ of ground equations of terms in an Abelian group and a total
ordering on constants, a canonical rewrite system $R_S$ is generated from the rewrite rules from standardized equations generated from $S$ by rewriting and inter-reduction among rules.
\end{thm}

\begin{proof}
Let $S'$ be the standardized equations generated from $S$. Starting with a standardized equation with the largest constant symbol with the least positive integer coefficient. Using inter-reduction on two rules at a time (or using extended gcd computation using multiple rules), $R$ consisting of rewrite rules with at most one rule for every constant on its left side is generated using the above algorithm.
\end{proof}

{\bf Example 6:} As another example, consider $S = \{ a + a + b + c = -a + b + b -c, ~a + b = -a + c + 0, ~-b -b-c = a -b + c\};$  Under a total ordering $a \grt b \grt c,$ the above equations are standardized to 
$\{ 3 a = b - 2c,~ 2 a = -b + c,~ a = -b-2c \},$
which orient from left to right as rewrite rules. The third rule $a \rew -b -2c$ rewrites the other two rules to: 
$\{ -3 b -6c = b - 2c, ~-2b-4c=-b + c\},$
which after standardization, generate: 
$\{ 4 b \rew -4c,~ b \rew -5c\}.$
Using the rewrite rule for $b$, the canonical rewrite system 
$\{ a \rew -b - 2c, ~b \rew -5c, ~16 c \rew 0\}$
is generated, further fully reduced to,
$\{ a \rew 3c,~ b \rew -5c, ~16 c \rew 0\}$
which represents the associated $AGCC(S)$.

\section{Combining Canonical Rewrite Systems Sharing Constants}

In the forthcoming sections, canonical rewrite systems generated from ground equations expressed using various AC symbols and uninterpreted symbols are combined for the case when the rewrite systems only share common constants. A critical operation performed on an already generated canonical rewrite system is when a new equality on constants implied by another rewrite system, is added/propagated to it. This section explores such an update in depth. It is done for two distinct cases but with similar analysis: (i) a constant equality generated from an AC canonical rewrite system such as from Section 3 is combined with a canonical rewrite system generated from uninterpreted ground equations such as using \cite{KapurRTA97}, and (ii) 
a constant equality generated from a canonical rewrite system generated from uninterpreted ground equations or another AC canonical rewrite system is added to another AC canonical rewrite system. 
For simplicity, it is assumed that in the monomial termination orderings $\gg_f$  used in generating canonical rewrite systems,
nonconstant ground terms  are bigger than constants as is typically the case; this restriction will be relaxed later in Section 9 however to allow great flexibility in choosing termination orderings on ground terms.

\ignore{Given two ground canonical rewrite systems that share only constants, they can be combined to get a canonical rewrite system over the mixed signature in the general case when constants can be bigger in orderings than nonconstant ground terms. For orderings in which constants are lower than nonconstant ground terms and no constant equalities are propagated from the generation of one rewrite system to another, it suffices to take the union of the two rewrite systems because of disjointness of the left sides of the rewrite rules in the two disjoint rewrite systems. 

In case constant equalities are propagated from one rewrite system to another and/or constants could be bigger in an ordering than ground nonconstant terms, then orientation of rewrite rules may change or replacement of a constant by another constant in the left sides of rewrite rules may generate additional overlaps among rewrite rules. A possible restoration of the canonicity of a rewrite system when its confluence has been violated, may lead to additional iterations, but they are guaranteed to be finite many because of finitely many constants.

When constant equalities are propagated from one rewrite system to another, it becomes necessary to restore canonicity of the rewrite systems which may get violated since constants get replaced in the left sides of rules, or in case a new constant is introduced to delay/postpone decision about which among many possible normal forms of certain constants with different outermost symbols should be picked as their canonical form. Further, such restoration may have to be done multiple times. An example above illustrated that. Below, we discuss yet another example illustrating such behavior. 

}

Let $R_U$ be a reduced canonical rewrite system representing a congruence closure of ground equations on uninterpreted symbols assuming a total ordering $\grt$ on constants in $C$; sometimes, $\gg_C$ is used to stand for a total ordering on constants to be consistent with the notation for monomial orderings $\gg_f$.

Let $R'_U = {\bf update_U}(R_U, \{ a_i =  b_i \mid 1 \le i \le k \})$ be a reduced canonical rewrite system generated from $R_U$, when a finite set of constant equalities $\{ a_i = b_i \mid 1 \le i \le k \}$ are added to $R_U$; if needed, the total ordering on constant is extended to consider any new constants not in $R_U$. $R'_U$ is generated as follows.
(i) The constant congruence relation defined by $R_U$ is updated using
the constant rules $\{ a_i = b_i \mid 1 \le i \le k \}$ by merging congruence classes.
Using the total ordering $\gg_C$ on constants.
the least element in a merged congruence class is picked as its canonical form and a new canonical constant rewrite system is generated. 
Later this update operation on constant rules due to the addition of new constant rules is called ${\bf update_C} (R_C, \{ a_i =  b_i \mid 1 \le i \le k \})$, where $R_C$ is a subset of $R_U$ consisting only of constant rules.
(ii) Flat rules with uninterpreted symbols from $U$ are then updated by replacing constants appearing in them by their new canonical representatives, if any; all flat rules having identical left sides $h(c_1, \cdots, c_j) \rew d_1, 
h(c_1, \cdots, c_j) \rew d_2, \cdots,
h(c_1, \cdots, c_j) \rew d_l$ are replaced by a single rule $h(c_1, \cdots, c_j) \rew d_i$, where $d_i$
is the least constant among $\{d_1, \cdots, d_l\}$; additional constant rules from $ d_j \rew d_i, d_j \neq d_i$ are added. If during this second step, any new constant rule is generated, then the above steps are repeated until no new constant rules are generated. 
\ignore{ (i) if $a$ does not appear in the left side of any rules in $R_U$ then $R'_U = \{ a \rew b \} \cup \bigcup_{l \rew r \in R_U} \{ l \rew \bar{r} \}$ where $\bar{r}$ is a normal form of $r$ using other constant rules in $R_U$. (i) If $a$ appears in the left side of a flat nonconstant rule in $R_U$, then that flat rule
$h(c_1, \cdots, c_k) \rew d$ changes to reflect new canonical forms of constants to be $h(c'_1, \cdots, c'_k) \rew d'$. This change may lead to two flat nonconstant rules with identical left sides: $h(c'_1, \cdots, c'_k) \rew d$ as well as $h(c'_1, \cdots, c'_k) \rew d'$, generating a new constant equality $d = d'$. Compute
${\bf update_U}(R_U - \{ h(c_1, \cdots, c_k) \rew d\}
\cup \{ h(c'_1, \cdots, c'_k) \rew d', ~a \rew b \}, 
\{ d \rew d'\}$ assuming $d > d'$ (otherwise if $d' > d$, then replace $d$ by $d'$ and vice versa everywhere above). ${\bf Update_U}$ is repeatedly applied until all constant equalities have been propagated. The result of this repeated ${\bf update_U}$ step is a canonical rewrite system $R'_U$, consisting of constant rules and flat rules with distinct left sides, all constants on the right sides of the rules and the left sides of the flat rules in their canonical forms. Observe that the orientation of flat-rules due to ${\bf update_U}$ do not change even though their left sides may change due to possible changes in their canonical forms.}

Note that the ${\bf update_U}$ is implicitly used in Section 2.2 during the construction of a canonical rewrite system representing $CC(S).$ The correctness of ${\bf update_U}$ follows from the correctness of a congruence closure algorithm as given in \cite{KapurRTA97,BK20}.

\begin{thm}
$R'_U$ as defined above is a reduced canonical rewrite system for the congruence closure $CC(S_U \cup \{a_i = b_i \mid 1 \le i \le k\})$ using a total ordering $\gg_C$ on constants, where $S_U$ is the set of equations obtained from the rewrite rules in $R_U$.
\end{thm}
\ignore{generating an updated canonical rewrite system corresponding to the modified congruence closure,  reflecting the newly added $j$ constant equalities $a_i = b _i, 1 \le i \le j.$ Below, we use $R_f$ to stand for $R_C \cup R_U$ when there is no need to separate out a constant rewrite system from a nonconstant rewrite system, and use it as the first argument to $update$ and similarly, $R'_f$ to be $R'_C \cup R'_U.$

{\bf Case 1: Combining an AC canonical rewrite system with a canonical rewrite system of uninterpreted symbols when a constant equality from an AC rewrite system is added to a rewrite system for uninterpreted symbols.}
Given a canonical rewrite system representing a congruence closure on uninterpreted ground equations, as $R_C \cup \bigcup_{h \in F_U} R_h$ \cite{KapurRTA97},
revisit how a new constant equality $a = b$ from an AC canonical rewrite system updates it. Canonical forms of $a$ and $b$ using the union-find data structure are checked for equality, and if unequal, they are merged, selecting one of the canonical forms to be that of the merged equivalence classes. As a result, canonical forms of constant arguments in flat rules $h(c_1, \cdots, c_k) \rew d$ in $R_h$s change to reflect new canonical forms to be $h(c'_1, \cdots, c'_k) \rew d'$. If there are two nonconstant rules such that $h(c'_1, \cdots, c'_k) \rew d$ as well as $h(c'_1, \cdots, c'_k) \rew d'$, then a new constant equality $d = d'$ is generated as a result.

This propagation continues but eventually terminates since there are only finitely many constants and hence finitely many constant equalities can possibly be generated.
A reader can thus observe that adding a constant equality to a canonical rewrite system can in general result in recomputing superpositions and critical pairs to reestablish its canonicity. 

The operation {\bf update} captures this:
$(R'_C, R'_U) = update((R_C, R_U), \{ a_i = b_i | 1\le i \le j \})$, generating an updated canonical rewrite system corresponding to the modified congruence closure,  reflecting the newly added $j$ constant equalities $a_i = b _i, 1 \le i \le j.$ Below, we use $R_f$ to stand for $R_C \cup R_U$ when there is no need to separate out a constant rewrite system from a nonconstant rewrite system, and use it as the first argument to $update$ and similarly, $R'_f$ to be $R'_C \cup R'_U.$
}
Another ${\bf update_U}$ could have been defined in which constant rewrite rules, instead of constant equations, are added to $R_U$. 

The second case is ${\bf update_{AC}}$ when a finite set of constant rewrite rules are added to an AC canonical rewrite system $R_f$ consisting of rewrite rules on $f$-monomials. This case is more interesting since adding a constant equality not only changes the left side of the AC rewrite rules but due to this change, unlike in the uninterpreted case $\bf{update_U}$, the orientation of the updated rewrite rules can also change if the left side of an updated rule becomes smaller than its updated right side. This implies that new superpositions and critical pairs need to be considered; thus ${\bf SingleACCompletion}$ must be rerun. Heuristics and data structures can be developed so as to avoid repeating unnecessary work, especially if a pair of rules do not change. New constant rewrite rules generated during the update are kept track since they need to be propagated to other rewrite systems, if needed.

All rules in $R_f$ that get updated due to new constant rules are collected in the variable $S$. They are normalized, oriented and then their superpositions with other rules in $R_f$ as well as among themselves are recomputed (steps 7--18). 

Variable $T$ keeps track of critical pairs generated from superpositions as rules change; it also includes new equations generated from rewrite rules whose left sides are rewritten when a new rule is added; initially $T$ is $S$. 
Variable $NR_C$ is the set of any new constant rules generated which may have to be propagated to other rewrite systems in a combination, leading to changes in them. Otherwise, if $NR_C$ remains empty, that indicates that even though $R^{'}_f$ may be different from $R_f$, no new constant rewrite rules were generated; consequently, other rewrite systems involved in a combination are not affected.
$NF(l, R_C)$ computes a normal form of an $f$-monomial $l$ using the constant rules in $R_C$, and $NF(l, R_f)$ computes a normal form of $l$ using an AC rewrite system $R_f$. $NR_C$ stores any new constant rewrite rules generated.

Updates due to new constant rules are repeated until no new constant rules are generated.

Observe that interreduction in Step 16 does not introduce any new constant rule, since if the left side of a rule already in $R^{'}_f$ gets reduced by the addition of a new rule $l' \rew r'$, then that rule is moved from $R^{'}_f$ to $T$; if only the right side of a rule in $R^{'}_f$ gets reduced, then the original rule is not a constant rule.

The output of ${\bf update_{AC}}$ is an updated canonical AC rewrite system $R'_f$ generated from a canonical AC rewrite system $R_f$ due to constant rules in $R_C$ and the set $NR_C$ of new constant rewrite rules generated.

\ignore{

Let
$R'_f = {\bf update_{AC}}(R_f, \{ a \rew b \})$ be the canonical rewrite system generated from $R_f$ using the constant rule $a \rew b$ as follows: Case (i):
If $R_f$ consists of AC rules $f(c_1, \dots, c_k) \rew f(d_1, \cdots, d_l)$, where $k \ge 2$, where $a$ does not appear in any of the left sides, then 
$R'_f = \{ a \rew b \} \cup \bigcup_{ f(c_1, \dots, c_k) \rew f(d_1, \cdots, d_l) \in R_f} \{  f(c_1, \dots, c_k) \rew \overline{f(d_1, \cdots, d_l)} \}$ where $\overline{f(d_1, \cdots, d_l)}$ is a normal form of 
$f(d_1, \cdots, d_l)$ using other constant rules.

Case (ii): When $a$ also appears in the left sides of rules in $R_f$,
equations 
$\overline{f(c_1, \dots, c_k)}= \overline{f(d_1, \cdots, d_l)}$ are generated from such rules and 
may also result in their reorientation. In either case, if the left sides of the rules change, additional superpositions and critical pairs are generated to reestablish local confluence of the updated AC rewrite rules. This process may lead to additional constant equalities, generating an updated canonical AC rewrite system, in effect having to recompute the canonical rewrite system using updated AC rewrite rules.
}
\ignore{If a rule $f(c_1, \dots, c_k) \rew f(d_1, \cdots, d_l)$ in $R_f$ changes due to new canonical forms of constants appearing in it, let
the result be an equation
$f(c'_1, \dots, c'_k) = f(d'_1, \cdots, d'_l)$. The two sides of the equations may further be normalized leading to a possibly terminating rewrite rule which may result in overlaps among this rule and other rules in $R_f$.}

${\bf update_{AC}}$ can also be designed so that as soon as new constant equality is generated in Step 15 generating a constant rewrite rule $c \rew d$, the rewrite rule is eagerly used by moving to $T$, all those rules in $R'_f$ whose left sides have occurrences of $c$.

Generation of additional constant equalities and additional critical pairs among AC rules terminates because of the Dickson's lemma, as in the case of {\bf SingleACCompletion}, and the fact that there are only finitely many constants which can be possibly made equal.
The correctness of ${\bf update}_{AC}$ is patterned after the correctness of ${\bf SingleACCompletion}.$

\begin{algorithm}[H]
  \DontPrintSemicolon
  \caption{$\bf {update}_{AC}(R_f, R_C, \gg_f)$}
  \label{updateAC}
  \kwInput{$R_f:\text{An AC canonical rewriting system on $f$-monomials}$,    
       $R_C:\text{A canonical rewriting system on constants}$, $\gg_f:\text{a monomial ordering on $f$-monomials extending} \gg_C \text{constants}$
       }
       
  \kwOutput{\text{An updated AC canonical rewrite system}, \text{new implied constant rules}}
  $S \gets \set{l = r  \mid NF(l, R_C) \neq l, l \rightarrow r \in R_f}$\;
  $R_f \gets R_f - \{ l \rew r \mid l = r \in S\}$\;
  $R^{'}_f \gets R_C \cup R_f$\;
  $NR_C \gets \emptyset$\;
  \While{$S \neq \emptyset$}{
    $T \gets \emptyset$\;
    \For{each $l = r \in S$}{
      $S \gets S - \{l = r\}$\;
      \uIf{$NF(l, R^{'}_f) \neq NF(r, R^{'}_f)$}{
        Orient $(NF(l, R^{'}_f), NF(r, R^{'}_f))$ as $l^{'} \rightarrow r^{'} \text{using} \gg_f$\;
        \For{each $l^{''} \rew r^{''} \in R^{'}_f$}{
         Compute a critical pair $(s_1, s_2)$ between $l^{'} \rightarrow r^{'}$ and $l^{''} \rew r^{''}$\;
         \lIf{$NF(s_1, R^{'}_f) \neq NF(s_2, R^{'}_f)$}{$T \gets T \cup \{ 
         NF(s_1, R^{'}_f) = NF(s_2, R^{'}_f)\}$}
        }
    $R^{'}_f \gets R^{'}_f \cup \set{l^{'} \rightarrow r^{'}}$\;
    Interreduce other rules in $R^{'}_f$ using $l^{'} \rightarrow r^{'}$ 
        as in step 4 of ${\bf SingleACCompletion}$\;
        {\bf if} $l^{'} \rightarrow r^{'}$ is a new constant rule {\bf then } $NR_C \gets NR_C \cup \set{l^{'} \rightarrow r^{'}}\;$
    }
    }
    $S \gets T$\;
  }
  \Return $(R^{'}_f, NR_C)$\;
\end{algorithm}

\begin{thm}
$R'_f$ as computed by the above algorithm is a canonical rewrite system for the congruence closure $ACCCC(S_f \cup S_C)$ using a total ordering $\gg_f$ on $f$-monomials extending $\gg_C$ on constants, where $S_f$ is the set of equations obtained from the rewrite rules in $R_f$ and $S_C$ are constant equations obtained from constant rules in $R_C$.
\end{thm}

\begin{proof}{\it Sketch:}
If the left sides of rules in $R_f$ do not change due to $R_C$, then $R_f$ is still a canonical rewrite system leading to $R_f \cup R_C$ being a canonical rewrite system for the congruence closure $ACCCC(S_f \cup S_C)$. 

Otherwise, all rewrite rules in $R_f$ whose left sides change due to $R_C$ are processed as equations and normalized, and critical pairs among them as well as with the rest of the rules in $R_f$ are generated as in {\bf SingleACCompletion}. Completion is redone on all the modified rules, generating a new $R'_f$ which may generate new constant equalities, leading to further propagation; this is kept track using a flag $new.$ The process is repeated until $R'_f$ is canonical and no new constant rules are generated.
\end{proof}

\ignore{Many constant equalities can be added to update an AC canonical rewrite system all together instead of adding them one by one. In general, adding a constant equality to a canonical rewrite system can lead to recomputing its canonical rewrite system because of reorientation of equations, possibly new superpositions leading to new critical pairs and thus new rules, and new constant equalities. Let $R'_f$ be an updated AC canonical rewrite system, consisting of constant rules and AC nonconstant rewrite rules relating $f$-monomials as the result of ${\bf update}_{AC}$.}

\ignore{

\begin{thm}
Given a ground canonical rewrite systems $R_U$ using a total ordering $\grt$ on ground terms in its signature
such that constants are lower than nonconstant terms, extending $R_U$ by adding a finite set of new constant equalities $\{ a_i = b_i | 1 \le i \le j\}$ generates a new canonical rewrite system $R'_f = update(R_U, \{ a_i = b_i | 1 \le i \le j\})$ such that
$ACCC(R_f) \subseteq ACCC(R'_U) = ACCC(R_U \cup \{ a_i = b_i | 1 \le i \le j \}) $. 
\end{thm}

\begin{thm}
Given a ground canonical rewrite systems $R_f$ using a total ordering $\grt$ on ground terms in its signature
such that constants are lower than nonconstant terms, extending $R_f$ by adding a finite set of new constant equalities $\{ a_i = b_i | 1 \le i \le j\}$ generates a new canonical rewrite system $R'_f = update(R_f, \{ a_i = b_i | 1 \le i \le j\})$ such that
$ACCC(R_f) \subseteq ACCC(R'_f) = ACCC(R_f \cup \{ a_i = b_i | 1 \le i \le j \}) $. Further, $R'_f = R_f \cup \{ a_i \rew b_i | 1 \le i \le j \},$ if $a_i \grt  b_i$ and $a_i$'s occurs only in the right hand sides of the rules in $R_f.$
\end{thm}

\begin{proof}{\it Sketch.}

If $a_i$ does not occur in the left side of any rule in $R_U$, then the left sides of rules in $R_U$ do not change, thus not causing any additional superpositions; the canonicity of $R_U$ is preserved.  Occurrences of $a_i$'s in the right hand side of rules in $R_U$ if any, do not affect its confluence; however, $R'_U$ is no more reduced so its right hand sides may have to be further reduced. 

If $a_i$ occurs in the left side of some rule in $R_U$, then the left side of that rule is changed which may lead to its reorientation; furthermore new overlaps because of the modified left side and/or reorientation can result in additional critical pairs, affecting confluence. Consequently, completion must be redone on all the modified rules, generating a new $R'_f$ which may generate new constant equalities, leading to further propagation.\footnote{See an example in the next section illustrating such propagation of constant equalities.} The result thus follows from the results in the previous section.
\end{proof}

}

As the reader would notice, ${\bf update}$ operations are used in the combination algorithms in Sections 7 and 8.

\section{Computing Congruence Closure with Multiple AC symbols}

The extension of the algorithms for computing congruence
closure with a single AC symbol in Section 3 to multiple AC symbols is
straightforward when constants are shared.
Given a total ordering $\gg_C$ on constants, for each AC symbol $f$, define a total well-founded admissible ordering $\gg_f$ on $f$-monomials extending $\gg_C$. It is not necessary for nonconstant $f$-monomials to be comparable with nonconstant $g$-monomials for $f \neq g,$ thus providing considerable flexibility in choosing termination orderings for computing AC ground congruence closure.

We will abuse the terminology and continue to call the AC congruence closure of a finite set $S$ of ground equations with multiple AC symbols to be $ACCC(S)$. Whereas the general definition of semantic congruence closure of ground equations for uninterpreted and interpreted symbols is given in Section 2.3, $ACCC(S)$ can be obtained as follows: 
if $f(M_1) = f(M_2) $ and $f(N_1) = f(N_2)$ in $ACCC(S)$, where $M_1, M_2, N_1, N_2$ could be singleton constants, then clearly $f(M_1 \cup N_1) = f(M_2 \cup N_2)$ is in $ACCC(S)$. Further,
for every AC symbol $g \neq f$, $g(d_1, e_1) = g(d_2, e_2) \in ACCC(S)$, where $d_1, d_2, e_1, e_2$ are new constants introduced for $f(M_1), f(M_2), f(N_1),  f(N_2)$, respectively, assuming they are already not constants; this purifies mixed AC terms.

It is possible to generate a combined reduced canonical rewrite system for multiple AC symbols in many different ways. One possibility is to independently generate a reduced canonical rewrite system for each $f \in AC$ using the {\bf SingleACCompletion} algorithm discussed above
from a finite set $S_f$ of equations on $f$-monomials using a well-founded ordering $\gg_f$ on $f$-monomials; combine the resulting reduced canonical rewrite systems for each $f$ by propagating constant rewrite rules among them with ${\bf update_{AC}}$ until no new constant rewrite rules are generated.

Another possibility is to successively use 
{\bf SingleACCompletion} to generate a canonical rewrite system for each $f \in AC$; if any new constant rewrite rules are generated, propagate them first to every reduced canonical rewrite system generated so far, to obtain any new reduced canonical rewrite systems due to the propagation. This process is repeated until no new constant rewrite rules are generated, resulting in a combined reduced canonical rewrite system for all AC symbols considered so far. Only after that, $S_f$ for an AC symbol $f$ not considered so far is processed. 
Below, the first option is employed as it is simpler and easier to understand as well as prove correct.

The algorithm below is divided into two parts:
Steps 4--9 constitute the first part; using $\bf SingleACCompletion$, a reduced canonical rewrite system $R_f$ for each AC symbol is generated after $S_f$ has been normalized using constant rules in $R_C$. Any new constant rules generated are accumulated in the variable $NR_C$. The new rules in $NR_C$ are then used to update $R_C$.

If $R_C$ changes implying that new constant rules must be propagated to each $R_f$, then the second part of the algorithm (steps 12--21) is executed. The propagation of new constant rewrite rules to an AC rewrite system $R_f$ is computed using the ${\bf update}_{AC}$ operation presented in the previous section, possibly generating an updated canonical rewrite system $R^{'}_f$ with perhaps additional constant rules, which must be further propagated. Variable $NNR_C$ keeps track of the new constant rules during this phase. This process is repeated until $NNR_C$ becomes empty, implying no new constant rules get generated during updates.

\ignore{In any given iteration of the outer loop (step 4) of the algorithm, the invariant is that (i) there is a canonical rewrite system $R_C$ of constants, and (ii) for a subset of $f \in AC$ such that equations on $f$-monomials on canonical constants have already been oriented into rewrite rules using $\gg_f$, a canonical rewrite system $R_f$ has already been generated. For an AC symbol $f$ for which a canonical rewrite system still needs to be computed, 
such that $S_f$, the set of equations on $f$-monomials, an $S_f$ is picked for generating a canonical rewrite system from $S_f$ using {\bf SingleACCompletion}. If new constant rules get generated, they are propagated to the rewrite system $R_C$ of constants and then to each of the AC canonical rewrite systems already generated. The propagation of a constant rewrite rule to an AC rewrite system is computed using the ${\bf update}_{AC}$ operation presented in the previous section, possibly generating an updated canonical rewrite system with perhaps additional constant rules, which must be further propagated. The flag variable $new$ keeps track of whether any new constant rules are generated while updating previously generated AC canonical rewrite systems in $R^{'}_f.$
}

\ignore{
Let $R_f$ be a reduced canonical rewrite system generated from ground equations in $f$ using the algorithms in Section 3.\footnote{An AC operator $f$ could have additional semantic properties as well for which algorithms were discussed in Section 4.}Let $R_{AC} = \{ R_f | f \in F_{AC} \}$ be a finite set of such reduced canonical rewrite systems for different AC symbols, computed independently from each other. Since these canonical rewrite systems in general share constants, equalities on shared constants in different canonical rewrite systems must be propagated, leading to additional superpositions and critical pair computations, as discussed below.

Let $R_C$ be the set of constant rules in each $R_f, f \in F_{AC},$ and 
$R_{AC} = R_C \cup \{ R'_f | f \in AC\}$, where and each $R'_f$ has no constant rule, i.e., it consists only of nonconstant rules of the form $f(A) \rew f(B),$ where at least one $|B| > 1, |A| > 1.$

$R_C$ is first normalized to give a canonical rewrite system on constants using the union-find data structure; each $R'_f$ is then reduced to be on the canonical constants. Henceforth, without any loss of generality, it is assumed that $R_{AC} = R_C \cup \{ R'_f | f \in AC\}$ where $R_C$ is a reduced canonical rewrite system on constants and for each $f \in F_{AC}$, $R'_f$ is a (reduced) canonical rewrite system of nonconstant rules in $f$ on constants in canonical forms.

A subtle issue is when two distinct $R_f$ and $R_g$, $f \neq g$, have rewrite rules with the same constant on their left sides, i.e.,
there is a rule $c \rew f(M) \in R_f$ and $c \rew g(N) \in R_g, f \neq g$ with $|M| > 1, |N| > 1.$
This implies that a constant $c$ is congruent to both nonconstant $f$-monomial as well as $g$-monomial. The above can happen 
if in $\gg_f$ and $\gg_g$, $c \gg_f f(M)$ and $c \gg_g g(N)$, respectively. If an additional requirement on term orderings that nonconstant terms are bigger than constants is imposed, then this issue does not arise since there can never be a rule with a constant on its left side and a nonconstant term on its right hand side. This restriction will be imposed below  on orderings $\gg_g$ for simplicity. 
There are however orderings that do not satisfy the above condition, an example being a pure lexicographic ordering in which a constant is bigger than any nonconstant term without that constant, insofar as all constants in the term are smaller. In Section 8, it will be shown how this restriction can be relaxed since it is not a limitation of the proposed approach.

To generate a combined reduced canonical rewrite system $R_S$ representing $ACCC(S)$ from $R_{AC}$, $R_f$s are combined using $R_C$ as follows. It is easy to see that
a $f$-rule in $R'_f$ and a $g$-rule in $R'_g, g \neq f$ do not interact since their outermost symbols are different.}

\ignore{
For every constant rule $c \rew d \in R_C$ such that $c$ appears in the left side of some rule, say $f(A) \rew f(B)$ in some $R_f$ with $c \in A$, the rule is replaced by $f(A-\{c\} \cup \{d\})$ which can result in the violation of the canonicity property of $R_f$ since the modified rule can superpose with other rules in $R_f$. Staring with ${R^0}_f = R_f, {R^0}_C = R_C,$ let ${S^i}_f= normalize({R^i}_f {R^i}_C)$ be the equations generated after rewriting rules  in ${R^i}_f$ using constant rules in ${R^i}_C$; if no left side of a rule in ${R^i}_f$ changes, then ${S^i}_f$ can be oriented into ${R^{i+1}}_f$ without affecting its canonicity. Otherwise, canonicity of rewrite rules corresponding to
${S^i}_f$ needs to be restored by considering superpositions and checking for joinability, possibly leading to adding new rules including derivation of new constant rules. Any additional constant rules generated in this iteration are removed from ${R^{i+1}}_f$ and added instead to ${R^i}_C$ to generate a canonical ${R^{i+1}}_C,$ which may rewrite the left sides of rules in ${R^{i+1}}_f$. When ${R^j}_C$ does not change and each ${R^j}_f$ is canonical after normalized using ${R^j}_C$, the propagation of constant equalities leading to possible violation and then restoration of canonicity of ${R^j}_f$ terminates, which is guaranteed since there are only finitely many constant and no new constants are being generated. }

\ignore{This can be done by also introducing a new constant, say $v$; wlog, if the canonical form
of $c$ is $f(M)$, then $c \rew v$ is included in $R_C$ with $c \rew g(N) \in R_g$ replaced by $g(N) \rew v$ and local confluence by considering superpositions between the new rule and other rules in $R_g.$}

\ignore{

The input to the {\bf Combination Algorithm} for generating congruence closure from ground equations with multiple AC symbols is a finite set $S_C$ of constant equations and for each AC symbol $f$, a finite set $S_f$ of equations on $f$-monomials in which at least one monomial is nonconstant, and a total ordering on $\gg_f$ on $f$-monomials extending a total ordering on $\grt$ constants in $C$.
}
\ignore{
\vspace{2mm}
\hspace*{-5mm}
{\bf Combination Algorithm ($S = S_C ~\cup ~\bigcup_{f \in F_{AC}} S_f, ~~\{ \gg_f | f \in f_{AC} \}$):}

\begin{enumerate}
 
 \item Generate a reduced canonical rewrite system $R_C$ from $S_C$ using the total ordering $\grt$ on constants in $C$; equivalently, as in step 1 of the algorithm for the single AC symbol, a union-find data structure, can be employed.
 
\item Normalize each $S_f$ using $R_c$, resulting in equations on $f$-monomials on canonical constants. Abusing the notation, continue to call the result to be $S_f$. 
  
 \item Run the AC congruence closure algorithm on each $S_f$  from Section 3 using the ordering $\gg_f$, generating a reduced canonical rewrite system $R_{f}$ for the congruence closure $ACCC(S_{f})$. Below rules on shared constants are propagated among different canonical rewrite systems, possibly generating additional superpositions and critical pairs, modifying canonical rewrite systems.
 
 \item Let $R^0_C := R_C,$ for each 
 $f \in F_AC$, $R^0_f:= R_f, f \in F_{AC}.$ 
 
 \item Let $NC := \bigcup_{f\in AC} \{ c_j \rew d_j \in R^{i}_f\}$; these are the new constant rules generated in various $R^i_f$s which need propagation to other canonical rewrite systems in which they occur.
 
 \item While $NC \neq \emptyset$
 
 \begin{itemize}
     \item $R^{i+1}_C := update(R^{i}_C, NC)$. This step updates the union-find data structure possibly generating new canonical forms for constants. 
     \item For each $f\in AC$, if any $R^i_f$ has a rule with the left side whose constant arguments changed canonical forms in $R^{i+1}_C,$ $R^{i+1}_f := update (R^{i}_f, R^{i+1}_C).$ Any AC canonical rewrite system with rules whose left sides changed due to new canonical forms of constants must be checked for confluence and joinability of any new critical pairs, if any.
     \item $i := i+1.$
     \item Let $NC = \bigcup_{f\in AC} \{ c_j \rew d_j \in R^{i}_f\}$. Any new constant equalities generated as a result of reestablishing canonicity must be propagated.
     
    Endwhile.
     
\end{itemize}

  \ignore{\item {\bf Shared constant with canonical forms in different $AC$ symbols:}
  
  If two different canonical rewrite subsystems $R_f, R_g, f \neq g$ have identical constants as the left sides,
  i.e. if there is a rule $c \rew f(M) \in R_f$ and $c \rew g(N) \in R_g, f \neq g$,  introduce a new constant $u$, make $c \grt u$ extending $\gg_f$ on $f$-monomials with $u$ making $f(M) \gg_f u$ and $g(N) \gg_g u$, 
and add a rule $c \rew u \in R_C$, replace
rules $c \rew f(M) \in R_f$ by $f(M) \rew u$ and $c \rew g(N) \in R_g$ by $g(N) \rew u$. Since $u$ is a new symbol, orderings on $f$-monomials and $g$-monomials are extended to satisfy these requirements.

Replacing $c \rew f(M)$ by  $f(M) \rew u$ can result in additional superpositions with other rules in $R_f$ and possibly new rules using $u$; this applies to  $R_g$ as well.
After ensuring local-confluence of all new superpositions and adding new rules if needed, reduced canonical rewrite systems are generated for each $R_f$.

\item If no new constant equalities are generated and the set of rewrite systems $R_f$'s do not satisfy the shared constant condition, the algorithm terminates.
  
   \item Output the combined rewrite system consisting of a reduced canonical $R_C$ on constants and a reduced canonical $R_f$ for each $f \in F_{AC}$. These canonical Rewrite systems do not share
  a constant symbol appearing on the left side of any rule. }

\end{enumerate}

}

\ignore{

  \begin{algorithm}[H]
  \DontPrintSemicolon
  \caption{Combination Algorithm for multiple AC Symbols}
  \label{combination_prof_kapur}
  \kwInput{$S = S_C \cup \bigcup_{f \in F_{AC}}{S_f}, \set{\gg_f \mid f \in
      F_{AC}}$} 
  \kwOutput{Canonical Rewriting System $\mathcal{R} = R_C \cup \bigcup_{f \in AC} R_f$}
  $\mathcal{T} \gets S$\;
  $\mathcal{R} \gets \emptyset$\;
  Let $R_C$ be the reduced canonical rewrite system generated from $S_C$ using $\gg_C$\;
  \While{$\mathcal{T} \neq \emptyset$}{
     $\mathcal{T} \gets \mathcal{T} - S_f$\;
    Replace constants in $S_f$ by their canonical forms generated by $R_C$\;
    $NR_C \gets \emptyset$\;
    $R_f \gets singleACCompletion(S_f, \gg_f)$\;
    $NR_C \gets \{ a \rew b \in R_f \mid a, b \in C$\;
    \While{$NR_C \neq \emptyset$}{
      $new \gets false$\;
      $\mathcal{R^{'}} \gets \mathcal{R}$\;
      $\mathcal{R} \gets \emptyset$\;
      \While{$\mathcal{R^{'}} \neq \emptyset$}{
        Remove $A$ from $\mathcal{R^{'}}$\;
        $(A^{'}, new^{'}) \gets updateAC(A, R_C)$\;
        $new \gets new \lor new^{'}$\;
        Let $R_C$ be the updated rewrite system of constant rules if $new'$\;
        $\mathcal{R} \gets \mathcal{R} \cup \set{A^{'}}$\;
      }
 $\mathcal{R} \gets \mathcal{R} \cup \set{R_f}$\;
    }
  }
  
  \Return $R_C \cup \bigcup_{R_f\in \mathcal{R}} R_f $\;
  
\end{algorithm}
}

  \begin{algorithm}[H]
  \DontPrintSemicolon
  \caption{{\bf CombinationMultAC}($S, \gg$)}
  \label{combination_MultAC}
\kwInput{$S = S_C \cup \bigcup_{f \in F_{AC}}{S_f}$, a finite set of constant and ground equations over multiple AC symbols;
$\gg = \set{ \gg_C, \gg_f \mid f \in
      F_{AC}}$, a family of monomial orderings extending a total ordering $\gg_C$ on constants}
      
\kwOutput{$R$, a canonical rewriting system representing AC congruence closure over multiple AC symbols} 
  
  Let $R_C$ be the reduced canonical rewrite system generated from $S_C$ using $\gg_C$\;
  $NR_C \gets \emptyset$\;
  $new \gets false$\;
  \For{each $f \in F_{AC}$}{
   Replace constants in $S_f$ by their canonical forms using $R_C$\;
   Abusing the notation, call the normalized $S_f$ also as $S_f$\;
   $R_f \gets {\bf SingleACCompletion}(S_f, ~\gg_f)$\; 
  $NR_C = NR_C \cup \{ l \rew r \in R_f \mid l, r \in C\}$\;
  }
  $R^{'}_C \gets {\bf update_C}(R_C, NR_C)$\;
  \lIf{$R^{'}_C \neq R_C$}{$new \gets true$}
  \While{$new$}{
     $NNR_C \gets \emptyset$\;
    \For{each $R_f$ with a rule $l \rew r$ such that $l$ is not in canonical form wrt $R_C$}{
        $(R^{'}_f, NR_C) \gets {\bf update_{AC}}(R_f, R_C, \gg_f)$\;
        $R_f \gets R^{'}_f$\;
        $NNR_C \gets NNR_C \cup NR_C$\;
    
    }
    \lIf{$NNR_C = \emptyset$} {$new \gets false$}
    \lElse{$NR_C \gets NNR_C$ ~~;
    $R^{'}_C \gets {\bf update_C}(R_C, NR_C)$}
    }
 
 \Return $R_S = R_C \cup \bigcup_{R_f\in \mathcal{R}} R_f $\;
  
\end{algorithm}

This repeated propagation of constant rules and subsequent recomputing of canonical AC rewrite systems eventually terminate because (i) recomputing of an updated canonical rewrite system terminates due to the Dickson's lemma, and (ii) there are only finitely many constants and no new constants are introduced during the combination; steps 11--21 in the algorithm compute this fixed point.

The termination and correctness of the combination algorithm follows from the termination and correctness of the algorithm for a single AC symbol, the fact there are finitely many new constant equalities that can be added, and ${\bf update_{AC}}$ on each $R_f$ terminates when new constant equalities are added.
The result of the combination algorithm is a finite canonical rewrite system $R_S =
R_{C} \cup \bigcup_{f \in F_{AC}} R_f$,
a disjoint union of sets of canonical rewrite rules on $f$-monomials for each AC symbol $f$, along with a canonical rewrite system $R_C$ consisting of constant rules such that the left sides of rules
are distinct. 

Given that {\bf SingleACCompletion} and ${\bf update}_{AC}$ generate reduced canonical rewrite systems, $R_S$ is 
unique for a given set of ground equations $S$ in the extended signature assuming a family of total admissible orderings
on $f$-monomials for every $f \in F_{AC}$, extending a total ordering on constants; this assumes a fixed choice of new constants standing for the same set of subterms during purification and flattening. 

\ignore{In the original signature, however,
$R_S$ is neither unique nor even locally confluent (or canonical)  if it includes a shared constant having multiple canonical forms in two different AC symbols as illustrated in the above example. It then becomes necessary to compare monomials in different $AC$ symbols. Assuming a union-find data structure is used to keep canonical forms of constants up to date, the right hand sides of AC rules in each $R_f$ are also reduced.The theorem below follows from the corresponding results for a canonical rewrite system for each AC symbol $f$ and the properties of $\bf{update}$ particularly, it generates from a canonical rewrite system, an updated canonical rewrite system representing the extended congruence closure that includes
the implied constant equations generated from the updates.}

\begin{thm}
\ignore{Given a set of well-founded admissible orderings $\gg_f$ on $f$-monomials for each $f \in F_{AC}$ which extend a total ordering on constants in $C$, let $R_{S} = R_C \cup \bigcup_{f \in F_{AC}} R_{f}$ be a
finite rewrite system, where 
$R_C$ is a unique reduced canonical rewrite system of constant rules, and
(ii) disjoint unique reduced canonical rewrite subsystems $R_f$ computed using

 $u \rew g(v'_1, \cdots, v'_{k'})$, 
 then $R_{S}$ is reduced and canonical.}
 
 $R_S$ as defined above is a  canonical rewrite system associated with $ACCC(S)$
 in the extended signature for
 a given family $\{ \gg_f | f \in f_{AC} \}$ of admissible orderings on $f$-monomials extending a total ordering $\gg_C$ on constants,
 such that $ACCC(R_S)$, with rewrite rules in $R_S$ viewed as equations, and restricted to the original signature is $ACCC(S)$.
\end{thm}

The proof of the theorem follows from the fact that (i) each $R_f$
is canonical using $\gg_f$, (ii) the left sides
of all rules are distinct, (iii) these rewrite systems are normalized using $R_C$. A stronger result follows if {\bf SingleACCompletion} generates a reduced unique canonical rewrite system for an AC rewrite system using $\gg_f$ and ${\bf update}_{AC}$ keeps the rules in the rewrite system in reduced form, thus maintaining a unique reduced AC canonical rewrite system after updates by new constant rules.

{\bf Example 7:}
To illustrate, consider the following set of ground equations after purification, involving $+$ and $*$: $\{~1. ~a*a*b*b = a,~ 2. ~a*b*b*b = b,~3.~ a*a*a*b = a, ~4.~ a + b = b, ~5. ~b + b = a \}.$ A subsystem of $*$-equation is: $\{ 1, 2, 3 \},$ and a subsystem of $+$-equations is $\{4, 5\}.$ Using a total ordering $ b > a$ on constants, and degree-ordering on $+$ and $*$ monomials, the equations are oriented from left to right.

A canonical rewrite system for $+$-equations is generated using the algorithm {\bf SingleACCompletion}: There is a critical pair between rules 4 and 5 with the superposition $ a + b + b$ producing the critical pair:
$(a + a, b + b)$, giving a new rule: $6. ~a+a\rew a$; The set $\{4, 5, 6\}$ is a reduced canonical rewrite system on $+$-rules. 

Similarly, {\bf SingleACCompletion} is run on the rules corresponding to $1, 2, 3$;  a new rule $7.~a * b \rew a* a$ is generated from $1, 3$ (the superposition is $a*a*a*b*b$ leading to the critical pair $(a*a, a*b)$), reducing all other rules, giving a canonical
rewrite system  on $*$-terms:$\{ 8. ~a*a*a*a \rew a, ~9.~ b\rew a\}.$ 

The two canonical rewrite system are combined using the constant rewrite rule $b \rew a$. It reduces the canonical rewrite system for
$+$, collapsing all $+$ rules to $10.~a + a \rew a.$ The combined
canonical rewrite system for the two subsystems is: $\{ a*a*a*a \rew a,~ a + a \rew a, ~b \rew a\}$. 

Below, variations of Example 7 are considered to illustrate other subtle aspects of the combination algorithm. If the set of ground equations in Example 7 include additional ground equations in $+$-subsystem: $\{ a + c = d, b + d = c\},$ with $d > c > b > a$, further propagation is caused. Given that rule 9 reduces $b$ to $a$ from the above steps, the additional equations when oriented from left to right, would lead to a critical pair between $a + c \rew d, a+d\rew c,$ generating new rule $d + d \rew c + c.$ This illustrates that the propagation of new constant rules can lead to additional superpositions among rules which were previously canonical.

{\bf Example 8:} Consider another variation on Example 7 in which the $+$-subsystem is different: $\{c + c = b + b, c + b = c + c\}.$ Using the lexicographic ordering induced by $b > c > a$ on $+$ monomials, the canonical rewrite system is: $\{b + b \rew c + c, ~ c + b \rew c  + c.\}$. The canonical rewrite system on $*$ monomials is still $\{8. ~a* a * a * a \rew a,~9.~ b \rew a\}.$ With the propagation of $b \rew a$ generated from the canonical rewrite system for $*$-monomials to the canonical rewrite system for $+$-monomials, the orientation of the rewrite rules in it changes to $c + c \rew a + a, c + c \rew c + a$, giving a new rule: $ c + a \rew a + a.$ The resulting combined reduced rewrite system is: $\{a * a * a * a \rew a, c + c \rew a + a, c + a \rew a + a, b \rew a\}.$

\subsection{A Nondeterministic Version}

The above combination algorithm can also be presented as an inference system; inference rules can be non-deterministically applied, dovetailing  various steps in generating AC rules in different $R_f$'s without generating a full canonical rewrite system for any $R_f$ one at a time. Constant rules can be eagerly propagated and applied everywhere as soon as they are generated, making the combination algorithm more efficient in practice, since the complexity of an AC congruence closure algorithm as well as the size of the associated canonical rewrite systems depend upon the number of constants used in the algorithm: fewer the number of constant, fewer the steps. A possible disadvantage of this view is that it becomes necessary to give correctness and termination arguments using complex proof and termination orderings (i.e., union of all $\gg_f$ in this case) as a single object; consequently, proofs cannot be factored out across different AC symbols, sacrificing modularity.

\ignore{Let us revisit the case when a shared constant appears as the left side of rules in two different reduced canonical AC subsystems. Consider a simple example: $\{ f(a, c) = b, f(a, b) = u, g(a, b) = u, g(a, c) = b \}.$ Using a lexicographic ordering on $f$-terms and $g$-terms induced by a total ordering $u > c > b > a$, $f$-equations are oriented as:
$\{ f(a, c) \rew b, u \rew f(a, b) \} $; similarly, $g$-equations are oriented as:
$\{ g(a, c) \rew b, u \rew g(a, b) \} $. Both of these subsystems are also reduced canonical and they share $u$. Using the method discussed above, a new symbol $v$ is introduced with $ u > v$; if we make $c > b > a > v$, we get:
$\{ u \rew v \}, \{f(a, c) \rew b, f(a, b) \rew v \}$, and
$\{g(a, c) \rew b, g(a, b) \rew v \}.$ $R_f$ and $R_g$ are no longer canonical as there are additional superpositions because of reorientation of rules corresponding to $f(a, b) = u, g(a, b) = u. $ The new canonical rewrite system for $f$-rules is:
$\{f(a, c) \rew b, f(a, b) \rew v, f(c, v) \rew f(b, b) \};$ constant $u$ does not appear in this subsystem; similarly for $g$-rules, the canonical rewrite system is:
$\{g(a, c) \rew b, g(a, b) \rew v, g(c, v) \rew g(b, b) \}.$ The combined canonical rewrite system representing the congruence closure is:
$\{f(a, c) \rew b, f(a, b) \rew v, f(c, v) \rew f(b, b), 
g(a, c) \rew b, g(a, b) \rew v, g(c, v) \rew g(b, b),  u \rew v \}$. Both $f(a, b)$ and $g(a, b)$ have the same canonical form, $v$; however $f(a, b, c)$ and $g(a, b, c)$ do not have the same canonical form and hence are not in the congruence closure. }

\section{Congruence Closure with Multiple AC and Uninterpreted symbols}

The algorithm presented in the previous section to compute AC congruence closure with multiple AC symbols is combined with the congruence closure algorithm for 
uninterpreted symbols in Section 2.2. Similar to combining canonical rewrite systems for various AC congruence closures in the previous section,
the key operation is to update various canonical rewrite systems when equalities on shared constants are generated and propagated. Since the algorithm for generating a rewrite system for congruence closure over uninterpreted symbols is of the least computational complexity in contrast to the complexity of AC congruence closure algorithms, it is always preferred to be run before running other AC congruence closure algorithms. The combination algorithm is thus centered around the congruence closure algorithm for the uninterpreted symbols.

\ignore{

To  simplify the presentation, the general algorithm is divided into two stages. In the first stage,

\begin{enumerate}
    
\item Run the union-find constant closure algorithm on $S_C.$ The output is a forest of directed trees with constants as their roots, serving as canonical forms of the equivalence classes of constants in each tree.

\item Run Kapur's congruence closure algorithm on flat rules with the output being for each uninterpreted symbol $h \in S_U$, $R_U = \{ h(c^i_1, \cdots, c^i_k) \rew d_i | 1 \le i \le j\}$ with distinct left sides on roots. In case of two nonconstant flat rules with identical left sides on canonical constants, an equality on constants on their right hand sides is generated which is propagated into updating the Union-find data structure and hence the $R_U$ possibly leading to further propagation of constant equalities.
\end{enumerate}

The result of the above is a union-find data structure giving a canonical constant for each equivalence class of constants and for  each $h \in F_U,$ flat rules with distinct left sides.
}

\ignore{In the combination algorithm below, a canonical rewrite system $R_C$ on constant equations is first generated using $Canonical_C$ from a finite set $S_C$ on constant equations using a total ordering $\gg_C$ on constants. A canonical rewrite system  $R_U$ of flat rules on uninterpreted symbols is then generated using $Canonical_U$ from a finite set of flat equations $S_U$ using $R_C$ and $\gg_U$ in which every nonconstant uninterpreted symbol is greater than every constant.}

In the combination algorithm below, $CC$ stands for the congruence closure algorithm on constant equations $S_C$ and uninterpreted ground equations $S_U$, using a total ordering $\gg_C$ on constants and $\gg_U$, an ordering in which uninterpreted function symbols are bigger than constants in $C$. $R_U$, the output of $CC$, is a reduced ground canonical rewrite system representing the congruence closure of $S_C \cup S_U$. $R_C$ stands for the constant rules in $R_U$.

For every $f \in AC$, a finite set $S_f$ of equations on $f$-monomials is successively picked; constants in the equations are replaced by their canonical forms using $R_C$ and the rewrite rules are made from the resulting equations using a total admissible ordering $\gg_f$ on $f$-monomials.
{\bf SingleACCompletion} algorithm is used to generate a canonical AC rewrite system $R_f$ from $S_f$ using a monomial ordering $\gg_f$. Variable $NR_C$ accumulates the constant rules generated for each $f$ in its respective reduced canonical rewrite system $R_f$ (steps 5--10).

The function ${\bf update_C}$
is used to update the canonical constant rewrite system $R_C$, using any new constant rules in $NR_C$.
If the update results in a new canonical constant rewrite system, then ${\bf update_U}$ updates the ground canonical rewrite system for the congruence closure of uninterpreted equations, possibly leading to new constant equalities and associated rewrite rules. A flag variable $new$ keeps track of whether any new constant rules have been generated which must be propagated.

Steps 14--26 perform the propagation of new constant rewrite rules to a reduced canonical rewrite system $R_f$ for each AC symbol $f$ by invoking ${\bf update_{AC}}$.
If new constant rewrite rules are generated, they are first propagated to $R_U$ (as in step 13), checking for any new constant rewrite rules, which again must be propagated back to canonical AC rewrite systems in which those constants appear.
This propagation among various AC canonical rewrite systems and a rewrite system for congruence closure over uninterpreted symbols, continues until no additional constant rewrite rules are generated.

\ignore{A flag variable $new$ keeps track of whether new constant rewrite rules are generated in which case, these rules are first propagated to $R_C$ and then to $R_U$ which may further generate new constant rewrite rules; once the propagation stabilizes, resulting in $R_C \cup R_U$ which is reduced and canonical (steps 11), the new constant rules then must be successively propagated to each reduced canonical AC rewrite system already generated (step 14). Whenever any additional constant rules are generated by ${\bf update}_{AC}$ as tracked by $new$, they are immediately propagated to $R_C$ and $R_U$, updating them (steps 14 and 17). New constant rewrite rules are always eagerly applied to $R_C, R_U$ and once no more constant rules are generated there, they are used to update already generated canonical rewrite systems $R_f$'s for AC symbols in $F_{AC}$.}

\ignore{ Assuming $R_C, R_U$, as constructed above, one way to combine congruence closures over  multiple AC symbols with uninterpreted symbols is to do it incrementally by generating a canonical rewrite system for each $R_f$ separately and propagating any constant equalities thus generated into $R_C, R_U$, updating the union-find data structure on constants and flattened rules in uninterpreted symbols. Then, a canonical rewrite system for each of the remaining AC symbols are considered one by one, possibly generating equalities on shared constants which are propagated in the second stage.}

\ignore{An alternative is to compute canonical rewrite systems for AC congruence closure for each $f \in F_{AC}$ as in Section 6, collect the constant rules from each $R_f, f\in F_{AC}$ and update $R_C, R_U$ using them, with propagating new constant equalities back to each $R_f, f \in F_{AC}$ as is done in the algorithm in Section 6. This algorithm is presented below.}

\ignore{
pick an AC symbol, say $f \in F_{AC}$, compute a canonical rewrite  system $R_f$ for ground equations over $f$ using the algorithm presented in Section 4, (ii) use constant rules in $R_f$, if any, to update $R_C \cup R_U$ which may generate additional constant equalities which are used to update $R_f$. Do this for every $R_f, f \in F_{AC}$. Upon termination, a canonical rewrite system representing a combination of congruence closure of all $R_f$'s considered so far with the uninterpreted ground equations is generated, which is then combined with the next $R_f$

for  repeatedly using the {\bf update} function introduced in Section 5 for combining canonical rewrite systems using the following loop.

\begin{enumerate}
    
\item For each $f \in F_{AC},$
generate a canonical rewrite system  $R_f$ from $S_f$ after it has been normalized using $R_C \cup R_U$.

\item For all constant rules $C_f := \{ c_i \rew d_i | 1 \le i \le j_f \} in R_f$, $(R^{i+1}_C, R^{i+1}_U) := update(R^i_C \cup R^i_U, C_f)$.

\item For all new constant rules in $R^{i+1}_C$ generated in the previous step, $R^{i+1}_f := 
update(R^i_f, R^{i+1}_C - R^{i}_C).$

\end{enumerate}

Given a total ordering $\grt$ on $C$, let $\gg_U = \gg \cup \{ h \gg_U c, h \in F_U, c \in C\}$. For each AC symbol $f$, define a total admissible ordering $\gg_f$ on $f$-monomials
extending $\grt$ on $C$.\\

\noindent
{\bf General Congruence Closure($S = S_C \cup S_U \cup \cup_{f \in F_{AC}} S_f,  {\gg_U} \cup \{ \gg_f | f \in F_{AC}\}$)}

The combination is again straightforward given that the output of the congruence closure algorithm on uninterpreted symbols is a unique reduced canonical rewrite system consisting of flat rules of the form $h(a_1, \cdots, a_k) \rew b$ and constant rules $a \rew b.$
There is no interaction between flat rules and other rules generated from AC monomials. When new constant equalities are generated, they can reduce flat rules, making the left sides of some flat rules equal,  resulting in additional equalities which are handled in the same way as constant equalities generated during completion on equations on  $f$-monomials. All other steps are the same as in the case of the congruence closure algorithm in the previous section for multiple AC symbols. The output of this general algorithm shares the properties of the output of the congruence closure over multiple AC symbols.

}

\ignore{Reduced canonical rewrite systems are separately generated from $S_C$, the set of constant equations, from $S_U,$ the set of flat equations with uninterpreted symbols, and from $S_f, f \in F_{AC}$, equations on $f$-monomials.

Without any loss of generality, assume reduced canonical rewrite systems $R^0_C, R^0_U$ and $R^0_f, f \in F_{AC}$ are already generated, where $R^0_C$ is the canonical set of constant rules with distinct left sides, $R^0_U$ is the canonical set of nonconstant flat rules in uninterpreted symbols with distinct left sides, and for each $f \in F_AC$, $R^0_f$ is the canonical set of nonconstant rules on $f$-monomials.

$R_C$ be a canonical rewrite system of all the constant rules from $R_U$ and each $R_f$; let $R_h$ be a canonical rewrite rules consisting of flat rules $h(a_1, \cdots, a_k) \rew b$ for each uninterpreted symbol $h$; for each $f \in F_{AC}$, let $R'_f$ be the nonconstant $f$ rules after moving constant rules from $R_f$ to $R_C$; abusing the notation and without any loss of generality, $R_f$ will be used instead of $R'_f$. 
}

\ignore{It is preferable to generate $R_C$ first and then generate $R_U$ to check if any implied constant equalities are generated. The result is
a reduced canonical rewrite system $R_C \cup R_U$ for the congruence closure of $S_C \cup S_U$ over uninterpreted symbols. 
$R_C, R_U$ can be computed very fast in $O(n*log(n))$ steps, whereas
computing $R_f$ from a set of $f$-monomial equations is very expensive, so it always pays off to deduce constant equalities from $R_C$ and $R_U$. During the computation of generating canonical rewrite systems for $R_f$ from $S_f, f \in F_{AC},$ if a new constant equality is implied and generated, it is eagerly used to update $R_C \cup R_U$ to generate any new implied equalities, and used to update monomial equations and monomial rewrite systems constructed so far.}

\ignore{In the following steps, reduced canonical rewrite systems $R_f$'s are generated from each $S_f$ which is already reduced by $R_C$. This is the algorithm for computing the congruence closure for multiple AC symbols.
If during the generation of a reduced canonical rewrite system of any $R_f$, an implied constant equality gets generated, it is propagated by adding to $R_C$,  changing left sides of the rules, if any, in which those constants appear, leading to generation of additional superpositions and critical pairs.

The algorithm from the previous subsection for computing reduced canonical rewrite systems from each $S_f$ is applied, looking for new constant equalities generated and checking shared constant condition. As discussed in the previous subsection, in both cases, local confluence of $R_f$'s may have to restored for checking additional superpositions, leading to possibly new rules.}

\begin{algorithm}[H]
  \DontPrintSemicolon
 \caption{{\bf GeneralCongruenceClosure}($S, ~\gg$)}
\label{GeneralCC}

 \kwInput{$S = S_C \cup S_U \cup \bigcup_{f \in F_{AC}}{S_f}$, a finite set of equations on constant, terms with uninterpreted symbols, and $AC$ monomials, $\gg =  \set{ \gg_C, \gg_U }
    \cup \set{\gg_f \mid f \in F_{AC}}$
 a family of orderings on constants, uninterpreted symbols,and $f$ monomials extending a total ordering on constants with uninterpreted symbols bigger than constants} 
    
  \kwOutput{A canonical rewriting system $R_S = R_C \cup R_U \cup \bigcup_{f \in AC} R_f$, representing general congruence closure}

  $R_U \gets CC(S_C \cup S_U, \gg_U \cup \gg_C)$\;
  $R_C \gets \{ l \rew r \in R_U \mid l, r\in C\}$\;
 $NR_C \gets \emptyset$\;
  $new \gets false$\;
  \For{each $f \in F_AC$}{
   Replace constants in $S_f$ by their canonical forms using $R_C$\;
   Abusing the notation, call the normalized $S_f$ also as $S_f$\;
   $R_f \gets {\bf SingleACCompletion}(S_f, ~\gg_f)$\; 
  $NR_C = NR_C \cup \{ l \rew r \in R_f \mid l, r \in C\}$\;
  }
$R^{'}_C \gets {\bf update_C}(R_C, NR_C)$\;
  \lIf{$R^{'}_C \neq R_C$}{$new \gets true$\;
  $R_U \gets {\bf update_U}(R_U, NR_C)$}\
 \While{$new$}{
     $NNR_C \gets \emptyset$\;
     $new \gets false$\;
    \For{each $R_f$ with a rule $l \rew r$ such that $l$ is not in canonical form wrt $R_C$}{
        $(R^{'}_f, NR_C) \gets {\bf update_{AC}}(R_f, R_C, ~\gg_f)$\;
        $R_f \gets R^{'}_f$\;
        $NNR_C \gets NNR_C \cup NR_C$\;
    
    }
    \uIf{$NNR_C \neq \emptyset$}{$NR_C \gets NNR_C$\;
    $R_U\gets {\bf update_U}(R_U, R_C \cup NR_C)$\;
    $new \gets true$\;
    }
 }
  
\Return $R_U \cup \bigcup_{R_f\in R} R_f $\;
\end{algorithm}

Let $\gg_U = \grt \cup \{ h \gg c, h \in F_U, c \in C\}$. For each AC symbol $f$, define a total admissible ordering $\gg_f$ on $f$-monomials
extending $\grt$ on $C$. The input to the {\bf General Congruence Closure} algorithm is a union of finite set $S_C$ of constant equations, finite set $S_U$ of nonconstant flattened equations in uninterpreted symbols in $U$, and for each AC symbol $f$, a finite set $S_f$ of $f$-monomial equations in which at least one $f$-monomial is not a constant, and total orderings
$\gg_U$ and $\gg_f$ on $f$-monomials extending a total ordering $\gg_C$ on constants in $C$.

\ignore{

\vspace{2mm}
\hspace*{-5mm}
{\bf General Congruence Closure($S = S_C \cup S_U \cup \cup_{f \in F_{AC}} S_f,  {\gg_U} \cup \{ \gg_f | f \in F_{AC}\}$)}

{\bf Stage I:}

\begin{enumerate}
 
 \item Generate a reduced canonical rewrite system $R_C$ from $S_C$ using the total ordering $\grt$ on $C$; equivalently, as in step 1 of the algorithm for the single AC symbol, a union-find data structure  can be employed.
 
 \item Flattened equations in $S_U$ are oriented into rewrite rules. For every uninterpreted function symbol $h \in F_U$, there are rewrite rules in $R_h = \{ h(c^i_1, \cdots, c^i_k) \rew d^i | 1 \le i \le j \}$ on roots (canonical forms of constants) of the Union-Find data structure. Any two rules with identical left sides generate an implied constant equality equating their respective right hand sides, which is included in $R_C$, thus updating union-find data structure and hence, each $R_h$. At the end of this step, the result is a union-find data structure on canonical constants and $R_h$ for each uninterpreted symbol $h \in S_U$ with rules having distinct left sides.
 
 \ignore{By a canonical rewrite system for congruence closure over uninterpreted symbols, $R_C$ is assumed to be the Union-find data structure maintaining canonical forms for equivalence classes over constants with flat rules for each uninterpreted symbol $h \in F_{U}.$}
As various canonical rewrite systems are being combined below, a canonical rewrite system of congruence closure over uninterpreted symbols as the union-find data structure on constants, and distinct flat rules for each uninterpreted symbol $h \in F_{U}$ is maintained as an invariant.

\item Normalize each $S_f$ using $R_C$, resulting in equations on $f$-monomials on canonical constants; after normalization, abusing the notation, the resulting equations on $f$-monomials are continued to be called $S_f.$
  
 \item Run the AC congruence closure algorithm on each $S_f$  from Section 3, using the ordering $\gg_f$, generating a reduced canonical rewrite system $R_{f}$ for the congruence closure $ACCC(S_{f})$.
\end{enumerate}

At the conclusion of stage I, 
 let $R^0_C := R_C,$ for each $h \in F_U$, $R^0_h = R_h,$ for each 
 $f \in F_{AC}$, $R^0_f := R_f, f \in F_{AC}.$ In the second stage, equalities on shared constants are propagated using {\bf update} operation. As discussed above,  if rules generated from shared constant equalities only change the right side of other rules in canonical rewrite systems, then changes are minimal. Otherwise, if the left sides of rules in canonical rewrite systems change, then modified rules may need to be reoriented and additional superpositions may be generated, leading to additional new rules including new constant equalities.

 \begin{enumerate}
 
 \item $NC_U = \emptyset.$ $NC_{AC} := \bigcup_{f\in AC} \{ c_j \rew d_j \in R^{i}_f \}$. $NC_U$ includes the new constant equalities generated in the uninterpreted case; $NC_{AC}$ includes the new constant equalities generated in the canonical AC rewrite systems.
 
 \item While $NC_{AC} \neq \emptyset$
 
 \begin{enumerate}
  \item  While $NC_U \neq \emptyset$
     \item $R^{i+1}_C := update(R^{i}_C, NC_{AC} \cup NC_U)$.
     \item For each $h \in F_U,$ update $R_h$ to reflect changes in the union-find data structure for $R^{i+1}_C.$ Let $NC_U = \bigcup_{h \in F_U} 
     \{ c_j \rew d_j \in R^{i+1}_h\}$. 
     
    Endwhile. 
    
    This ensures a canonical rewrite system for congruence closure on uninterpreted equations.
    \end{enumerate}
     
     \item for each $f\in F_{AC}$, if any $R_f$ has a rule with the left side whose constant arguments changed canonical forms in $R^{i+1}_C,$ $R^{i+1}_f := update (R^{i}_f, R^{i+1}_C).$ Changes in the left sides of rules in $R_f$ may lead to reorientation and additional superpositions; as a result, the confluence of the modified rewrite system must be restored.
     \item $i := i+1.$
     \item Let $NC_{AC} = \bigcup_{f\in F_{AC}} \{ c_j \rew d_j \in R^{i}_f\}$. Any new constant equalities generated while restoring confluence of AC rewrite systems must be propagated.
    
   Endwhile.

\end{enumerate}
}
\ignore{
\begin{algorithm}[H]
  \DontPrintSemicolon
  \caption{General Combination Algorithm}
  \label{combination_prof_kapur_section_8}
  \kwInput{$S = S_C \cup S_U \cup \bigcup_{f \in F_{AC}}{S_f}, \set{ \gg_U }
    \cup \set{\gg_f \mid f \in F_{AC}}$} 
  \kwOutput{A Canonical Rewriting System $R_S = R_C \cup R_U \cup \bigcup_{f \in AC} R_f$}
  
  $R_C \gets Canonical_C(S_C, \gg_C)$\;
  $R_U \gets Canonical_U(S_U, R_C, \gg_U)$\;
  $\mathcal{T} \gets \bigcup_{f \in F_{AC}}{S_f}$\;
  $\mathcal{R} \gets \emptyset$\;
  \While{$\mathcal{T} \neq \emptyset$}{
    Pick and remove $S_f$ from $\mathcal{T}$\;
    Replace constants in $S_f$ by their canonical forms generated by $R_C$\;
    $(R_f, new^{'}) \gets singleACCompletion(S_f, \gg_f)$\;
    Let $R_C$ be the updated rewrite system of constant rules if $new'$\;
    $\mathcal{R} \gets \mathcal{R} \cup \set{R_f}$\;
    \uIf{$new^{'} = true$}{
      $(R_U, new) \gets {\bf update}_U(R_U, R_C)$\;
      Let $R_C$ be the updated rewrite system of constant rules if $new$\;
      \While{$new^{'} = true$ or $new = true$}{
        \For{each $R_{f}$ in $R$ and $new^{'} = false$}{
          $(R_f, new^{'}) \gets {\bf update}_{AC}(R_f, R_C)$\;
          Let $R_C$ be the updated rewrite system of constant rules if $new'$\;
        }
        \uIf{$new^{'} = true$} {$(R_U, new) \gets {\bf update}_U(R_U, R_C)$\;
        Let $R_C$ be the updated rewrite system of constant rules if $new'$\;
        
        }
      }

    }
  }
   \Return $R_C \cup R_U \bigcup_{R_f\in R} R_f $\;
\end{algorithm}
}

The termination and correctness proofs of the above algorithm are factored: they follow
from
the termination and correctness of the algorithm for a single AC symbol, ${\bf update_U}$.
${\bf update_{AC}}$ and the fact there are finitely many new constant equalities that can be added. Further the outputs in each case are reduced canonical rewrite systems.

The result is 
a modular combination, whose termination and correctness is established in terms of the termination and correctness of its various  components: (i) the termination and correctness of algorithms for generating canonical rewrite systems from ground equations in a single AC symbol, for each AC symbol in $F_{AC}$, and their combination together with each other, and (ii) the termination and correctness of congruence closure over uninterpreted symbols and its combination with the AC congruence closure for multiple AC symbols, and (iii) there are only finitely many constants being shared and finitely many possibly equalities on the shared constants.

The result of the above algorithm is a finite reduced canonical rewrite system from $S$, $R_S=
R_{C} \cup R_{U} \cup \bigcup_{f \in F_{AC}} R_{f}$,
a disjoint union of sets of reduced canonical rewrite rules on $f$-monomials for each AC symbol $f$, along with a canonical rewrite system $R_C$ consisting of constant rules and $R_U$ consisting of flat rules on uninterpreted symbols, such that the left sides of rules
are distinct and the right sides are reduced. $R_S$ is 
unique in the extended signature assuming a family of total admissible orderings
on $f$-monomials for every $f \in F_{AC}$ and making uninterpreted symbols $h \in U$ extending a total ordering on constant if $R_C, R_U$ are reduced, {\bf SingleACCompletion} generates a reduced AC canonical rewrite system and ${\bf update}_U$ as well as ${\bf update}_{AC}$ ensure that the canonical rewrite system generated by them are also reduced.

\ignore{
\begin{enumerate}

    \item From $S_C \cup S_U,$ generate a reduced canonical rewrite system $R_C \cup R_U$ representing the congruence closure over uninterpreted symbols such that $R_C$ is the canonical reduced rewrite system only constants generated using $\grt$ on $C$ and $R_U = \cup_{h \in U} R_h,$ where $R_h$ consists of nonconstant rules of the form $h(c_1, \cdots, c_k) \rew c$ with distinct left sides obtained after normalizing using $R_C$.
    
    \item Normalize $S_f, f \in F_{AC}$ using $R_C$; abusing the notation, call the normalized result also to be $S_f$.
    
     \item Run the AC congruence closure for multiple AC symbols from the previous subsection, on the output from the previous step, using $\gg_f$ for each $f \in F_{AC}$.
     
     Let $R_f$ be a reduced canonical rewrite system generated from $S_f$ using $\gg_f, f \in F_{AC}.$ Remove all constant rules from each $R_f$ after including them in $R_C$. Interreduce $R_C$ to generate a canonical rewrite system on constants; abusing the notation, this is also called $R_C$.

    \item Let ${R_C}^0 = R_C, {R_h}^0 = R_h, h \in F_U, {R_f}^0 = R_f, f \in F_{AC}.$ 
  
 \item For any ${R_h}^i$ such that the left side of a rule in it is normalized using ${R_C}^i,$ generate new equalities on constants if two left sides in the normalized ${R_h}^i$ become identical, i.e., ${S_h}^{i} = update({R_h}^i, {R_C}^i)$ has two equations with identical nonconstant flat $h$ terms; in that case, generate ${R_h}^{i+1}$ from ${S_h}^{i} = update({R_h}^i, {R_C}^i)$ and
 include new constant rule into ${R_C}^{i+1}$.
 
 For any $f \in F_{AC}$ such that the left side of a rule in ${R_f}^i$ is normalized by a rule in ${R_C}^i$, generate ${R_f}^{i+1}$ from ${S_f}^{i} = update({R_f}^i, {R_C}^i),$ where if for some rule $f(A) \rew f(B) \in {R_f}^i$, if a constant $c \in A$ is rewritten using ${R_C}^i$, then ${S_f}^i = \{ normalform(f(A), {R_C}^i) = normalform(f(B), {R_C}^i) | f(A) \rew f(B) \in {R_f}^i\}.$
 
 \item Orient equations in ${S_f}^i$ to rewrite rules and restore their canonicity generating ${R_f}^{i+1}.$ If ${R_f}^{i+1}$ includes a constant rule $c \rew d,$ then remove it from ${R_f}^{i+1}$ and add it to ${R_C}^i$ after orienting it.  Generate ${R_C}^{i+1}$ from ${R_C}^i.$
 
 \item If no new constant rule is generated in the previous step and the canonicity of each ${R_f}^{i+1}$ is restored, terminate, generating a canonical rewrite system $R = {R_C}^{i+1} \cup {R_f}^{i+1},$ representing $ACCC(S)$.
  
   \item Output the combined rewrite system consisting of a reduced canonical $R_C$ on constants, a reduced canonical system of nonconstant $h$-rules for every uninterpreted symbol $h \in U,$ and a reduced canonical $R_f$ consisting of $f$-rules for each $f \in F_{AC}$. These canonical Rewrite systems do not share a constant symbol appearing on the left side of any rule.

\end{enumerate}

The termination and correctness of the algorithm follows from
the termination and correctness of each of the component subalgorithms for constants, uninterpreted symbol, and a single AC symbol, the  and the fact there are finitely many new constant equalities that can be added during the propagation. 
}

\ignore{As a result, the termination of the general algorithm follows  from the termination of the algorithms for each of the components--$S_U$ and $S_f$, for every AC symbol $f$. The correctness proof of the general algorithm is also structured
in a modular fashion using the correctness proofs of the components $S_U$ and $S_f$, $f \in F_{AC}$.}

\begin{thm}
Given  a set $S$ of ground equations, the above algorithm generates a reduced canonical rewrite system $R_S$ on the extended signature such that $R_S = R_C \cup R_U \cup \bigcup_{f\in F_{AC}} R_f$, 
where each of $R_C, R_U, R_f$ is a reduced canonical rewrite system using a set of total admissible monomial orderings $\gg_f$ on $f$-monomials, which extend orderings $\gg_U$ on uninterpreted symbols and $\grt$ on constants as defined above, and $ACCC(R_S),$ with rules in $R_S$ viewed as equations, is $ACCC(S)$ when restricted on the original signature.
Further, for this given set of orderings, $\gg_U$ and $\{ \gg_f ~|~f \in F_{AC}\}$, $R_S$ is unique for $S$ in the extended signature. 
\end{thm}

The proof follows from the respective proofs of each of the component algorithm and is routine.

As stated in the previous section on the combination algorithm for multiple AC symbols, a nondeterministic version can be formulated in which various inference steps are  interleaved along with
any new constant equalities generated to reduce constants appearing in other $R_f$ eagerly, instead of having to complete an AC rewrite system for each AC symbol. Correctness and termination proofs are however more complicated requiring complex orderings.

\ignore{Even though there are new constants introduced in the generation of a reduced canonical rewrite system if two different $R_f, R_g$
have the same constant appearing on the left sides of rules with $f$-monomials and $g$-monomials on their right side, the number of choices for canonical forms is finite given that there are only finite many AC symbols and finitely many constants.}

If any AC symbol $f$ has additional properties such as idempotency, nilpotency, identity, cancelativity, or being a group operation, and various combinations (see discussion in Sections 4 and 5), the corresponding reduced canonical $R_f$ generated from the specialized extensions of {\bf SingleACCompletion} algorithms discussed in Sections 4 and 5 can be used and the above results extend.

\section{Pure Lexicographic Orderings: Constants bigger than Nonconstant Terms}

So far, it is required that constant symbols are smaller in monomial orderings than nonconstant terms for any function symbol, be it AC or nonAC. This restriction can be relaxed. In case of uninterpreted (non-AC) symbols, this is achieved by introducing a new constant symbol as already done in Kapur's congruence closure \cite{KapurRTA97}. A similar approach works in case of nonconstant terms with multiple AC symbols.

Using a monomial ordering in which a constant is bigger than a nonconstant monomial, the above combination algorithm for generating a canonical rewrite system in the presence of multiple AC symbols could produce at least two canonical rewrite subsystems for two different AC symbols $f$ and $g$ such that $R_f$ has a rule
$c \rew f(A)$ and $R_g$ has another rule
$c \rew g(B)$, with $A$ and $B$ of size $> 1$ (implying that the right sides are nonconstants). Then, $c$ can possibly have either a canonical form with $f$ as its outermost symbol or $g$ as its outermost symbol.
Call this case to be that of {\em a shared constant possibly having multiple distinct normal forms in different $AC$ symbols}. This case must also be considered along with the propagation of constant rewrite rules.

One way to address the above case is by introducing a new constant $u$, making $c \grt u$; $\gg_f$ and $\gg_g$ are extended to include monomials in which $u$ appears with the constraint that $f(A) \gg_f u$
and $g(B) \gg_g u$.
Add a new rule $c \rew u \in R_C$, replace the
rule $c \rew f(A) \in R_f$ by $f(A) \rew u$ and $c \rew g(B) \in R_g$ by $g(B) \rew u$; this way, $c, f(A), g(B)$ all have a normal form $u$.
The replacements of $c \rew f(A)$ to $f(A) \rew u$ in $R_f$ and similarly of $c \rew g(B)$  to $g(B) \rew u$ in $R_g$ may violate the local confluence of $R_f$ and $R_g$ since reorientation of these equations may result in superpositions with the left sides of other rules.
New superpositions are generated in $R_f$ as well as $R_g$, possibly leading to new rules including new constant rules. After local confluence is restored, the result is new $R'_f$ and $R'_g$ with
$u$ being the canonical form of $c$.

To illustrate, consider $S = \{ c = a + b, c = a*b \}$ with AC $+, *$. For an ordering $c \grt b \grt a$ with both $\gg_+$ and $\gg_*$ being pure lexicographic, $R_+ = \{c \rew a + b\}, R_* = \{ c \rew a * b\}.$  These reduced canonical rewrite systems have a shared constant $c$ with two different normal forms implying that $R_+ \cup R_*$ is not canonical.
Introduce a new constant $u$ with $c \grt b \grt a \grt u$ (other orderings including $c \grt u \grt b \grt a$ or $c \grt u \grt a \grt b$ are also possible); make $R_S = \{ a + b \rew u, a * b \rew u, c \rew u\}$, which is reduced canonical in the extended signature; however, it is not even locally confluent when expressed in the original signature. A choice about whether the canonical form of $c$ is an $f$-monomial or a $g$-monomial is not made as a part of this algorithm since nonconstant $f$-monomials and $g$-monomials are noncomparable.

Other cases, for instance, the one in which $c \rew f(A)$ in $R_f$ but $g(B) \rew c$ in $R_g$, can be considered in a similar way by introducing a new constant $u$, to preserve modularity by not including mixed rules on $f$ monomials and $g$ monomials.

Only finitely many new constants need to be introduced in this construction; their number depends upon the number of canonical forms a constant can have with different outermost AC symbols. \ignore{Any new constant, say $u'$, in the next iteration is introduced if there is another constant $c' \rew f(A'), c' \rew g(B')$, where $A', B'$ may or may not have the new constant $u$ and further,  $u'$ is made smaller than $u$ when introducing $c' \rew u', f(A') \rew u', f(B') \rew u'$. Irrespective $c'$ is congruent to an $f$-term as well as $g$-term in the original signature obtained after replacing $u$, if needed, by $f(A), g(B),$ respectively. This means that only new constants are introduced for the constants in the original signature to express relationship between $f$ and $g$ terms in the original signature and these constants. And, there are only finitely many constants in the original signature.}

\vspace*{-5mm}
\section{Examples}
The proposed algorithms are illustrated using several examples which are mostly picked
from the above cited literature with the goal of not only to show how the proposed
algorithms work, but also contrast them with the algorithms reported in the literature.

{\bf Example 9:} Consider an example from \cite{BTVJAR03}: $S = \{ f(a, c) =
a, ~ f(c, g(f(b, c))) = b, ~ g(f(b, c)) =f(b, c) \}$ with a single AC symbol $f$ and one uninterpreted symbol $g$. 

Mixed term $f(c,g(f(b, c)))$ is purified by introducing
new symbols $u_1$ for $f(b, c)$ and $u_2$ for $g(u_1)$; thus
$\{ f(a, c) = a, ~  f(b, c) = u_1,~g(u_1) = u_2,~ f(c, u_2) = b, 
u_2 = u_1 \}$. (i)  $S_C = \{ u_2 = u_1 \}$, (ii) $S_U = \{ g(u_1) = u_2 \}$, and (iii) $S_f = \{ f(a, c) = a, f(b, c) = u_1, f(c, u_2) = b\}$. 

Different total orderings on constants are used to illustrate how
different canonical forms can be generated.
Consider a total ordering $f \gg g \gg a \gg b \gg u_2 \gg u_1$. $R_C = \{1.~  u_2 \rew u_1 \}$
normalizes the uninterpreted equation and it is oriented as: $R_U = \{
2.~ g(u_1) \rew u_1\}$. 

To generate a reduced canonical rewrite system for AC $f$-terms,
an admissible ordering on $f$-terms must be chosen. The degree-lexicographic ordering on monomials will be used for simplicity: $f(M_1) \gg_f f(M_2)$ iff 
(i) $|M_1| > |M_2|$ or (ii) $|M_1|= |M_2| \land M_1-M_2$ includes a constant bigger than every 
constant in $M_2-M_1$.
$f$-equations are normalized using rules 1 and 2, and are oriented as: $\{3.~  f(a, c)
\rew a,~ 4.~  f(c, u_1) \rew b, ~5. ~f(b, c) \rew u_1 \}$. 

Applying the {\bf SingleACCompletion} algorithm,
the superposition between rules $3, 5$ is $f(a, b, c)$ with the
critical pair: $(f(a, b), f(a, u_1))$, leading to a rewrite rule
$6.~ f(a, b) \rew f(a, u_1)$; the superposition between the rules $3, 4$ is
$f(a, c, u_1)$ with the critical pair: $(f(a, u_1), f(a, b))$
which is trivial by rule 6. The superposition between the rules $4, 5$ is $f(b, c,
u_1)$ with the critical pair: $(f(b, b), f(u_1, u_1))$ giving: $7. ~f(b, b) \rew f(u_1, u_1)$.
The rewrite system $R_f = \{3, 4, 5, 6, 7\}$ is a reduced canonical rewrite
system for $S_f$. $R_S = \{1, 2 \} \cup R_f$ is a reduced canonical rewrite system
associated with the AC congruence closure of the input and serves as its
decision procedure.

In the original signature, the above rewrite system $R_S$ is:
$\{ g(f(b, c))) \rew f(b, c), ~f(a, c) \rew a, ~ f(b, c, c) \rew b, ~ f(a, b) \rew f(a, b, c), 
~ f(b, b) \rew f(b, b, c, c)\}$ with 5 being trivial. The
reader would observe this rewrite system is locally confluent but not terminating.

\ignore{To try the above example with a different ordering,
reverse the ordering on constants $u_1, u_2$ to make $a \grt b \grt u_1
\grt u_2$. This leads to the canonical rewrite system on uninterpreted symbols: $\{ u_1
\rew u_2, g(u_2) \rew u_2 \}$. After running completion on the AC rewrite system 
$\{f(a, c) \rew a, f(c, u_2) \rew b, f(b, c) \rew u_2 \}$ gives
$\{f(a, c) \rew a, f(c, u_2) \rew b, f(b, c) \rew u_2, f(a, b)
\rew f(a, u_2), f(b, b)
\rew f(u_2, u_2) \}$. In the original signature, the
locally-confluent rewrite system is:
$\{ f(b, c) \rew g(f(b, c)),  f(a, c) \rew a,  f(c,
g(f(b, c))) \rew b,  f(a, b) \rew f(a, g(f(b, c))), f(b, b) \rew
f(g(f(b, c)), g(f(b, c)))\}$. }

Consider a different ordering:
$f \gg g \gg a \gg b \gg u_1 \gg u_2$.

This gives rise to a related rewrite system in which $u_1$ is rewritten to $u_2$:
$\{ u_1 \rew u_2, g(u_2) \rew u_2, f(a, c) \rew a, f(c, u_2) \rew b, f(b, c) \rew u_2, f(a, b) \rew f(a, u_2), f(b, b) \rew f(u_2, u_2)\}.$ In the original signature, the rewrite system is:
$R'_S = \{ f(b, c) \rew g(f(b, c)), g(g(f(b, c))) \rew g(f(b, c)), f(a, c) \rew a, f(c, g(f(b, c))) \rew b, f(a, b) \rew f(a, g(f(b, c))), f(b, b) \rew f(g(f(b, c)), g(f(b, c)))\}$. The reader should compare $R'_S$ which seems more complex using complicated terms, with the above system $R_S$ with different orientation between $u_1, u_2.$

\ignore{
In the original signature, the system is:

Rule 4'
reduces to $4''.~ f(b, c, c) \rew b$, and $5'.~ u_2 \rew f(b, c)$.
The superposition between the rules $3, 4''$ is $f(a, b, c, c)$ with the
critical pair:$(f(a, b, c), f(a, b))$, which is trivial by rule 3. The
AC rewrite system $\{3, 4'', 5'\}$ along with $\{1, 2\}$ constitutes  
another reduced canonical rewrite system for the AC congruence
closure of the input equations.

In the original signature, the
above rewrite system is: $\{g(f(b, c))) \rew f(b, c), f(a, c)
\rew a, f(b, c, c) \rew b\}$. 
As the reader would notice this reduced canonical rewrite
system is much simpler than the other, generating simpler canonical forms.
The above example illustrates that sometimes it is better to keep the
original signature as much as possible by judiciously choosing
orderings relating new constants with original terms.}

\ignore{Putting them altogether gives a canonical rewrite system:
$ \{1. ~f(a, c) \rew a, ~2.~ f(b, c, c) \rew b, 3. g(u_1) \rew u_1$ with $f(b, c)$ serving as the 
definition of $u_1$. It can be used to generate canonical
signature of the above set of ground equations expressing using
$AC$ symbols.

On the original signature, the resulting system is:
$ \{1. ~f(a, c) \rew a, ~2.~ f(b, c, c) \rew b, 3. g(f(b, c)))
\rew f(b, c)$ which is confluent and terminating.}

\ignore{On the original signature,   it is:  $\{1.~ f(a, c) \rew a, ~2~ f(c, f(b, c)) \rew
  b, ~ 2.2~g( g(f(b, c))) \rew g(f(b, c)), ~3.~f(b, c) \rew
  g(f(b, c)), 4. f(a, b) \rew f(a, a, b)\}$ which is nonterminating
  but confluent, a situation similar to the uninterpreted case
  \cite{KapurRTA97}.}

{\bf Example 10:}
This example illustrates interaction
between congruence closures over uninterpreted symbols and AC symbols.
Let $S = \{ g(b) = a, g(d) = c,  a * c = c, b * c = b, a * b =
d\}$ where  $g$ is uninterpreted and $*$ is AC. Let $* \gg g \gg a \gg b \gg c \gg d$. Applying the steps of the algorithm, $R_U$, the congruence closure over uninterpreted symbols, is  $\{g(b) \rew  a, g(d) \rew c\}$,  Completion
on the $*$-equations using
degree lexicographic ordering, oriented as $\{ a * c \rew  c, b * c \rew b, a * b \rew d \}$, generates an implied constant equality
$b = d$ from the critical pair of $a * c \rew c, b * c
\rew b $. Using the rewrite rule $b \rew d$, the AC rewrite system reduces to: $R_* = \{ a * c \rew c, c * d \rew 
d, a * d \rew d\}$, which is canonical.

The implied constant rule $b \rew d$ is added to $R_C: \{ b \rew d \}$ and is also propagated to $R_U$, which
makes the left sides of $g(b) \rew a$ and $g(d) \rew c$ equal, 
generating 
another implied constant rule $a \rew c$ which is added to $R_C$: $\{a \rew c, ~b \rew d\}$. $R_U$ becomes $\{ g(d) \rew c \}.$ The implied constant rule is propagated to the AC rewrite
system on $*$. $R_C$ normalizes 
$R_*$ to $\{ c * c \rew c,~ c * d \rew d\}.$

The output of the algorithm is a canonical rewrite system: $\{ g(d) \rew c, ~b \rew d, ~a \rew c,~ c * c \rew
c,~ c * d \rew d \}$. In general, the propagation of equalities can result in the left sides of the rules in AC subsystems change, generating new superpositions.

{\bf Example 11:} Consider an example with multiple AC symbols and a single uninterpreted symbol from \cite{NR91}: $S = \{a + b = a * b, a *
c = g(d), d = d' \}$ with AC symbols $+, *$ and an uninterpreted symbol $g$. This example also illustrate flexibility in choosing orderings in the proposed framework.

Purification leads to introduction of new constants : $\{ 1.~ a + b = u_0, ~2. ~a * b = u_1,~3.~ a * c = u_2, 
4. ~g(d) = u_2, 5.~ d = d', 6.~ u_0 = u_1\}$. Depending upon a total ordering on constants
$a, b, c, d, d'$, there are thus many possibilities depending upon the desired
canonical forms. 

Recall that $a + b$ cannot be compared with $a*b$, however $u_0$ and $u_1$ can be compared.
If canonical forms are desired so that the canonical form of $a + b$ is $a
* b$, the ordering should include $u_0 ~\grt~ u_1$. 
Consider a total ordering $a \gg b \gg c \gg d \gg d' \gg u_2 \gg u_0 \gg u_1$.
Degree-lexicographic ordering is used on $+$ and $*$ monomials.
This gives rise to:
$R_C = \{6. ~u_0 \rew u_1, 5.~ d \rew d'\}, ~R_U = \{ 4'.~ g(d') \rew u_2 \}$ along with $R_+ = \{1.~ a + b \rew u_1\}$ and $R_* = \{2.~ a * b \rew u_1, 3.~ a*c \rew u_2 \}$.
The canonical rewrite system for $*$ is: $\{2.~ a * b \rew u_1, 3.~ a*c \rew u_2, 7. ~b * u_2 \rew c * u_1\}.$

The rewrite system $R_S = \{1.~ a + b \rew u_1, ~2. ~a * b \rew u_1,  ~3.~ a * c \rew u_2, 
~4'.~ g(d') \rew u_2, ~ 5.~ d \rew d', ~6.~ u_0 \rew u_1, ~7.~ b *
u_2 \rew c * u_1 \}$ is canonical. Both $a + b$ and $a * b$ have the same
normal form $u_1$ standing for $a * b$ in the original signature.

With $u_1 \gg u_0$, another canonical rewrite system $R_S = \{1.~ a + b \rew u_0, ~2. ~a * b \rew u_0,
~ 6'.~ u_1 \rew u_0,  ~3.~ a * c \rew u_2, 
~4.~ g(d') \rew u_2, ~ 5.~ d \rew d', ~7.~ b *
u_2 \rew c * u_0 \}$ in which $a + b$ and $a * b$ have the 
normal form $u_0$ standing for $a + b$, different from the one above. 

\ignore{Similarly, if $u_3 \grt u_2$ leading
to the orientation of 4.1 from right to left, with $g(f), a * c$
have the normal form $u_2 = a * c$. 
Each of these rewrite systems  generate different
canonical forms for interpreted, uninterpreted as well as mixed terms.

With $g \gg *$, on the original signature, the system is:
$\{ e \rew f, g(f) \rew a * c, a * b \rew a + b \}$, which is
terminating as well as confluent, but with $* \grt g$, the system
is: $\{ e \rew f, a * c \rew g(f), a * b \rew a + b, b * g(f)
\rew a * b * c \}$, which is locally confluent but not
terminating (consider for instance $a * b * c$).}

\section{A \Groebner Basis Algorithm as an AC Congruence Closure}

Buchberger's \Groebner basis algorithm when extended to 
polynomial ideals over integers \cite{KapurGBERJSC88} can be interestingly viewed as a special congruence closure algorithm, extending the congruence closure algorithm over multiple AC symbols $+$ and $*$ which in addition, satisfy the properties of a commutative ring with unit, especially the distributivity property. $+$ satisfies the properties of an Abelian group with
$0$ as its identity;
$*$ is AC with identity 1 and additionally with the distributivity property: $\forall x, y, z, ~
x * (y + z) = (x * y) + (x * z)$; $x * 0 = 0$ as well. To incorporate the distributivity property in the combination, additional superpositions and associated critical pairs must be considered.
Relationship between \Groebner basis algorithm and the Knuth-Bendix completion procedure has been investigated in \cite{KRK,Winkler,KRKW,MarcheNormalized}, but the proposed insight is novel. 
Below, we illustrate this new insight using an example from \cite{KapurGBERJSC88}. 

\ignore{The ring structure of polynomials gives rise to additional interaction when a canonical rewrite system for the congruence closure of $+$-monomials is combined with a canonical rewrite system for the congruence closure of $*$-monomials.
Along with $+$ being an Abelian group with the identity $0$ and Abelian monoid for $*$ with the identity 1 (i.e., AC $*$ with the identity 1),
there is the universal distributivity axiom relating  $*$ and $+$ along with $x * 0 = 0:$ $x * (y + z) = (x * y) + (x * z)$.

The axioms of an Abelian $+$ along with an Abelian monoid $*$
with $x * 0 = 0$ and the distributivity rule: $x * (y + z) = (x * y) + (x * z)$ can be used to generate a canonical rewrite system for
a commutative ring with unit and
are used for normalizing terms to polynomials.}

Abusing the notation, $x + -(y)$ will be written as $x - y$.
An additive monomial $c ~t$, where $c \neq 0$, is an abbreviation of repeating the monomial $t$, $c$
times; if $c$ is positive, say $3$, then $3~ t$ is an abbreviation
for $t + t + t$; if $c$ is negative, say $-2$, then $-2~
t$ is an abbreviation for $-t -t$. 
A $*$-monomial with the coefficient 1 is a pure term expressed in $*$, but
 a monomial $3~ y * y * y$ is a mixed term $y * y * y + y * y * y + y * y * y$
with $+$ as the outermost symbol which has $*$ monomials as subterms. 

It can be easily proved that
a ground term in $+, *$ and constants can be brought into a polynomial form using the rules
of a commutative ring with unity and by collecting identical $*$ monomials; for example,
$ y * ( y + y) + y * y$ is equivalent to its canonical form $3~ y * y * y$.
Similarly, a ground equation built using these function symbols can be transformed to a polynomial equation.
\ignore{In fact, that is
how the algorithms proposed in \cite{KapurMACSYMA,KapurGBERJSC88}, were motivated from the research reported in \cite{KapRTA85}.} 

\ignore{
The universal distributivity equation relating $*$ to $+$ leads to new superpositions and critical pair
generation among the rewrite rules of $*$ and $+$ subsystems which
now interact even though their outermost symbols are disjoint.
To illustrate, consider ground AC rules $ a * b \rew b, a + a \rew 0$; they superpose with the
universal 
distributivity axiom $u * (v + w) = (u * v) + (u * w)$ to
generate mixed terms as superposition $(a * b) * (a + a)$ leading to the critical pair in mixed terms
$(a * b * 0, b * (a + a))$;
after simplification using the distributivity axiom, this leads to a new rule on $+$: $b + b \rew
0$. New rules are generated this way for all possible
pairs of a rule with $*$ with a rule with a $+$ term. Critical pairs
generated due to such interactions can have mixed terms in two different AC symbols.
Such mixed ground equations must be purified by extending the signature in general. Again, no extension rules in contrast to other approaches are used.}
Associated with a finite set $S$ of ground equations in $+, *$ and constants, is a basis $B$ of polynomial equations over the constants. It can be proved that the congruence closure $CR1CC(S)$ of $S$ using the properties of a commutative ring with 1 corresponds to the ideal $I(B)$ generated by the polynomials in
the basis $B$; the membership test in $CR1CC(S)$ is precisely
the membership test in the associated ideal $I(B)$.

Because of the distributivity property, a mixed term of the form $m * (t_1 + t_2)$ can possibly be rewritten by the rules in a rewrite system of $*$-monomials as well as by the rules in a rewrite system of $+$ monomials; further it can be expanded using the distributivity property as well to be
$m * t_1 + m * t_2$ which can be further rewritten. This is captured below using new superpositions and the associated critical pairs due to the distributivity property. This is first illustrated using an example.

{\bf Example 12:} Consider an example \cite{KapurGBERJSC88}.
The input basis is: $ B = \{ 7~ x*x* y = 3~ x, 4 ~x * y* y = x
* y, 3 ~ y*y*y = 0\}$ of a polynomial ideal over the integers \cite{KapurGBERJSC88}. Ground terms in $B$ can be written in many different ways, resulting in different sets of ground equations on terms built using $+, 0, -, *, 1$; the above basis $B$ is one such instance. 

None of these equations is pure in $+$ or $*$; in other words, they are mixed terms both in $*$ and $+$.
Purification of the above equations using new constants leads to: $\{ 1.~ 7 ~u_1 = 3~ x, 2.~ 4~ u_2 = u_3, 3.~ 3 ~u_4 = 0$ with $4.~ x * x * y
= u_1, 5.~ x *
y * y = u_2,  6.~ x * y = u_3, 7.~ y * y * y = u_4\}.$ These equations with new constant symbols are now pure with $1, 2, 3$ purely in $+$ and $4, 5, 6, 7$ purely in $*$.
Let a total ordering on all constants be: $u_1
\gg u_2 \gg u_4 \gg u_3 \gg y \gg  x$. Extend it using the degree-lexicographic ordering
on $+$-monomials as well as $*$-monomials, 
Orienting $*$ equations: $R_* = \{ 4.~ x * x * y
\rew u_1, 5. ~x * y * y \rew u_2, 6. ~ x * y \rew u_3, 7.~ y*y*y \rew
u_4 \}$. Orienting $+$ equations, $R_+ = \{1.~ 7 ~u_1 \rew 3 ~x, 2.~ 4~ u_2 \rew
u_3, 3.~ 3 ~u_4 \rew 0\}$. 

{\bf SingleACCompletion} on the ground equations on $*$ terms generates a
reduced canonical rewrite system: $ R_{*} = \{ 6.~ x * y \rew u_3, 7.~ y * y * y \rew u_4, 8. ~ u_3 * y \rew
u_2, 9.~ u_3 * x \rew u_1, 10.~ u_2 * x \rew u_3*u_3, 11.~ u_1 * y \rew u_3 *
u_3, 12.~ u_1*u_4 \rew u_2 * u_2,  13.~ u_4 * x * x \rew
u_2 * u_3, 14.~ u_3*u_3*u_3 \rew u_1 * u_2 \}$. This system captures
relationships among all product monomials appearing in the input.
This computation can be
utilized in other examples specified using the same monomial set
and thus does not have to be repeated. 
\ignore{In the original signature, the rewrite system is: $\{ 1. 7~ x*x* y 
\rew 3 ~x, 2. 4~ x * y * y \rew x * y, 3. 3~ y * y * y \rew 0\}$, a subset
of the original input with all other rules corresponding to trivial equalities.} 


{\bf SingleACCompletion} generating the congruence closure of the ground equations on $+$ terms produces $R_+$ as a reduced canonical rewrite system:
$\{1.~ 7 ~u_1 \rew 3 ~x, 2.~ 4~ u_2 \rew u_3, 3. ~3 ~u_4 \rew 0\}$.

There are no constant rewrite rules to be propagated from one rewrite system to another. 

Due to the distributivity property relating $*$ and $+$, however, there are interactions between a rule on $*$ monomials and a rule on $+$ monomials. To illustrate, a $+$ rule 3, $ 3 ~u_4 \rew 0$ and a $*$ rule 12 interact $ ~u_1 * u_4 \rew u_2 * u_2$
because of the common constant $u_4$; Using the distributivity property,
$u_1 * (3~ u_4) = u_1 * (u_4 + u_4 + u_4) = u_1 * u_4 + u_1 * u_4 + u_1 * u_4$. Rule 3 applies
on the pure $+$ monomial in the mixed term $u_1 * (3 * u_4)$ giving another mixed term $u_1 * 0$ which rewrites to 0 using a property of CR1; rule 12 repeatedly applies on the pure $*$ monomials in the mixed term $u_1 * u_4 + u_1 * u_4 + u_1 * u_4$, producing $u_2 * u_2 + u_2 * u_2 + u_2 * u_2$. Consequently, a new equality on mixed terms $u_2 * u_2 + u_2 * u_2 + u_2 * u_2 = 0$ is derived.
It is purified resulting in new relations on $*$-monomials as well as $+$-monomials. Using a new constant, say $v_1$, for $u_2 * u_2$,a new $+$ rule
$ 3 v_1 \rew 0$ is generated and added to $R_+$; $u_2 * u_2 \rew v_1$ is added to $R_*$. 

The pair $(u_1 * 0, u_2 * u_2 + u_2 * u_2 + u_2 * u_2)$ is a new kind of critical pair, generated from the superposition 
$u_1 * (3~ u_4) = u_1 * (u_4 + u_4 + u_4) = u_1 * u_4 + u_1 * u_4 + u_1 * u_4$ due to the distributivity property. 

In general, given a $+$ rule $c~ u \rew r \in R_+$ and a $*$ rule  $u * m \rew r' \in R_{*}$ with a common constant $u$, where $c \in \mathbb{Z}-\{0\}$, $r$ is a $+$-monomial, $m$ is $*$-monomial, the superposition is $(c~ u) * m $ generating the critical pair $(r * m, c ~r')$ on mixed terms.
$r * m $ and $c~ r'$ are normalized using distributivity and other rules. The resulting mixed terms are purified by introducing new constants, leading to new equalities on $*$ monomials as well as $+$ monomials. They are, respectively, added to $R_*$ and $R_+$
after normalizing and orienting them using the monomial ordering.

It can be proved that only 
this superposition needs to be considered to check local confluence of $R_+ \cup R_*$ \cite{KapurWeirdGB21}.

For more examples of superpositions and critical pairs between rules in $R_+$ and $R_*$, consider
rules $3$ and $13$, which give a trivial critical pair.

Considering all such superpositions, the method terminates, leading to the
canonical rewrite system  $\{ 3~ x \rew
0, u_1 \rew 0, u_2 \rew u_3, 3 ~u_4 \rew 0 \}$ among other rules. When converted into the original signature, this is
precisely the \Groebner basis reported as generated using Kandri-Rody and Kapur's algorithm: $\{ 3~ x \rew 0, x*x * y
\rew 0, x * y * y \rew x * y, 3 ~y * y * y \rew 0 \}$ as reported
in \cite{KapurGBERJSC88}.

A very simple example taken from \cite{KapurGBERJSC88} consisting of 2 ground equations
$ 2~ x*x * y = y, ~ 3 ~ x * y * y = x$ is also discussed in \cite{MarcheNormalized} illustrating 
how normalized rewriting and completion can be used to simulate \Groebner basis computations on polynomial equations over the integers. Due to lack of space, a detailed comparison could not be included; an interested reader is invited to contrast the
proposed approach with the one there. 
Comparing with \cite{MarcheNormalized}, it is reported there that a completion procedure 
using normalized rewriting generated 90 critical pairs in contrast to AC completion procedure \cite{PS81} which computed 1990 critical pairs. In the proposed approach, much fewer critical pairs are generated in contrast,
without needing to use any AC unification algorithm or extension rules. 

The above illustrates the power and elegance of the proposed combination framework. The proposed approach can also be used to compute \Groebner basis of polynomial ideals with coefficients over finite fields such
as $\mathbb{Z}_p$ for a prime number $p$, as well as domain with zero divisor such as $\mathbb{Z}_4$ \cite{KapurZero}. For 
example, in case of $\mathbb{Z}_5$, another rule is added: $1 + 1 + 1 + 1 + 1 \rew 0$ added to the input basis.

\section{Conclusion}

A modular algorithm for computing the congruence closure of ground equations expressed using multiple AC function symbols and uninterpreted symbols is presented; an AC symbol could also have additional semantic properties including idempotency, nilpotency, identity, cancelativity as well as being an Abelian group.

The algorithms are derived by generalizing the
 framework first presented in
\cite{KapurRTA97} for generating the congruence closure of ground equations over uninterpreted symbols. The key insight from \cite{KapurRTA97}--flattening of subterms by introducing new constants and
congruence closure on constants, is generalized by flattening mixed AC terms and purifying them by introducing new constants to stand for pure AC terms in every (single) AC symbol. The result of this transformation on a set of equations on mixed ground terms is a set of constant equations, a set of flat equations relating a nonconstant term in an uninterpreted symbol to a constant, and a set of equations on pure AC terms in a single AC symbol possibly with additional semantic properties. Such decomposition and factoring enable using congruence closure algorithms for each of the pure subproblems on a single AC symbol and/or uninterpreted symbols independently, which propagate equalities on shared constants. Once the propagation of constant
equalities stabilizes (reaches a fixed point), the result is (i) a unique reduced canonical rewrite system for each subproblem and  finally, (ii) a unique reduced canonical rewrite system for the congruence closure of a finite set of ground equations over multiple AC symbols and uninterpreted symbols. The  algorithms extend easily when AC symbols have additional properties such as idempotency, identity, nilpotency, cancelativity as well as their combination, as well as an AC symbol satisfying the properties of a group.

Unlike previous proposals based on the Knuth-Bendix completion procedure and its generalizations to AC theories and other equational theories, the proposed algorithms do not need to use extensions rules, AC/E unification or AC/E compatible termination ordering on arbitrary terms. Instead, the proposed framework provides immense flexibility in using different termination orderings for terms in different AC symbols and uninterpreted symbols.

The modularity of the algorithms leads to easier and simpler  correctness and termination proofs  in contrast to those in \cite{BTVJAR03,MarcheNormalized}. The complexity of the procedure is governed by the complexity of generating a canonical rewrite system for AC ground equations on constants.

The proposed algorithm is a direct generalization of Kapur's algorithm for the uninterpreted case, which has been shown to be efficiently integrated into SMT solvers including BarcelogicTools \cite{NOEncode}. We believe that the AC congruence closure can also be effectively integrated into SMT solvers. Unlike other proposals, the proposed algorithm neither uses specialized AC compatible orderings on nonground terms nor extension rules often needed in AC/E completion algorithms and AC/E-unification, thus avoiding explosion of possible critical pairs for consideration. 

By-products of this research are new insights into combination algorithms for ground canonical rewrite systems-- (i) making orderings on ground terms more flexible and general in which constants are allowed to be greater than nonconstant terms as in orderings for \Groebner basis algorithms,
as well as (ii) a new kind of superpositions to incorporate semantic properties such as cancelativity, which can be viewed as a superposition between the left side of a rule with the right side of another rule.
An instantiation of the proposed combination framework to an AC symbol with identity and another AC symbol enriched with the properties of an Abelian group, leads to a new way to view a \Groebner basis algorithm for polynomial ideals over integers, as a combination congruence closure algorithm. Canonical bases of polynomial ideals over the integers (and more general rings with zero divisors) on extended signature can be computed for which rewriting (simplification) of a polynomial by another polynomial, when viewed on the original signature, need not terminate, opening possible new research directions.

Integration of additional semantic properties of AC symbols is done on a case by case basis by modifying critical pair constructions but avoiding extension rules or complex AC termination orderings compatible with these semantic properties. 
Using the proposed framework, other types of combinations of congruence closure algorithms in the presence of other semantic properties relating different AC symbols may be possible as well.
A general approach based on the properties of the semantic properties that avoids extension rules and sophisticated machinery due to the AC properties, needs further investigation.

\vspace*{1mm}
\noindent
{\bf Acknowledgments:} Thanks to the FSCD conference and LMCS referees for their reports which substantially improved the presentation. Thanks to Christophe Ringeissen for detailed discussions which clarified the contributions of the proposed approach in contrast to other approaches, particularly \cite{BaaderT97,ArmandoBRS09}. Thanks also to Jose Abel Castellanos Joo for help in formulating some of the algorithms. 


\bibliographystyle{alphaurl}
\bibliography{FSCD}

\end{document}